\documentclass[11pt]{article}
\usepackage[margin=1in]{geometry}
\usepackage{bbm}
\usepackage{graphicx,hyperref,color}
\usepackage{amsmath,amssymb,amsthm}
\usepackage{enumerate}
\usepackage{enumitem}
\usepackage{fullpage}
\usepackage{algorithm} 
\usepackage{mathrsfs}
\usepackage[noend]{algcompatible}
\usepackage{authblk}
\usepackage{varwidth}
\usepackage{mathrsfs}

\usepackage{tcolorbox}

\newtheorem{theorem}{Theorem}
\newtheorem{corollary}{Corollary}
\newtheorem{lemma}{Lemma}
\newtheorem{claim}{Claim}

\newtheorem{definition}{Definition}
\newtheorem{observation}{Observation}

\newcommand{\opt}{\mathrm{OPT}}
\newcommand{\dm}{\mathrm{diam}}
\newcommand{\edm}{\mathrm{ediam}}

\newcommand{\mb}{\mathcal{B}}
\newcommand{\mc}{\mathcal{C}}
\newcommand{\md}{\mathcal{D}}

\newcommand{\mf}{\mathcal{F}}
\newcommand{\mg}{\mathcal{G}}
\newcommand{\mv}{\mathcal{V}}
\newcommand{\me}{\mathcal{E}}
\newcommand{\ms}{\mathcal{S}}
\renewcommand{\mp}{\mathcal{P}}
\newcommand{\mq}{\mathcal{Q}}
\newcommand{\mt}{\mathcal{T}}
\newcommand{\mx}{\mathcal{X}}
\newcommand{\my}{\mathcal{Y}}

\newcommand{\wsp}{\mathtt{Ws}}
\newcommand{\ssp}{\mathtt{Ss}}
\newcommand{\lgt}{\mathtt{L}}

\newcommand{\mbv}{\mathbf{v}}
\newcommand{\mbe}{\mathbf{e}}
\newcommand{\mbx}{\mathbf{x}}
\newcommand{\mby}{\mathbf{y}}
\newcommand{\mbz}{\mathbf{z}}
\newcommand{\mbc}{\mathbf{c}}
\newcommand{\mbtw}{\mathrm{tw}}

\newcommand{\pie}{path edge}

\newcommand{\ce}{c(\epsilon)}
\newcommand{\ee}{$\tilde{\hat{\mbox{e}}}$}
\newcommand{\red}{\mathrm{red}}
\DeclareMathOperator{\st}{\textsc{ST}}
\DeclareMathOperator{\MST}{\mathrm{MST}}
\DeclareMathOperator{\sr}{\mathrm{SR}}
\DeclareMathOperator{\poly}{\mathrm{poly}}
\DeclareMathOperator{\polye}{\mathrm{poly}(\frac{1}{\epsilon})}

\DeclareMathOperator{\Opst}{\mathrm{\textsc{W}}}
\DeclareMathOperator{\enc}{\mathrm{Enc}}
\DeclareMathOperator{\dec}{\mathrm{Dec}}
\DeclareMathOperator{\mgw}{\boldsymbol{\omega}}
\DeclareMathOperator{\smin}{\subseteq_{\min}}
\DeclareMathOperator{\defi}{\overset{\text{def}}{=} }

\date{}

\title{A PTAS for subset TSP in minor-free graphs\thanks{A major part of this work was done while the author was a graduate student at Oregon State University.}}
\author{Hung Le}
\affil{Department of Computer Science \\ University of Victoria, Canada\\
  \texttt{hungle@uvic.ca}}
\begin{document}
\maketitle
\begin{abstract}

We give the first PTAS for the subset Traveling Salesperson Problem (TSP) in $H$-minor-free graphs.  This resolves a long standing open problem in a long line of work on designing PTASes for TSP in minor-closed families initiated by Grigni, Koutsoupias and  Papadimitriou in FOCS'95. The main technical ingredient in our PTAS is a  construction of a nearly light subset $(1+\epsilon)$-spanner for any given edge-weighted $H$-minor-free graph. This construction is based on a necessary and sufficient condition given by \emph{sparse spanner oracles}: light subset spanners exist if and only if sparse spanner oracles exist. This relationship allows us to obtain two new results:
\begin{itemize}
\item An $(1+\epsilon)$-spanner with lightness $O(\epsilon^{-d+2})$ for any doubling metric of constant dimension $d$. This improves the earlier lightness bound $\epsilon^{-O(d)}$ obtained by Borradaile, Le and Wulff-Nilsen~\cite{BLW19}.
\item An $(1+\epsilon)$-spanner with sublinear lightness for any metric of constant correlation dimension. Previously, no spanner with non-trivial lightness was known.
\end{itemize}

\end{abstract}

\section{Introduction}

Given an edge-weighted graph $G$ and a set of terminals $T$ in $G$, the \emph{subset TSP} problem asks for a shortest tour that visits every terminal in $T$ at least once. This problem generalizes the well-known TSP problem in which $T$ contains every vertex of the graph. In practice, subset TSP is typically more interesting than TSP:  it is often that the set of vertices we want to visit in a graph is much smaller than the whole vertex set.  Indeed, subset TSP has been studied extensively in operational research since 1985~\cite{CFN85,Salazar03,BMT13,ZTXL15,LN15,ZTXL16} under a different name -- \emph{Steiner TSP} problem. Arora, Grigni, Karger, Klein and Woloszyn~\cite{AGKKW98} observed that subset TSP in planar graphs  generalizes the well-studied TSP in Euclidean plane.

In general graphs, one can reduce subset TSP to TSP by taking metric completion on the terminals. However, if the input graph has a special structure, such as excluding a fixed minor, taking metric completion would destroy the structure that may otherwise be used to algorithmic advantage. Over the past 20 years, much research have been spent on exploiting minor-closed properties to design polynomial time approximation schemes\footnote{A polynomial-time approximation scheme is an algorithm which, for a given fixed error parameter $\epsilon$, finds a solution whose value is within $1\pm\epsilon$ of the optimal solution in polynomial time.} (PTAS) for TSP and subset TSP~\cite{GKP95,AGKKW98,Klein05,Klein06,BDT14,DHK11,BLW17,GS02}. In FOCS'17,  Borradaile, Le and Wulff-Nilsen~\cite{BLW17} fully resolved the approximation complexity of TSP in $H$-minor-free graphs for any fixed graph $H$ by designing an efficient PTAS (EPTAS). However, designing a PTAS, even an inefficient one, for subset TSP in the same setting remains a widely open problem that has been raised several times~\cite{DHK11,BDT14,BLW17}. In this paper, we provide a positive answer to this problem.

\begin{theorem}\label{thm:stsp-ptas}
For any given fixed $\epsilon > 0$, there is a polynomial time algorithm that, given an edge-weighted $H$-minor-free graph $G$ and a set of terminals $T$ in $G$, can find a tour visiting $T$ whose length is at most $(1+\epsilon)$ times the length of the optimal tour. 
\end{theorem} 

\noindent The precise running time of the algorithm in Theorem~\ref{thm:stsp-ptas} is $n^{O_H(\poly(\frac{1}{\epsilon}))}$ where $O_H$ notation hides the dependency of the constant on the size of  $H$. Though our PTAS is not efficient, it is a crucial stepping stone toward an efficient one.

\subsection{TSP and subset TSP in minor-closed families}

In FOCS'95, Grigni, Koutsoupias and  Papadimitriou~\cite{GKP95}  showed that TSP admits a PTAS in \emph{unweighted} planar graphs.  Their result has triggered a long line of research on designing PTASes for TSP and subset TSP in minor-closed graph families.
In SODA'98, Arora, Grigni, Karger, Klein and Woloszyn~\cite{AGKKW98} designed a PTAS for TSP in \emph{edge-weighted} planar graphs. Their algorithm differs the algorithm of Grigni, Koutsoupias and  Papadimitriou largely in the preprocessing step: it runs on light spanners for planar graphs by Alth\"{o}fer, Das, Dobkin, Joseph and Soares~\cite{ADDJS93}. Both PTASes are based on the cycle separator theorem by Miller~\cite{Miller86} and have running time $O(n^{\polye})$. In the same paper, Arora, Grigni, Karger, Klein and Woloszyn~\cite{AGKKW98} gave a quasi-polynomial time approximation scheme (QPTAS) for subset TSP in edge-weighted planar graphs and conjectured that a PTAS is possible.  In FOCS'05, Klein~\cite{Klein05} introduced a \emph{contraction decomposition framework}  that allows one to reduce designing a PTAS for (subset) TSP in planar graphs to finding a (subset) \emph{spanner} whose weight is at most $c(\epsilon)w(\opt)$ for some constant $c(\epsilon)$ depending on $\epsilon$. Here, $w(.)$ is the weight function on the edges of the graph and $\opt$ is an optimal solution. The framework in combination with the light spanner by Alth\"{o}fer, Das, Dobkin, Joseph and Soares~\cite{ADDJS93} implies an EPTAS for TSP in planar graphs. Soon after, in FOCS'06, Klein~\cite{Klein06} constructed the first light subset spanner, thereby obtaining an EPTAS for subset TSP in planar graphs. This answered the open question asked earlier by Arora, Grigni, Karger, Klein and Woloszyn~\cite{AGKKW98}. Klein's results  in planar graphs were then generalized to bounded-genus graphs by Borradaile, Demaine and Tazari~\cite{BDT14}. They  asked whether one can design a PTAS for TSP and subset TSP in $H$-minor-free graphs for any fixed graph $H$. 

In STOC`11, Demaine, Hajiaghayi and Kawarabayashi~\cite{DHK11} generalized the contraction decomposition framework of Klein to $H$-minor-free graphs. With Grigni and Sissokho's spanners~\cite{GS02}, they obtained a PTAS for TSP in $H$-minor-free graphs, improving upon the early work by  Grigni and Sissokho~\cite{GS02} who designed a QPTAS for the same problem.  Since Grigni and Sissokho's spanners have weight $O(\log n\polye) w(\opt)$, the PTAS  by Demaine, Hajiaghayi and Kawarabayashi is not efficient. An EPTAS was then obtained by Borradaile, Le and Wulff-Nilsen~\cite{BLW17} in FOCS'17  via a new spanner of weight at most $O(\polye)w(\opt)$. They left designing a PTAS for subset TSP in $H$-minor-free graphs as the central open problem of the field~\cite{BLW17,DHK11}.  Note that even a QPTAS was not known for this problem.

In this paper, we design the first PTAS for the subset TSP problem  in $H$-minor-free graphs (Theorem~\ref{thm:stsp-ptas}). Our main contribution  a \emph{nearly light subset spanner} construction based on \emph{sparse spanner oracles}, a new concept we introduce in this work.  We show that spanner oracles with weak sparsity are both necessary and sufficient to construct light subset spanners, even for general graphs. This is somewhat surprising given that previous constructions need to make use of special properties, such as bounded dimension or minor-freeness, to construct light spanners from sparse spanners (see Section~\ref{sec:our-technique}). We hypothesize that the concept of sparse spanner oracles will find many other applications, and we support this hypothesis by giving applications in two different settings (Theorem~\ref{thm:dd-metric} and Theorem~\ref{thm:light-spanner-cd}).

In a broader context, our PTAS for subset TSP in $H$-minor-free is significant for several reasons. First, despite the fact that many beautiful meta algorithmic ideas~\cite{DHK11,BFLPST16,FLRS11} for designing PTASes have been developed for $H$-minor-free graphs over the year, a PTAS for subset TSP has still been out of reach. Second, it seems that a PTAS for subset TSP is not possible beyond $H$-minor-free graphs. The problem was proved to be MAXSNP-hard~\cite{GKP95} in topologically minor-closed graphs, which contains $H$-minor-free graphs. It is even MAXSNP-hard in $1$-planar graphs~\cite{Borradaile13}, which generalize planar graphs by allowing at most one crossing per edge in the embedding. (A PTAS for TSP in $1$-planar graphs remains unknown however.)  Furthermore, subset TSP is nontrivial even in unweighted graphs, while TSP in unweighted graphs is often easier:  the first PTAS for TSP is in unweighted graphs~\cite{GKP95}. Indeed, by rounding and subdividing long edges, a PTAS for unweighted subset TSP can be turned into a PTAS for weighted subset TSP.  (The subdivision step does not introduce bigger clique minors.) Finally, techniques developed for solving problems in $H$-minor-free graphs are often very different from techniques for the same problems in planar and bounded genus graphs -- representative example problems are padded decompositions with  strong diameter~\cite{AGGNT19}, eigenvalue bounds~\cite{BLR10} and light spanners~\cite{BLW17}, and that the techniques for $H$-minor-free graphs often find applications in different contexts. This holds for our technique as well (see Section~\ref{sec:other-apps}).

\paragraph{Other related work}  subset TSP has also been studied from the parameterized complexity point of view. The classical dynamic programming algorithm of Held and Karp~\cite{HK71} can solve the problem in $O(2^k)n^{O(1)}$ time where $k$ is the number of terminals. Klein and Marx~\cite{KM14} designed the first sub-exponential $(2^{O(\sqrt{k}\log k)} + W)n^{O(1)}$-time algorithm for subset TSP in planar graphs with maximum  edge-weight $W$. Marx, Pilipczuk and Pilipczuk~\cite{MPP18} generalized the algorithm of Klein and Marx to directed planar graphs and improved the running time to $2^{O(\sqrt{k}\log k)}n^{O(1)}$.

\subsection{Techniques}

Most PTASes for TSP and subset TSP, including ours, are based on the contraction decomposition framework. It was initially developed for planar graphs by Klein~\cite{Klein05}, and then extended to bounded genus graphs~\cite{DHM07} and $H$-minor-free graphs~\cite{DHK11}. PTASes following this framework have four steps: (1) construct  a light (subset) spanner $S$ that preserves the distance between every pair of (terminated) vertices up to $(1+\epsilon)$ factor, and has weight at most $L(\epsilon) \opt$ where $L(\epsilon)$ is a constant depending on $\epsilon$ only, (2) use the shifting technique to contract a subset of edges of $S$ to obtain a bounded treewidth graph, (3) apply dynamic programming to find an optimal solution of the contracted graph and (4) lift the solution found in step (3) to a solution of the input graph. The treewidth of the contracted graph in step (3) is $\polye L(\epsilon)$ where $L(\epsilon)$ is the constant in step (1), called the \emph{lightness} of the spanner\footnote{A more formal definition will be provided in Section~\ref{sec:prim}.}.

Sometimes, to have a PTAS, it suffices to have the lightness constant in step (1) relaxed to $L(\epsilon) \log n$ provided that there is a dynamic programming algorithm of running time $2^{O(\mathrm{tw})}$ in step (3) for the problem, where $\mathrm{tw}$ is the treewidth of the contracted graph. This is because $L(\epsilon) \log n$ lightness implies that the treewidth of the graph in step (3) is  $\log n\polye L(\epsilon)$,  and hence the final running time is $2^{O(\log n\polye L(\epsilon))} = n^{L(\epsilon)\polye}$. This relaxation was exploited by  Demaine, Hajiaghayi and Kawarabayashi~\cite{DHK11} in their PTAS for TSP.  Spanners with an extra logarithmic factor in the lightness are called \emph{nearly light spanners}. 

Constructing (nearly) light spanners  has become the most difficult task in designing PTASes following the contraction decomposition framework; prior work on approximating subset TSP in planar and bounded-genus graphs~\cite{Klein06, BDT14}  focused solely on this task. In this paper, we solve the same task for $H$-minor-free graphs.

\begin{theorem}\label{thm:main} Let $T$ be a of $k$ terminals in  an edge-weighted $H$-minor-free graph $G$ and $\st$ be a minimum Steiner tree of $G$ for $T$.  There is a polynomial time algorithm that can find a subgraph $S$ of $G$ such that:
\begin{enumerate}[noitemsep]
\item[(i)] $d_G(x,y) \leq d_S(x,y) \leq (1+\epsilon)d_G(x,y)$ for every two distinct terminals $x,y \in T$.
\item[(ii)] $w(S) = O_H(\poly(\frac{1}{\epsilon})\log k)w(\st)$.
\end{enumerate}
where $O_H(.)$ hides the dependency of the constant on $|H|$. Furthermore, if $G$ has constant treewidth, then $w(S) = O(\poly(\frac{1}{\epsilon}))w(\st)$.
\end{theorem}

Since $w(\st) \leq w(\opt)$, property (ii) implies that $w(S) \leq  O_H(\poly(\frac{1}{\epsilon})\log k)w(\opt)$. Thus, $S$ is nearly light when $k$ is polynomial in $n$. To obtain a PTAS, we need a dynamic programming algorithm that can solve subset TSP optimally in treewidth-$\mbtw$ graphs in $2^{O(\mbtw)}$ time. Such a dynamic programming algorithm can be obtained using standard techniques; readers are referred to Appendix~\ref{app:dp} for full details. In the following section, we will review known constructions of light subset spanners. 

\subsubsection{Previous techniques}\label{sec:tech-prev}

To the best of our knowledge, there are two light subset spanner constructions: one for planar graphs by Klein~\cite{Klein06} and another for bounded-genus graphs by Borradaile, Demaine and Tazari~\cite{BDT14}. Both constructions heavily rely on non-crossing embeddings of input graphs.

Klein's construction has two main components: \emph{strip decomposition} and \emph{bipartite spanner}. A strip decomposition of $G$ is constructed as follows. First, find a 2-approximation $\st$ of the optimal Steiner tree for $T$ using, say, Mehlhorn's algorithm~\cite{Mehlhorn88}.  Then double the edges of $\st$ to make a new face $f_T$ consisting of two copies of each edge in $\st$(see Figure~\ref{fig:cutting} in Appendix~\ref{app:figs}),    and designate $f_T$  as the infinite face of $G$.  Let $G'$ be the new planar graph.  A strip decomposition $K$ initially contains $f_T$ only. We then add vertices and edges of $G'$ to $K$ recursively by (1) walking  along the boundary of $G'$ to find a minimal subpath, say $\partial G'[x,y]$, whose endpoints' distance in $G'$ is less than the length of the subpath by a factor of $(1+\epsilon)$, (2) adding a shortest path, say $P$, between $x$ and $y$ to $K$  and (3) recursively applying the first two steps to the subgraph of $G'$ enclosed by $P \cup (\partial G' \setminus \partial G'[x,y])$ (see Figure~\ref{fig:strip-decomp} in Appendix~\ref{app:figs}).

Using a charging argument, Klein showed that $w(K) \leq O(\epsilon^{-1})w(\st)$. By construction, each face of the strip decomposition $K$ is composed of two paths, say $P$ and $Q$, between the same endpoints, in which one of them, say $P$, is a shortest path in $G$. A \emph{bipartite spanner} is then constructed for each strip. The final subset spanner is the union of all  the bipartite spanners. Each bipartite spanner, say $\hat{K}$, has two properties: (a) for every two vertices $x \in P, y \in Q$, $d_{\hat{K}}(x,y) \leq (1+\epsilon)d_G(x,y)$ and (b) $w(\hat{K}) \leq \polye w(P\cup Q)$.  Since $w(K) \leq O(\epsilon^{-1})w(\st)$, property (b)  guarantees that $w(S) \leq \polye w(\st)$, thereby implying the lightness property of $S$. To show the distance preserving property for every two distinct terminals $u$ and $v$, Klein used planarity to  argue that the shortest path in $G$ between $u,v$  either lies completely inside a strip or crosses  strip boundaries. Note that $u,v \in \partial G'$ by the preprocessing step in which we construct a new infinite face consisting of copies of the edges of $\st$. By the minimality of  $\partial G'[x,y]$ in step (1) and the property (a) of bipartite spanners, one can show that there is an $(1+\epsilon)$ approximate  shortest path between $u,v$ in $S$.

For bounded-genus graphs, Borradaile, Demaine and Tazari~\cite{BDT14} used the cutting technique to cut the input graph into a planar graph, and then applied Klein's construction. In a certain sense, their subset spanner construction still heavily relies on planarity. Borradaile, Demaine and Tazari~\cite{BDT14} conjectured that a similar construction can be applied to $H$-minor-free graphs using Robertson and Seymour's decomposition~\cite{RS03}. However, this direction has not been fruitful. Even in a very restricted setting where $G$ has bounded treewidth, it is not known whether light subset spanners exist. In this work, we follow an entirely different approach to bypass the embedding in constructing subset spanners.

\subsubsection{Our techniques}\label{sec:our-technique}

 In the lightness analysis of greedy spanners for Euclidean and doubling metrics~\cite{BLW19,LS19}, the authors implicitly used sparsity of the greedy spanner to bound the weight. We observe that their analysis can be turned into a non-greedy algorithm that explicitly uses sparse spanners to construct light spanners. In this way, their analysis can be seen as an implicit construction.  This inspires us to follow the same strategy: build light spanners from sparse spanners.  However, there are two fundamental issues with this idea. The major issue is  in defining  a ``sparsity'' measure for subset spanners.  Simply counting the number of edges does not work: one can subdivide an edge of the spanner infinitely many times without changing the number of terminals. To get around this problem, we introduce \emph{weak sparsity} that can be loosely regarded as ``counting'' the number of shortest paths needed to preserve distances between terminals. Thus, this measure is robust to edge subdivision.

 Another fundamental issue is the subtlety in the way sparse spanners were used in prior light spanner constructions in Euclidean and doubling metrics~\cite{BLW19,LS19}. These constructions often require another property of the input such as the packing property, that we don not have in our setting. The light spanner construction by Borradaile, Le and Wulff-Nilsen~\cite{BLW17} for $H$-minor-free graphs uses the fact that a contracted graph of an $H$-minor-free graph has a linear number of edges. However, if one applies this construction to terminals, the contracted graph may have  a super-linear number of edges w.r.t the number of terminals. We get around this issue by introducing an abstraction called  \emph{spanner oracles} that hide all subtleties in previous light spanner constructions. As a result, we can show that constructing a weakly sparse spanner oracle is sufficient to have a  light subset spanner (in general graphs) and surprisingly that having a weakly sparse spanner oracle is also necessary. 

Our sparse spanner oracles are inspired by previous (implicit) light spanner constructions in Euclidean and doubling metrics~\cite{BLW19,LS19}. Specifically, these constructions repeatedly select a specific \emph{subset of points } $Q$ in the input metric and construct a sparse spanner $S$ for $Q$ with two properties: (a) it preserves, up to $(1+\epsilon)$ factor, the distance between every pair of points of $Q$  in range $[\ell/2, \ell]$ for some positive real number $\ell$ and (b)  $w(S) = O(|Q|\ell)$.  Property (b) follows exactly from the sparsity of $S$: if  it has sparsity at most $s(\epsilon)$, by deleting every edge of weight more than $(1+\epsilon)\ell$, $w(S)  \leq s(\epsilon)(1+\epsilon)|Q|\ell = O(|Q|\ell)$. Property (b) is then used to bound the lightness of the final spanner. Note that many such sparse spanners are constructed  (implicitly) for different subsets of points determined by the algorithm.  We instead look at this as an oracle: the algorithm repeatedly queries the oracle by giving it a subset of points $Q$ and the oracle must return a spanner $S$ with weight at most $O(|Q|\ell)$. We then regard $\frac{w(S)}{|Q|\ell}$ as a sparsity measure of $S$.  The following definition formalizes this intuition. 

\begin{definition}[Spanner oracle]\label{def:spanner-oracle} A  spanner oracle, denoted by $\mathcal{O}$, of a given graph $G$ is an algorithm that, given any set of terminals  $T$ and a real positive number $\ell$, outputs a minimal subgraph $\mathcal{O}(T,\ell)$ of $G$ spanning $T$ such that:
\begin{equation}
d_G(x,y) \leq d_{\mathcal{O}(T,\ell)}(x,y) \leq (1+\epsilon) d_G(x,y) \quad \forall x\not=y \in T \mbox{ s.t } \ell/8\leq  d_G(x,y)\leq \ell
\end{equation}
The \emph{weak sparsity} of $\mathcal{O}$ is defined as:
\begin{equation}
\wsp_{\mathcal{O}} = \sup_{\substack{\ell \in \mathbb{R}^+ \\\emptyset \subset T\subseteq V }} \frac{w(\mathcal{O}(T,\ell))}{|T|\cdot\ell}
\end{equation}
The \emph{strong sparsity} of $\mathcal{O}$ is defined as:
\begin{equation}
\ssp_{\mathcal{O}} = \sup_{\substack{\ell \in \mathbb{R}^+ \\\emptyset \subset T\subseteq V }} \frac{|E(\mathcal{O}(T,\ell))|}{|T|}
\end{equation}
\end{definition}

There is nothing special about constant $\frac{1}{8}$ in the definition of spanner oracles; any sufficiently small constant works. Observe that that strong sparsity implies weak sparsity:
\begin{equation}\label{eq:wsp-ssp}
\wsp_\mathcal{O}  \leq 2\cdot \ssp_{\mathcal{O}}
\end{equation}
since by minimality, every edge of $\mathcal{O}(T,\ell)$ has length at most $2\ell$; any edge of length more than $2\ell$ cannot be in a shortest path between two terminals of distance at most $(1+\epsilon)\ell$ in $\mathcal{O}(T,\ell)$ when $\epsilon < 1$.

Since Euclidean and doubling metrics are closed under taking metric completion on any subset of points, we can show that any graph representing a point set in Euclidean or doubling metrics of dimension $d$ has a spanner oracle with strong sparsity $\epsilon^{-O(d)}$ (see Section~\ref{sec:spare-oracle-Euc} for details). However, $H$-minor-free graphs do not have such a closure property, we need to use weak sparsity. Curiously, despite the fact that metrics of constant correlation dimension are also not closed under taking metric completion, we still can show that they have strongly sparse spanner oracles.

Weak sparsity allows us to draw connections to other areas and use their tools in our construction. Specifically, it has a close relationship to \emph{distance preserving minors}~\cite{KNZ14,CGH16} that arise in the vertex sparsification problem. We show that an (approximate) distance preserving minor with a linear number of Steiner vertices for any given set of terminals implies a spanner oracle of constant weak sparsity. Since graphs of bounded treewidth have  terminal preserving minors with such linearity~\cite{KNZ14}, they admit a sparse spanner oracle with constant weak sparsity.  We then use the weakly sparse spanner oracle to build a subset spanner with constant lightness for bounded treewidth graphs; a special case that has been open prior to our work. One might ask whether distance preserving minors can be used to get weakly sparse spanner oracles with constant sparsity for $H$-minor-free graphs. The answer is still unknown. The best-known approximate distance preserving minor for $H$-minor-free graphs has a quadratic number of Steiner vertices, which is not enough to imply a spanner oracle with non-trivial sparsity. Instead, we base our construction on shortest path separators~\cite{AG06} and our new light single-source spanners. The same idea has been used before in different problems~\cite{KKS11,CGH16}. 

\begin{theorem}\label{thm:sparse} Given an $H$-minor-free graph $G$ of $n$ vertices, in polynomial time, we can construct a spanner oracle $\mathcal{O}$ for $G$ with weak sparsity $O_H(\log n \poly(\frac{1}{\epsilon}))$. Furthermore, if $G$ has constant treewidth, then $\mathcal{O}$ has constant weak sparsity.
\end{theorem}

Our second contribution is to show that a weakly sparse spanner oracle is both  necessary and sufficient to have light subset spanners. Recall that the lightness of a subset spanner is the ratio of its weight to the weight of the minimum Steiner tree spanning the same set of terminals.

\begin{theorem}\label{thm:reduction} Let $G$ be a (general) edge-weighted graph. If $G$ has an $(1+\epsilon)$-spanner oracle $\mathcal{O}$ with weak sparsity $\wsp_{\mathcal{O}}$, then for any given set of terminals $T$, there exists a subset $(1+O(\epsilon))$-spanner  with lightness at most:
\begin{equation}
 \tilde{O}\left(\max(\wsp_{\mathcal{O}}, \epsilon^{-1})\epsilon^{-2}\right).
\end{equation}

Conversely, if for any given set of terminals $T$, $G$ has a subset $(1+\epsilon)$-spanner with lightness $\lgt$, then it has an $(1+\epsilon)$-spanner oracle with weak sparsity $O(\lgt)$.
\end{theorem}

Notation $\tilde{O}$ suppresses a $\log \frac{1}{\epsilon}$ factor. We can recover stretch bound $(1+\epsilon')$ for the subset spanner in Theorem~\ref{thm:reduction} by setting $\epsilon' = \epsilon/c$ where $c$ is the constant behind the big-O. 

We consider Theorem~\ref{thm:reduction} as a big leap in our understanding of  lightness and (weakly) sparsity. This is the first time a necessary and sufficient relationship between sparsity (of spanner oracles) and lightness is explicitly established without any special structure of the input.  Prior work~\cite{BLW17, BLW19,LS19} on light spanners exploited specific structures, such as bounded dimension or $H$-minor-freeness, to relate sparsity and lightness in a very subtle way\footnote{The work of Chechick and Wulff-Nilsen~\cite{CW16} studying general graphs does not apply to our case for two reasons: (1) they only consider spanners of stretch at least $3$ and (2) their work does not imply any black box reduction between sparsity and lightness.}.  Indeed, it has commonly been assumed that exploiting the structure of the input graph in constructing light spanners is unavoidable: one can easily come up with a graph of constant sparsity\footnote{Start with a complete graph of size $n$, we subdivide each edge sufficiently (but polynomially) many times so that the resulting graph is sparse. For each path $P$ that is subdivided from an edge in the original graph, we set weight $0$ to every edge of $P$ except for one of weight $1$. The $\mathrm{MST}$ of this graph will have weight $n-1$ while any $(1+\epsilon)$-spanner (of the new graph) for a given $\epsilon < 1$ must have weight $\Omega(n^2)$.} such that any spanner of the graph must have lightness polynomial  in $n$. That implies the gap between sparsity and lightness is polynomial in $n$. Our requirement on sparsity of an oracle is stronger than on  sparsity of a spanner in the sense the the former implies the latter, but this stronger assumption is indeed necessary by Theorem~\ref{thm:reduction}.

Another significant implication of Theorem~\ref{thm:reduction} is that the gap between sparsity and lightness is only $O(\epsilon^{-3})$ for any graph. This is somewhat surprising given the long line of research on sparse spanners and light spanners. In the Euclidean space of constant dimension $d$, it has been known since  early 90s that any set of $n$ points has a sparse spanner with sparsity $O(\epsilon^{1-d})$.  However, it took many years to figure out the optimal lightness bound: from $f(\epsilon,d)$ for some computable function $f(.)$~\cite{DHN93} in 1993, to $\left( \frac{d}{\epsilon}\right)^{-O(d)}$~\cite{DNS95} in 1995,  $\epsilon^{-O(d)}$~\cite{RS98} in 1998,  $O(\epsilon^{-2d})$~\cite{NS07} in 2007 and recently to $\tilde{O}(\epsilon^{-d})$~\cite{LS19} 2019, which is optimal~\cite{LS19}. All the proofs used  heavy machinery from Euclidean geometry. (We refer the readers to the paper by this author and Solomon~\cite{LS19} for a thorough historical discussion of this problem.) This sharply contrasts with Theorem~\ref{thm:reduction}; with a fairly easy argument to establish a spanner oracle with sparsity $O(\epsilon^{1-d})$,
 it gives a light Euclidean spanner with lightness $\tilde{O}(\epsilon^{-(d + 1)})$, without using any Euclidean geometry in the lightness proof. (Euclidean geometry is used implicitly in constructing the sparse spanner oracle.) The lightness bound we get is off the optimal bound~\cite{LS19} by just a factor of $\frac{1}{\epsilon}$. 

Theorem~\ref{thm:reduction} is not only conceptually interesting, but  it also has other applications. Observe that Theorem~\ref{thm:main} follows directly from Theorem~\ref{thm:reduction} and Theorem~\ref{thm:sparse} since $\wsp_{\mathcal{O}}  = O_H(\log n \poly(\frac{1}{\epsilon}))$ when $G$ is $H$-minor-free and $\wsp_{\mathcal{O}}  = O(1)$ when $G$ has constant treewidth. To replace the $\log n$ factor by a $\log k$ factor in Theorem~\ref{thm:sparse},  we pre-process the input graph using the distance preserving minor by Krauthgamer, Nguy{\ee}n, and Zondiner~\cite{KNZ14} to reduce $n$ to  $O(k^4)$. In the following section, we present additional applications of Theorem~\ref{thm:reduction} in different settings.

\subsection{Other applications of our techniques}\label{sec:other-apps}

In metrics of constant doubling dimension $d$, Borradaile, Le and Wulff-Nilsen~\cite{BLW19} showed that greedy spanners have  lightness $\epsilon^{-O(d)}$, improving upon  previous lightness bounds by Smid~\cite{Smid09} and Gottlieb~\cite{Gottlieb15}. In Section~\ref{sec:spare-oracle-Euc}, we construct a spanner oracle with strong sparsity $O(\epsilon^{-d})$.  Since strong sparsity implies weak sparsity (Equation~\ref{eq:wsp-ssp}), Theorem~\ref{thm:reduction} gives:

\begin{theorem}\label{thm:dd-metric}
Any metric of constant doubling dimension $d\geq 1$ has a spanner with lightness $\tilde{O}(\epsilon^{-(d+2)})$. 
\end{theorem}

Theorem~\ref{thm:dd-metric} further improves upon the bound achieved by  Borradaile, Le and Wulff-Nilsen~\cite{BLW19}. 

In SODA'08, Chan and Gupta~\cite{CG12} introduced correlation dimension of metric spaces as a way to capture global growth rate, as opposed to doubling dimension that only captures local growth rate. Correlation dimension is more general than doubling dimension in two ways: (a) a constant doubling dimension metric is a constant correlation metric and (b) there exists a constant correlation metric that has doubling dimension $\Omega(\log n)$. It should be noted that the doubling dimension of any metric space is $O(\log n)$. 

 As an application of our technique, we show for the first time that metrics of constant correlation dimension have subset spanners with sublinear lightness.

 \begin{theorem}\label{thm:light-spanner-cd} Given any terminal set $T$ in an $n$-point metric of constant correlation dimension $d$, a subset spanner for $T$ with lightness $\tilde{O}(\epsilon^{-(d/2 + 3)} \sqrt{n})$ can be constructed in polynomial time.
 \end{theorem}

By letting $T$ contain every point of the metric, we obtain a spanner of $\tilde{O}(\epsilon^{-(d/2 + 3)}\sqrt{n})$ lightness. 

We note that metrics of constant correlation dimension are not closed under taking sub-metrics: a sub-metric of a metric of constant dimension could have correlation dimension $\Omega(\log n)$ (see Figure 1.1 in~\cite{CG12} and discussions below it).  Thus, even if a light spanner construction is known, it still does not imply a light subset spanner since the standard technique that takes a sub-metric on the terminal set and applies the light spanner construction to the sub-metric does not work. However, Theorem~\ref{thm:light-spanner-cd} covers the subset spanner problem as well. 

One may ask whether it is possible to replace $\sqrt{n}$ by $\sqrt{|T|}$.  The answer seems negative. Chan and Gupta~\cite{CG12} gave an example graph with $n$ vertices, whose metric completion has constant correlation dimension, that contains a (unit-weighted) clique on $\sqrt{n}$ vertices. Thus, any subset spanner (of stretch $(1+\epsilon)$) on this clique must have lightness $\Omega(\sqrt{n}) = \Omega(|T|)$. 

We note that a subset spanner with lightness bound $O(\epsilon^{-1}|T|)$ is possible for general metrics since we can take metric completion on $T$ and then construct shallow-light trees~\cite{ABP90,ABP91,KRY93}  rooted at each point of the new metric. 

\subsection{Organization of the paper}

Section~\ref{sec:prim} reviews standard notation used in our paper.  We present a proof of Theorem~\ref{thm:stsp-ptas} in Section~\ref{sec:proof-ptas}. We construct sparse spanner oracles for minor-closed families in Section~\ref{sec:minor-close-oracles}. In Section~\ref{sec:reduction}, we present a proof of Theorem~\ref{thm:reduction}.  Finally, in Section~\ref{sec:spare-oracle-Euc}, we construct sparse spanner oracles for metric spaces.

\section{Preliminaries}\label{sec:prim}

We use $V(G)$ and $E(G)$ to denote the vertex set and the edge set of a graph $G$, respectively. When we need to explicitly specify a vertex set  $V$ and an edge set $E$ along with $G$, we write $G(V,E)$. Let $w_G : E(G)\mapsto \mathbb{R}^+$ be the weight function on edges of $G$. When the graph is clear from the context, we will drop the subscript in the weight function. We denote by $d_G(u,v)$ the shortest distance between two vertices $u$ and $v$.  For a vertex $v$ and a vertex set $V' \subseteq V$, we define the shortest distance between $v$ and $V'$, denoted by $d_G(v,V')$, to be $\min_{u \in V'}d_G(u,v)$. If $v \in V'$, then $d_G(v,V') = 0$. Let $\MST$ be a minimum spanning tree of $G$.

A \emph{walk} $W$ of length $d$ in $G$ is a sequence of vertices and edges $\{v_1,e_1,\ldots,e_d, v_{d+1}\}$ such that $v_i,v_{i+1}$ are the two endpoints of $e_i$, $1\leq i \leq d$. We call $W$ a \emph{closed walk} if $v_1 = v_{d+1}$. $W$ is a \emph{path} if no vertex is repeated; in this case, we denote the subpath of $W$ between $u$ and $v$ by $W[u,v]$. Let $W_1,W_2$ be two walks of $G$ such that the last vertex of $W_1$ is the first vertex of $W_2$. We define the composition of $W_1$ and $W_2$, denoted by $W_1\circ W_2$, to be the walk obtained by identifying the last vertex of $W_1$ and the first vertex of $W_2$. 

Let $S$ be a connected subgraph of $G$. By $w_G(S)$, we denote the total edge weight of $S$. We define the diameter of $S$, denoted by $\dm(S)$, to be $ \max_{u,v \in S}d_S(u,v)$. A shortest path $D$ in $S$ where $w_G(D) = \dm(S)$ is called a \emph{diameter path} of $S$. 

A \emph{$t$-spanner} is a subgraph of $G$ that preserves distances between all pairs of vertices up to a factor of $t$. Factor $t$ is called the \emph{stretch} of the spanner. When $t = 1+\epsilon$, we will drop the prefix $t$ in $t$-spanners. \emph{Lightness} of a $t$-spanner is the ratio of its weight to the weight of $\MST$. \emph{Sparsity} of a $t$-spanner is the ratio of its edges to vertices.  A \emph{subset $t$-spanner} is defined in a similar way, but it is only required to preserve the distances between pairs of vertices in a prescribed set $T$, called a set of \emph{terminals}. Lightness of a subset $t$-spanner is the ratio of its weight to the weight of a minimum Steiner tree spanning $T$. When $t=1+\epsilon$, we simply refer to a subset $t$-spanner as a subset spanner.

A graph $H$ is a \emph{minor} of $G$ if it can be obtained from $G$ by edge contractions, edge deletions and vertex deletions. $G$ is  $H$-minor-free if it excludes a fixed graph $H$ as a minor.  We say an edge-weighted graph $H$ is a \emph{strict minor} of $G$ if (i) $H$ is a minor of $G$, (ii) $V(H) \subseteq V(G)$ and (iii) for every edge $e \in H$ with two endpoints $x,y$, $w_H(e) = d_G(x,y)$.  If we replace every edge of $H$ by a shortest path between its endpoints in $G$, we obtain a graph, denoted by $\mathcal{D}(H)$, that we call a \emph{decompression} of $H$.

Given a terminal set $T$ of a graph $G$, Krauthgamer, Nguy{\ee}n, and Zondiner~\cite{KNZ14} showed that $G$ can be compressed by applying a minor transformation such that the distance between every pair of terminals is preserved.

\begin{lemma}[Theorem 2.1~\cite{KNZ14}]\label{lm:dp-minor} Let $T$ be a set of $k$ terminals in a graph $G$. There is a strict minor $G'$ of $G$ such that (i) $T\subseteq V(G')$, (ii) $V(G') = O(k^4)$ and $E(G') = O(k^4)$ and (iii) $d_{G'}(x,y) = d_G(x,y)$ for every two distinct terminals $x,y \in T$. Furthermore, $G'$ can be found in polynomial time.
\end{lemma}

If $G$ has bounded treewidth,  Krauthgamer, Nguy{\ee}n, and Zondiner~\cite{KNZ14} proved a stronger version of Lemma~\ref{lm:dp-minor}.

\begin{lemma}\label{lm:tw-tpm-linear}  Let $T$ be a set of $k$ terminals in a graph $G$ of treewidth at most $\mbtw$. There is a strict minor $G'$ of $G$ such that (i) $T\subseteq V(G')$, (ii) $V(G') = O(\mbtw^3 k)$  and (iii) $d_{G'}(x,y) = d_G(x,y)$ for every two distinct terminals $x,y \in T$. Furthermore, $G'$ can be found in polynomial time.
\end{lemma}
\section{Proof of Theorem~\ref{thm:stsp-ptas}} \label{sec:proof-ptas}

In this section, we give a proof of Theorem~\ref{thm:stsp-ptas}, given Theorem~\ref{thm:main} and a singly exponential time algorithm that can solve Subset TSP in graph of treewidth-$\mbtw$ in time $2^{O(\mbtw)}n^{O(1)}$ time (Appendix~\ref{app:dp}).

Given an $H$-minor-free graph $G$, we apply Theorem~\ref{thm:main} to obtain a subset spanner $S$ for terminal set $T$ of weight $w(S) \leq O(\log k\polye)w(\st) = O(\log k\polye)w(\opt)$. By the contraction decomposition theorem of  Demaine, Hajiaghayi and Kawarabayashi, given any integer $g \geq 1$, one can partition the edge set $E(S)$ of $S$ into $g$ parts $\mathcal{X}=\{ X_1,X_2,\ldots, X_g\}$ such that for any $i \in [1,g]$, contracting any set of edge $X_i$ in $S$ gives a graph of  treewidth at most $O_H(g)$. We denote by $S/X_i$ the graph obtained from $S$ by contracting $X_i$. 

Let $g = \lceil \frac{w(S)}{\epsilon \opt} \rceil$, and $X = \arg\min_{X_i \in \mathcal{X}} w(X_i)$. Then, $w(X) \leq \frac{w(S)}{g} = \epsilon \opt$, and that $S/X$ has treewidth at most $O_H(g) = O_H(\log k\polye)$. When we contract $X$ in $S$, we might contract a terminal to a non-terminal or a terminal to another terminal. In the former case, we designate the non-terminal to be a new terminal of the contracted graph, and in the later case, we delete one terminal from $T$.   Let $T_X$ be the resulting set of terminals in $S/X$. We find an optimal tour $\opt_X$ spanning $T_X$ in $S/X$ in time $2^{O_H(g)}n^{O(1)} = n^{O_H(\polye)}$. Note that $w(\opt_X) \leq w(\opt)$. We then can convert $\opt_X$ to a tour spanning $T$ by uncontracting $X$ (and adding a matching between odd vertices of $T$ if necessary) at a cost of $O(w(X)) = O(\epsilon w(\opt))$. Thus, the obtained tour has weight at most $(1 + O(\epsilon))w(\opt)$. By scaling $\epsilon$ appropriately, we obtain a tour of weight at most $(1 + \epsilon)w(\opt)$.

In the following section, we focus on proving Theorem~\ref{thm:main}.

 \section{Weakly sparse spanner oracles for minor-closed families}\label{sec:minor-close-oracles}
 
 In this section, we show how to construct weakly sparse spanner oracles as  stated in Theorem~\ref{thm:sparse}. This implies Theorem~\ref{thm:main} by Theorem~\ref{thm:reduction}. We first relate weakly sparse spanner oracles to approximate distance preserving minors. 
 
 \subsection{Weakly sparse spanner oracles from approximate terminal distance preserving minors}
 
 We say a strict minor $G'$ of $G$ is an \emph{$(1+\epsilon)$-approximate terminal distance preserving minor} for a terminal set $T$ if $T\subseteq V(G')$ and $d_{G'}(x,y) \leq (1+\epsilon)d_{G}(x,y)$ for every two distinct terminals $x,y \in T$. We call vertices in $V(H)\setminus T$ \emph{Steiner vertices}.

 \begin{lemma} \label{lm:minor-to-oracle}  If $G$ is $H$-minor-free and has an $(1+\epsilon)$-approximate terminal distance preserving minor with at most $s(\epsilon) |T|$ Steiner vertices for any terminal set $T$, then it has a spanner oracle with weak sparsity $O(s(\epsilon))$. 
 \end{lemma}
\begin{proof}
We construct an oracle $\mathcal{O}$ as follows.  Let $T$ be any set of terminal and $\ell > 0$ be a real positive number given as inputs to  $\mathcal{O}$. We first find an $(1+\epsilon)$ approximate distance preserving minor $G'$ of $G$ for $T$. We then remove every edge of length at least $2\ell$ from $G'$. Let $G''$ be the resulting graph. We return decompressed graph $\mathcal{D}(G'')$ as the output of the oracle. 

We first bound the weight of $\mathcal{D}(G'')$. Observe that $G'$ is $H$-minor-free since it is a minor of $G$. Thus, $|E(G')| \leq O_H(|V(G')|)$ (see~\cite{Kostochka82}). Since $G'$ has at most $s(\epsilon)|T|$ Steiner vertices by the assumption of the lemma, $|V(G')| \leq (s(\epsilon) + 1)|T|$. This implies:
\begin{equation*}
w\left(E(\mathcal{D}(G''))\right)  \leq w(E(G'')) \leq 2\ell |E(G')| \leq 2 \ell O_H(s(\epsilon)) |T| = O_H(s(\epsilon)) |T|\ell.
\end{equation*}
Therefore, the weak sparsity of $\mathcal{O}$ is $O_H(s(\epsilon))$. 

To complete the proof, it remains to show that for every two distinct terminals $x,y \in T$ such that $d_G(x,y) \in [\ell/8,\ell]$, their distance is preserved up to $(1+\epsilon)$ factor in $\mathcal{D}(G'')$. Since $G'$ is an $(1+\epsilon)$-approximate distance preserving minor of $G$, $d_{G'}(x,y) \leq (1+\epsilon)d_G(x,y) < 2\ell$ when $\epsilon < 1$. Thus, every edge in the shortest path between $x$ and $y$ is kept in $G''$. Hence $d_{G''}(x,y) \leq (1+\epsilon)d_G(x,y)$.  Since $d_{\mathcal{D}(G'')}(x,y) \leq d_{G''}(x,y)$, $d_{\mathcal{D}(G'')}(x,y) \leq (1+\epsilon)d_G(x,y)$ as desired.
\end{proof} 
 
\begin{proof}[Proof of Theorem~\ref{thm:sparse} for bounded treewidth graphs] Since $G$ has constant treewidth, it is $K_r$-minor-free for $r = \mbtw(G)+2$. By Lemma~\ref{lm:tw-tpm-linear}, $G$ has an exact (and hence $(1+\epsilon)$-approximate) distance preserving minor with at most $O(|T|)$ Steiner vertices for any set of terminal $T$. With Lemma~\ref{lm:minor-to-oracle}, this implies that $G$ has a spanner oracle with weak sparsity $O(1)$.
\end{proof}

One may ask whether we can apply Lemma~\ref{lm:minor-to-oracle} to obtain an oracle with constant weak sparsity for $H$-minor-free graphs. However, since the best-known approximate distance preserving minors in planar graphs have a quadratic number of Steiner vertices~\cite{CGH16}, Lemma~\ref{lm:minor-to-oracle} only gives us a spanner with lightness linear in $k$. Instead, in the following section, we pursue a different technique to construct a nearly light subset spanner for $H$-minor-free graphs.

\subsection{Weakly sparse spanner oracles from shortest path separators}

 Our starting point is the construction of single-source spanners for planar graphs by Klein (Theorem 4.1~\cite{Klein06}). We show that Klein's planar single-source spanners~\cite{Klein06} are light even without planarity. 

\begin{lemma} \label{lm:ss-spanner}Let $p$ be a vertex and $P$ be a shortest path in a graph $G$. Let $y_0 \in P$ be such that $d_G(p,y_0) = d_G(p, P)$. Let $R = d_G(p, P)$. Fix an endpoint of $P$ to be its left-most vertex. Let $\{y_1,\ldots, y_I\} \subseteq V(P)$ be a maximal set of vertices such that $y_{i}$ is the closest point to the right of $y_{i-1}$ such that:
\begin{equation} \label{eq:ss-spanner}
(1+\epsilon)d_G(p,y_i) < d_G(p,y_{i-1}) + d_P(y_{i-1}, y_i) \qquad 1 \leq i \leq I
\end{equation}
We symmetrically define a maximal set of points $(y_{-1}, y_{-2}, \ldots, y_{-J})$ to the left of $y_0$ on $P$.  Let $\mathcal{Q} = \{Q_{-J}, Q_{-J+1}, \ldots, Q_{-1}, Q_0, Q_1, \ldots,Q_{I}\}$ be a set of shortest paths where $Q_i$ is a shortest $p$-to-$y_i$ path in $G$, $-J \leq i \leq I$.  Then, it holds that:
\begin{enumerate}
\item[(1)] $d_{\mathcal{Q} \cup P}(p,q) \leq (1+\epsilon) d_G(p,q)$ for every $q \in P$.
\item[(2)] $w(\mathcal{Q}) \leq 8\epsilon^{-2}R$.
\item[(3)] $I\leq 8\epsilon^{-2}$ and $J \leq 8\epsilon^{-2}$.
\item[(4)] $d_P(y_0, y_{I}) \leq 4\epsilon^{-1}R$  and $d_P(y_{-J}, y_{0}) \leq 4\epsilon^{-1}R$. 
\end{enumerate}
\end{lemma}
\begin{proof}

\begin{figure}
\centering
\vspace{-20pt}
\includegraphics[scale = 1.0]{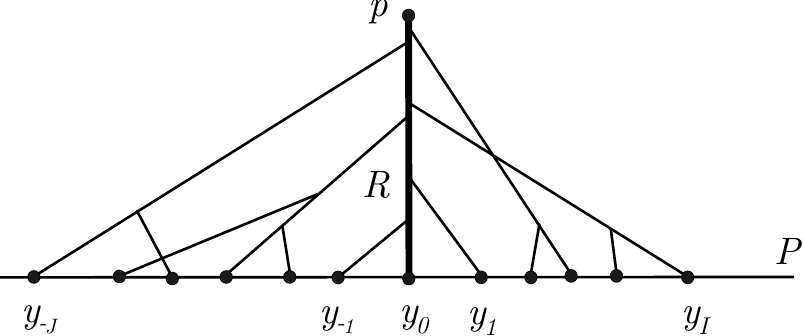}
\caption{A single-source spanner constructed by Klein's algorithm. The length of the thick path is $R$. }
\label{fig:klein-ssp}
\end{figure}
See Figure~\ref{fig:klein-ssp} for an illustration. Property (1) follows directly from the maximality of the set of points $y_{-J}, \ldots, y_0, \ldots, y_I$. We now show property (4). By symmetry, it is sufficient to show that:
\begin{equation}\label{eq:ss-spanner-proof-of-4}
d_P(y_0,y_I) \leq 4\epsilon^{-1}R 
\end{equation}
Suppose otherwise. Then, there exists $\ell \in \{0, \ldots, I-1\}$ such that $d_P(y_0,y_\ell) \leq  4\epsilon^{-1}R$ and $d_P(y_0,y_{\ell+1}) > 4\epsilon^{-1}R$. We have:
\begin{equation*}
\begin{split}
(1+\epsilon)d_G(p,y_{\ell+1})  &\geq (1+\epsilon)(d_G(y_0,y_{\ell+1}) - d_G(p,y_0)) \qquad (\mbox{by triangle inequality}) \\
&= (1+\epsilon)(d_G(y_0,y_{\ell+1}) +  d_G(p,y_0)) - 2(1+\epsilon)d_G(p,y_0)\\
&\geq (d_G(y_0,y_{\ell+1}) +  d_G(p,y_0)) + \epsilon d_G(y_0,y_{\ell+1}) - 2(1+\epsilon)d_G(p,y_0) \\
&= (d_G(y_0,y_\ell) +  d_G(p,y_0)) + d_P(y_\ell, y_{\ell+1}) + \epsilon d_G(y_0,y_{\ell+1}) - 2(1+\epsilon)d_G(p,y_0)\\
&\geq d_G(p,y_\ell) + d_P(y_\ell,y_{\ell+1}) + \epsilon d_G(y_0,y_{\ell+1}) - 2(1+\epsilon)d_G(p,y_0)\\
&> d_G(p,y_\ell) + d_P(y_\ell,y_{\ell+1}) + 4R - 2(1+\epsilon)R \qquad \mbox{(since } \epsilon d_P(y_0,y_{\ell+1}) > 4R) \\
&\geq d_G(p,y_\ell) + d_P(y_\ell,y_{\ell+1}) \qquad \mbox{(since } \epsilon < 1)
\end{split}
\end{equation*}
contradicting Equation~\eqref{eq:ss-spanner}. Thus, no such $\ell$ exists. 

Proof of property (3) is similar to that of Theorem 4.1 of Klein~\cite{Klein06}, and we defer to Appendix~\ref{app:missing}. To prove (2), we sum both sides of Equation~\eqref{eq:ss-spanner} for every $1 \leq i \leq I$. 
\begin{equation}
\begin{split}
(1+\epsilon)(w(Q_0) + \ldots + w(Q_I)) &\leq (w(Q_0) + \ldots + w(Q_I)) - w(Q_I) + d_P(y_0,y_I) \\
&\leq (w(Q_0) + \ldots + w(Q_I)) + 4\epsilon^{-1}R \qquad \mbox{(by Equation~\eqref{eq:ss-spanner-proof-of-4})}\\
\end{split} 
\end{equation}
That implies $w(Q_0) + \ldots + w(Q_I) \leq 4\epsilon^{-2}R$. By a symmetric argument, we can show that $w(Q_{-J}) + \ldots + w(Q_0) \leq 4\epsilon^{-2}R$.

\end{proof}

Let \textsc{SSSpanner($G, P,p,\epsilon$)}  be the set of paths rooted at the same vertex $p$ obtained by applying the construction in Lemma~\ref{lm:ss-spanner} to a shortest path $P$, a source vertex $p$ and distance parameter $\epsilon$. We can also generalize Klein bipartite spanners (Theorem 5.1~\cite{Klein06}) for non-planar graphs by using Lemma~\ref{lm:ss-spanner}. We believe that this result is of independent interest, though we do not use it in this paper. The proof is deferred to Appendix~\ref{app:path-to-path-proof}.

\begin{corollary} \label{cor:path-to-path-spanner}
Let $W$ be a walk and $P$ be a shortest path in a graph $G$. We denote by $R$ the distance between $W$ and $P$. That is $R = \min_{v \in W}d_G(v, P)$. Then, there is a subgraph $H$ of $G$ such that:
\begin{enumerate}
\item For every $p \in W, q \in P$, $d_{H\cup P}(p,q) \leq (1+\epsilon)d_{G}(p,q)$.
\item $w(H) \leq O(\epsilon^{-3})w(W) + O(\epsilon^{-2})R$.
\end{enumerate}
\end{corollary}

Next, we construct a spanner that preserves distances between terminal pairs prescribed by a set $\mathcal{Q}$ of shortest paths. The first step toward the construction is the following claim.

\begin{claim}\label{clm:multiple-source-spanner}Let $P$ be a shortest path of an edge-weighted graph $G$. Let $\mathcal{Q} = \{Q_1,Q_2,\ldots, Q_r\}$ be a set of shortest paths in $G$ such that $Q_i \cap P \not= \emptyset$ and $w(Q_i) \leq \ell$, for every $1\leq i \leq r$. We denote the endpoints of each $Q_i$ by $s_i$ and $t_i$. Let $k$ be the number of distinct endpoints of paths in $\mathcal{Q}$. There is a subgraph $H$ of $G$ with weight at most $O(k\epsilon^{-2}\ell)$ such that $d_H(s_i,t_i) \leq (1+\epsilon)d_G(s_i,t_i)$ for every  $1\leq i \leq r$.
\end{claim}
\begin{proof}

We first delete every edge of $G$ of length more than $\ell$ since no path in $\mathcal{Q}$ can contain such an edge. Let $X = \{x_1,x_2,\ldots, x_k\} $ be the set of endpoints of all paths in $\mathcal{Q}$. Let $R_j = d_G(x_j, P)$ and $y_j$ be the closet vertex of $x_j$ in $P$. Since $Q_i \cap P \not= \emptyset$ and $w(Q_i) \leq \ell$ for every $i$, $d_G(x_j, P) \leq \ell$ for every $1\leq j \leq k$.  For each $j$, let $\mathcal{Q}_j \leftarrow$ \textsc{SSSpanner}($G, P,x_j,\epsilon$).  Let $P_j$ be a minimal subpath of $P$ that contains every vertex of distance (in $P$) at most $4\epsilon^{-1}\ell$ from $y_j$. Since $P_j$ has no edge of length more than $\ell$, $w(P_j) \leq (8\epsilon^{-1} + 2)\ell$. Since $|R_j| \leq \ell$, by (4) of Lemma~\ref{lm:ss-spanner}, we have:

\begin{observation}\label{obs:Pj-vs-ends-of-Qj}
$P_j$ contains all endpoints on $P$ of paths in $\mathcal{Q}_j$.
\end{observation}

\noindent  Recall paths in $\mathcal{Q}_j$ share endpoint $x_j$.  Let: 
	\begin{equation}
	H = \bigcup_{j=1}^{k}\left((\cup_{Q\in \mathcal{Q}_j}Q)\cup P_j\right)
	\end{equation}
We first bound the weight of $H$. For any $j$, $1\leq j \leq k$, by (2) of Lemma~\ref{lm:ss-spanner},
\begin{equation*}
w\left(\cup_{Q\in \mathcal{Q}_j}Q\right) \leq O(\epsilon^{-2})R_j  \leq O(\epsilon^{-2}\ell)
\end{equation*}
Thus, $w(H) \leq O(k\epsilon^{-2})\ell$.

\begin{figure}
\vspace{-20pt}
\centering
\includegraphics[scale = 1.0]{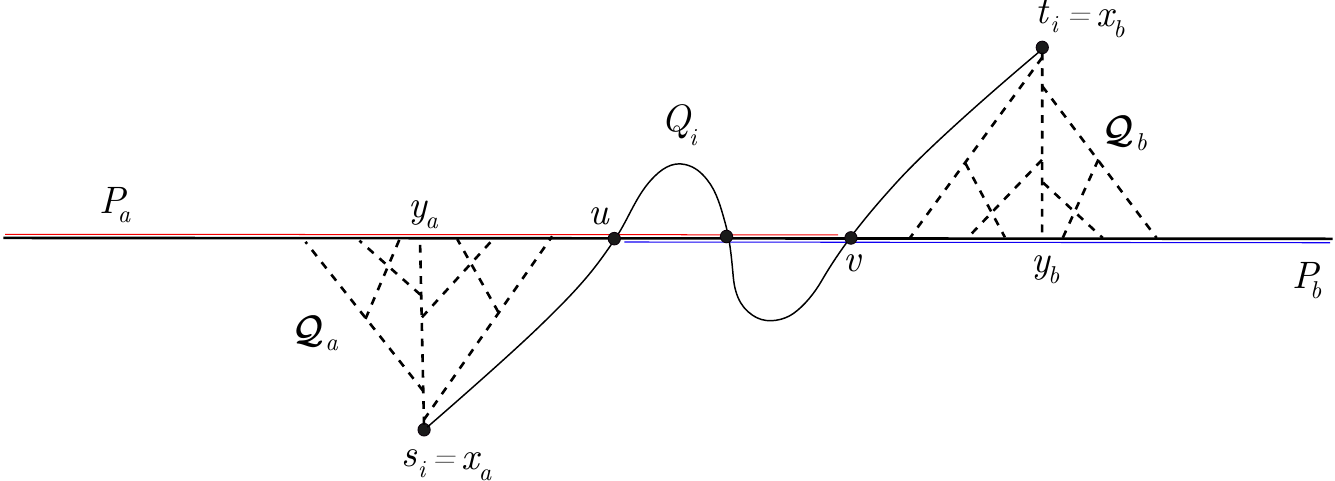}
\caption{Shortest path $P$ is the straight line and shortest path $Q_i$ between two terminals $s_i,t_i$ is the thin black curve. $P_a$ and $P_b$ are highlighted red and blue, respectively.}
\label{fig:p-t-p}
\end{figure}

We now show that $d_H(s_i) \leq (1+\epsilon) d_G(s_i,t_i)$ for any $1\leq i \leq r$. Let $u, v$ be the first vertex and the last vertex (from $s_i$) in $Q_i \cap P$, respectively. Suppose that $x_a = s_i$ and $x_b = t_i$ for some $a,b$, $1\leq a,b\leq k$ (see Figure~\ref{fig:p-t-p}). Since  $d_P(y_a, u) = d_G(y_a, u) \leq d_G({s_i,y_a}) + d_G({y_a,u}) \leq 2\ell$ which is at most $(4\epsilon^{-1}+1) \ell$ when $\epsilon < 1$. Thus, $u \in P_a$. Similarly, we can show that $v \in P_a$. That implies:

\begin{observation}\label{obs:uv-in-P}
Subpath $P[u,v]$ of $P$ is a subgraph of $H$.
\end{observation}

By a similar argument, we can show that $u,v$ both are in $P_b$ (see Figure~\ref{fig:p-t-p}). By (1) of Lemma~\ref{lm:ss-spanner} and Observation~\ref{obs:Pj-vs-ends-of-Qj}, we have:

\begin{equation} \label{eq:supp-si-ti}
d_H(s_i,u) \leq (1+\epsilon)d_G(s_i,u)   \quad \mbox{and} \quad   d_H(v,t_i) \leq d_G(v,t_i)
\end{equation}

\noindent Since $P$ is a shortest path of $G$, $w(P[u,v]) = w(Q_i[u,v)]$ and both have length at most $\ell$.  Thus, we have:
\begin{equation*}
\begin{split}
d_H(s_i,t_i) &\leq d_H(s_i,u) + d_H(u,v) + d_H(v,t_i) \\
&= d_H(s_i,u) +  w(Q_i[u,v]) +  d_H(v,t_i)\qquad \mbox{(by Observation~\ref{obs:uv-in-P})}\\
&\leq (1+\epsilon)d_G(s_i,u)  + w(Q_i[u,v] )+ (1+\epsilon) d_G(v,t_i) \qquad \mbox{(by Equation~\eqref{eq:supp-si-ti})}\\
&= (1+\epsilon)w(Q_i[s_i,u]) + w(Q_i[u,v] )+ (1+\epsilon)w(Q_i[v,t_i])\qquad \mbox{(since }Q_i \mbox{ is a shortest path)}\\
&\leq (1+\epsilon)w(Q_i[s_i,t_i]) = (1+\epsilon)d_G(s_i,t_i)
\end{split}
\end{equation*} 
\end{proof}

For any two paths $P$ and $Q$, we say $P$ \emph{crosses} $Q$ if $P\cap Q \not= \emptyset$. We say $P$ crosses a set of paths $\mathcal{Q}$ if there exists a path $Q \in \mathcal{Q}$ such that $P$ crosses $Q$. We now extend Claim~\ref{clm:multiple-source-spanner} to the case where a constant number of shortest paths cross $\mathcal{Q}$. Since we will be using shortest path separators, this  case would naturally arise in our final construction.

\begin{claim}\label{clm:multi-paths-spanner}
Let $\mathcal{P}$ be a set of shortest paths in an edge-weighted graph $G$. Let $\mathcal{Q} = \{Q_1,Q_2,\ldots, Q_r\}$ be another set of shortest paths in $G$ such that $Q_i$ crosses $\mathcal{P}$ and $w(Q_i) \leq \ell$, for every $1\leq i \leq r$. We denote the endpoints of each $Q_i$ by $s_i$ and $t_i$. Let $k$ be the number of distinct endpoints of paths in $\mathcal{Q}$. There is a subgraph $H$ of $G$ with weight at most $O(k\epsilon^{-2}\ell |\mathcal{P}|)$ such that $d_H(s_i,t_i) \leq (1+\epsilon)d_G(s_i,t_i)$ for every  $1\leq i \leq r$. Furthermore, $H$ can be found in polynomial time.
\end{claim} 
\begin{proof}
Fix an ordering of paths $P_1, P_2,\ldots, P_h$ in $\mathcal{P}$ where $h = |\mathcal{P}|$. For each path $P_j$, $1 \leq j \leq h$, let $\mathcal{Q}_j$ be the set of paths in $\mathcal{Q}$ such that each path in $\mathcal{Q}$ crosses $P_j$ and does not cross any $P_i$ for all $i < j$. Let $H_j$ be the subgraph of $G$ obtained by applying Claim~\ref{clm:multiple-source-spanner} with parameters $G, P_j, \mathcal{Q}_j,\epsilon$ and $\ell$. Let $H = \cup_{j = 1}^h H_j$. Then, $w(H) \leq \sum_{i=1}^h w(H_i) = O(k\epsilon^{-2}\ell |\mp|)$. The stretch guarantee of $H$ follows directly from Claim~\ref{clm:multiple-source-spanner}. 
\end{proof}

Let \textsc{PTPSpanner($G, \mathcal{P},\mathcal{Q},\ell, \epsilon$)} (\textsc{PTP} means path-to-path.) be the subgraph of $G$ obtained by applying Claim~\ref{clm:multi-paths-spanner} to $\mathcal{P},\mathcal{Q},\ell$ and $\epsilon$. We use this to construct a sparse spanner oracle.

\paragraph{A weakly sparse spanner oracle} Suppose that a terminal set $T$ and a real positive number $\ell$ are given as  inputs to  oracle $\mathcal{O}$. Let $\mathcal{Q}$ be the set of $O(|T|^2)$ shortest paths between all pairs of terminals. The oracle will call and return  {\sc EllCloseSpanner}$(G, T,\mathcal{Q}, \ell, \epsilon)$ in Figure~\ref{fig:ell-spanner-const}. This algorithm returns a subgraph of $G$ that preserves all the distances between every two distinct terminals in $T$ whose distance is at most $\ell$.  The algorithm uses the following shortest path separator theorem for $H$-minor-free graphs by Abraham and Gavoille~\cite{AG06}.
  
\begin{lemma}[Theorem 1~\cite{AG06}] \label{lm:sep-minor-free} For every connected $H$-minor-free graph $G$ of $n$ vertices, there is a family of $\gamma$ sets of paths $\Omega = \{\mathcal{P}_1, \mathcal{P}_2, \ldots, \mathcal{P}_\gamma\}$ of $G$ such that:
\begin{enumerate}[nolistsep,noitemsep]
\item $\sum_{i=1}^\gamma |\mathcal{P}_i| = O_H(1)$.
\item $\mathcal{P}_1$ is a set of shortest paths of $G$ and $\mathcal{P}_i$ is a set of shortest paths of $G\setminus V(\cup_{j < i}  \mathcal{P}_i)$ for $i\geq 2$.
\item Connected components of $G \setminus V(\Omega)$ have size at most $n/2$. 
\end{enumerate}
\end{lemma}

\begin{figure}
\vspace{-20pt}
\centering
\fbox{\begin{varwidth}{\dimexpr\textwidth-2\fboxsep-2\fboxrule\relax}
   \begin{tabbing}
 {\sc EllCloseSpanner}$(G, T,\mathcal{Q}, \ell, \epsilon)$\\
  \qquad \= \textbf{if} $|T| \leq 1$ return $\emptyset$\\
  \> $S \leftarrow \emptyset$\\
  \> $\mathcal{P}_0 \leftarrow \emptyset$; $\Omega \leftarrow \{\mathcal{P}_1,\ldots, \mathcal{P}_\gamma\}$ as in Lemma~\ref{lm:sep-minor-free}\\
  \> \textbf{for} $i \leftarrow 1$ to $\gamma$\\
  \> \qquad \= $G_i \leftarrow G\setminus (\cup_{j=0}^{i-1} \mathcal{P}_j)$\\
  \>\> $\mathcal{Q}_i \leftarrow$ the set of paths in $\mathcal{Q}$ that cross $\mathcal{P}_i$\\
  \>\> $S \leftarrow S \cup $ \textsc{PTPSpanner($G_i, \mathcal{P}_i,\mathcal{Q}_i,\ell, \epsilon$)}\\
  \>\> $\mathcal{Q} \leftarrow \mathcal{Q}\setminus \mathcal{Q}_i$\\
  \> \textbf{for each} component $G'$ of $G\setminus V(\Omega)$\\
  \> \qquad \= $T' \leftarrow T\cap V(G')$\\
  \>\> $\mathcal{Q}' \leftarrow $ remaining paths in $\mathcal{Q}$ with both endpoints in $T'$\\
  \>\> $S  \leftarrow S \cup$ \textsc{EllCloseSpanner}$(G', T',\mathcal{Q}', \ell,\epsilon)$\\
  \qquad \= return $S$
\end{tabbing}
\end{varwidth}}
\caption{An algorithm that constructs a spanner preserving distances prescribed by $\mathcal{Q}$.}
\label{fig:ell-spanner-const}
\end{figure}

We now show that the oracle has desired weak sparsity. We represent the execution of procedure {\sc EllCloseSpanner}$(G, T,\mathcal{Q}, \ell, \epsilon)$ by a recursion tree $\mathcal{T}$ where each node represents a recursive call on a subgraph, say $K$ of $G$, and its child nodes are recursive calls on connected components of $K \setminus \Omega_K$. Here $\Omega_K$ is a shortest-path separator of $K$ as in Lemma~\ref{lm:sep-minor-free}. The root node of $\mt$ is a call on $G$. Since the size of child graphs in recursive calls is at most half the size of the parent graph, $\mt$ has depth $O(\log n)$. 

We note that in each recursive call {\sc EllCloseSpanner}$(G', T',\mathcal{Q}', \ell, \epsilon)$ in the algorithm in Figure~\ref{fig:ell-spanner-const}, paths in $\mathcal{Q}'$ are shortest paths of $G'$ since they are shortest paths in $G$. Observe that none of the paths in $\mathcal{Q}'$ of the second \textbf{for} loop contains a vertex of $V(\Omega)$ since any path of $\mathcal{Q}$ that crosses at least one set of paths in $\Omega$ will be removed in the first \textbf{for} loop. 

We first bound the total weight of $S$ that is the output of {\sc EllCloseSpanner}$(G, T,\mathcal{Q}, \ell, \epsilon)$. Consider $i$-th iteration in the first \textbf{for} loop in the algorithm in Figure~\ref{fig:ell-spanner-const}. We have:
\begin{observation} \label{obs:Qi-sp-Gi} $\mathcal{Q}_i$ is a set of shortest paths in $G_i$. 
\end{observation}

By Claim~\ref{clm:multi-paths-spanner} and (1) of Lemma~\ref{lm:sep-minor-free}, the total weight of $S$ after the first for loop is at most:

\begin{equation*}
O(|T|\epsilon^{-2}\ell \sum_{i=1}^\gamma |\mathcal{P}_i|) = O_H(|T|\epsilon^{-2}\ell )
\end{equation*}

That implies at each level of $\mathcal{T}$, the weight of the returned subgraph of each node is $O_H(|T|\epsilon^{-2}\ell )$ plus the weight of the subgraphs returned from recursive calls. Since the depth of $\mathcal{T}$ is $O(\log n)$,  $w(S) \leq O_H(|T|\epsilon^{-2}\ell \log n)$. Thus, the weak sparsity of the oracle is $ O_H(\epsilon^{-2}\log n)$. 

To complete the proof of Theorem~\ref{thm:sparse}, it remains to show that $d_{S}(x,y) \leq (1+\epsilon) d_G(x,y)$ for every two distinct  terminals  $x,y \in \mathcal{T}$  whose distance is at most $\ell$. Let $Q_{x,y}$ be the shortest path between $x,y$ in $\mathcal{Q}$. By triangle inequality, we can assume that $Q_{x,y}$ contains no other terminals except $x$ and $y$. Since the algorithm only stops after each component of $G$ contains at most one terminal, $Q_{x,y}$  must be removed from $\mathcal{Q}$ at some node of $\mathcal{T}$, say $\tau$. More precisely, $Q_{x,y}$ is removed in some iteration, say $i$, in the first for loop of $\tau$. By Observation~\ref{obs:Qi-sp-Gi} and Claim~\ref{clm:multi-paths-spanner}, we have:
\begin{equation*}
d_{S}(x,y) \leq (1+\epsilon)d_{G_i}(x,y) = (1+\epsilon)d_{G}(x,y)
\end{equation*}

\section{Proof of Theorem~\ref{thm:reduction}}\label{sec:reduction}

In this section, we only consider spanner oracles with weak sparsity. Thus, we simply use \emph{sparse} and \emph{sparsity} to refer to  \emph{weakly sparse} and \emph{weak sparsity}, respectively. We first show that  sparse spanner oracles are necessary to construct light subset spanners.

\subsection{Light subset spanners imply sparse spanner oracles}\label{sec:light-to-sparse}

We show that the existence of a light subset spanner for any given set of terminals implies a sparse spanner oracle. Let $\mathcal{A}$ be an algorithm that given a set of terminal set $T$ in a graph $G$, returns a subset $(1+\epsilon)$-spanner, denoted by $\mathcal{A}(T,G)$, for $T$ with lightness $\lgt$. 

Given two disjoint subset of vertices $S_1,S_2$ of $G$, we define the distance between $S_1$ and $S_2$, denoted by $d_G(S_1,S_2)$, to be $\min_{s_1\in S_1, s_2\in S_2}d_G(s_1,s_2)$.

Suppose that $T$ and $\ell$ are given as an input to oracle $\mathcal{O}$ that we will construct. Let $M_T$ be the (complete) graph obtained by taking metric completion on $T$. Let $\MST$ be the minimum spanning tree of $M_T$. We remove from $\MST$ any edge of length bigger than $\ell$ to obtain a spanning forest $F$.  Let $\mathcal{T} = \{T_1,\ldots, T_k\}$ be the partition of $T$ induced by $F$.  We return $\cup_{T_i \in \mathcal{T}} \mathcal{A}(T_i,G)$ as the output of the oracle.

By the cut property of minimum spanning trees, $d_{M_T}(T_i,T_j) >  \ell$ for every $i\not= j $. Thus, the oracle does not need to preserve pairwise distances between terminal pairs in two different sets of $\mathcal{T}$. If $t_1,t_2$ are in the same set, say $T_i$, it is guaranteed that $d_{\mathcal{A}(T_i,G)}(t_1,t_2) \leq (1+\epsilon)d_G(t_1,t_2)$ since $\mathcal{A}$ is a subset $(1+\epsilon)$-spanner. Thus, their distance is preserved in $\mathcal{O}(T,\ell)$ as well.

It remains to bound the sparsity of $\mathcal{O}$. Recall that every edge of the tree spanning $T_i$ of $F$, denoted by $F[T_i]$,  has weight at most $\ell$. Thus, $w(F[T_i]) \leq \ell |T_i|$. Since the weight of the Steiner tree in $G$ for $T_i$ is at most $w(F[T_i])$, by the lightness assumption of $\mathcal{A}$, we deduce that
\begin{equation*}
w(\mathcal{A}(T_i,G)) \leq \lgt w(F[T_i]) = \lgt \ell |T_i|
\end{equation*}
Thus, $w(\mathcal{O}(T,\ell))  \leq \lgt   \ell  \sum_{T_i \in \mathcal{T}}|T_i| = \lgt \ell |T|$. Therefore, $\mathcal{O}$ has weak sparsity $O(\lgt)$.

\subsection{Sparse spanner oracles imply light subset spanners}\label{sec:sparse-to-light}

As mentioned in Section~\ref{sec:our-technique}, the notion of sparse spanner oracles is directly inspired by the way sparse spanners were used to construct light spanners in prior work~\cite{CW16, BLW17,BLW19,LS19}. Thus, it is natural to expect that we will use the same technique, namely \emph{iterative clustering}, to construct a light subset spanner.   The technique was first discovered by Chechick and Wulff-Nilsen~\cite{CW16} to solve the light $t$-spanner problem in general graphs. It was refined to the light $(1+\epsilon)$-spanner problem in $H$-minor-free graphs by Borradaile, Le and Wulff-Nilsen~\cite{BLW17}. The later idea was then adapted to solve the same problem in doubling metrics~\cite{BLW19} and Euclidean spaces~\cite{LS19}. Our proof closely follows the presentation of Borradaile, Le and Wulff-Nilsen in~\cite{BLW19}. Since several parts of the argument appeared earlier in the work by Chechick and Wulff-Nilsen~\cite{CW16} and the work by the same authors~\cite{BLW17}, we will refer to the argument as  \emph{BCLW technique}.

Our major contribution is to identify parts of the proofs in BCLW technique where special properties of the input were used to establish sparsity, and then replace them with sparse spanner oracles.  This eliminates the need for special properties of the input from the proof. Another contribution of this work is to frame their subtle argument in terms of a single Credit Lemma (see Lemma~\ref{lm:credit-informal}). This allows us to draw a clearer picture of how sparse spanner oracles fit into the construction, and also significantly simplify the lightness bound proof. Since the Credit Lemma is just a different way to look at a known technique, we only provide details of the modification. The full proof is deferred to Appendix~\ref{app:cluster-const} for reference.

We first take the metric completion of $G$ on $T$ to obtain an edge-weighted complete graph $M_T$. That is, each edge of $M_T$ has weight equal to the length of the shortest path between the two corresponding terminals in $G$. We can think of $M_T$ as a metric without any other special property.  We mostly work with $M_T$. Since edges of $M_T$ may not exist in $G$. To avoid confusion between edges of $M_T$ and $G$, we use a map  $\kappa: E(M_T)\rightarrow 2^{E(G)}$ that maps each edge $e$ to a shortest path $\kappa(e)$ between $e$'s endpoints in $G$. For a subset of edges $X$ of $M_T$, we define $\kappa(X) = \cup_{e\in X}\kappa(e)$, which is a subgraph of $G$.

We first set up the iterative clustering framework in the same way previous work did~\cite{BLW17,BLW19,LS19}. Let $\mathrm{MST}$ be the minimum spanning tree of $M_T$. It is well know that $\mathrm{MST} \leq 2w(\st)$ where $\st$ is an optimal Steiner tree spanning $T$ in $G$. We will construct a subset spanner, denoted by $S$, iteratively. Initially, $S$ has $V(S) = E(S) = \emptyset$. Let $k = |V(M_T)|$.  For each edge $e \in M_T$ of weight at most $\frac{w(\MST)}{\epsilon k^2}$, we add $\kappa(e)$ to $S$. Since there are at most $k(k-2)/2$ such edges, the total weight of all the edges is bounded by $O(w(\MST)\epsilon^{-1}) ~= ~O(w(\st))\epsilon^{-1}$. 

Let $w_0 = \frac{w(\MST)}{k^2}$.  We abuse notation by using $E(M_T)$ to denote the set of edges of $M_T$ weight more than $w_0/\epsilon$. Note that we only need to deal with terminal pairs whose edges in $M_T$ have weight a least  $\frac{w_0}{\epsilon}$ since the shortest paths between other pairs have been added to $S$. Recall that every edge in $E(M_T)$ has weight at most $w(\MST)  = k^2w_0$. Following BCLW technique, we partition edges in $E(M_T)$ into $O(\log \frac{1}{\epsilon})$ sets $E_1, E_2,\ldots,  E_J$ with $J = \lceil \log \frac{1}{\epsilon}\rceil$. Each $E_j$ is an exponential scale $E_j^{1}, \ldots, E_j{^I}$ with $I = \lceil \log_{1/\epsilon} k^2\rceil -1 ~=~ O(\log k)$ where each $E_j^{i}$ contains edges of  weight in range  $(\frac{2^{j-1}w_0}{\epsilon^{i}},\frac{2^{j}w_0}{\epsilon^{i+1}}]$. Such a partition can clearly be found in polynomial time.

We will find a spanner that preserves distances between the endpoints of edges in $E_j$ separately for each $j$. The final subset spanner will be union of at most $J$ such spanners, and thus the lightness bound is blown up by just a factor of $J ~= ~ O(\log \frac{1}{\epsilon})$. A nice property of the edge partitioning scheme is that edges in $E_j^{i+1}$ weight at least $\Omega(\frac{1}{\epsilon})$ times edges in $E_j^{i}$.

We now focus on constructing a spanner for edges in $E_j$ for a fixed $j$. Let $\ell_i = \frac{2^j}{\epsilon^{i+1}}w_0$ be the upper bound on the length of edges in $E_j^{i}$.  We refer to edges in $E_j^{i}$ as level-$i$ edges. We will construct a subset spanner iteratively by considering edges from level $0$ to level $I$.

\begin{definition} Let $S_i$ be the spanner constructed after level $i$. Initially, 
$S_0  = \kappa(\MST)$, and after level $I$, $S = S_I$.
\end{definition}

The construction of $S_i$ will depends on $S_{i-1}$ and $E_j^i$. Our final spanner $S$ has  stretch at most $(1 + s\epsilon)$ for a sufficiently large constant $s$ independent of $\epsilon$.

Let $\mathcal{O}$ be a sparse spanner oracle guaranteed by the assumption of Theorem~\ref{thm:reduction}. There are two major ideas in BCLW technique. The first idea is to construct a set of clusters, say $\mc_i$, for each level $i$. Each cluster will be \emph{a subgraph of $S_i$}. The fact that each cluster is a subgraph of $S_i$, instead of being a subgaph of $M_T$, is very important since we would repeatedly the routing argument in the stretch analysis. That is, we  route a shortest path between two terminals \emph{though clusters}  to obtain a short path of roughly the same length, and the new path would be in $S_I$ since clusters are subgraphs of $S_i$.

Clusters in level $0$ are constructed from subtrees of $\MST$.  Let $S_0 = \kappa(\MST)$.  Note that there is no level-$0$ edges in $E_j$ since edges in $E_j$ have length more than $\frac{w_0}{\epsilon}$.  To construct a spanner for level-$i$ edges for any $i \geq 1$, we use the set of clusters $\mc_{i-1}$ constructed in level $i-1$ as a guidance. Let $S_{i-1}$ be the subset spanner constructed before level $i$.

In our construction, for each cluster $C \in \mc_{i-1}$ that is incident to $\Omega(\frac{1}{\epsilon})$ level-$i$ edges, we will select one vertex and then call oracle $\mathcal{O}$ to construct a sparse spanner for the selected vertices. The sparsity of $\mathcal{O}$ guarantees that the output spanner has small weight. To ensure that terminal distances would not be blown up by much when re-routing the shortest paths through clusters in $\mc_{i-1}$, we maintain that:

\begin{tcolorbox}
(DC1) Each cluster in $\mc_{i-1}$ is a subgraph of $S_{i-1}$ and has  diameter at most $g\ell_{i-1}$ for some sufficiently big constant $g$ chosen later. 
\end{tcolorbox}

This is same (DC1) invariant in~\cite{BLW19}.  The intuition is that a level-$i$ edge $e$ has weight at least $\ell_i/2$ while the diameter of a cluster in $\mc_{i-1}$ is at most $g\ell_{i-1} ~=~ g\epsilon \ell_i ~ \leq~ 2g\epsilon w(e)$. Hence, we can re-route the shortest path between $e$'s endpoints through a level-$(i-1)$ cluster while the length of the path is increased by at most $2g\epsilon w(e)$, which is much smaller than $s\cdot\epsilon w(e)$ when $s$ is chosen sufficiently large. That is, the final stretch of $e$ is still $(1+s\cdot\epsilon)$.  By the same reason, we only need to preserve the distance between the endpoints of at most one  level-$i$ edge among all the level-$i$ edges that connect the same two level-$(i-1)$ clusters.

The second idea in BCLW technique is an amortized argument via credits to bound the weight of the output spanner.  The whole idea is to allocate some fixed amount of credits to $\MST$ edges and use these credits to buy all the spanner edges added during the construction of $S$. Suppose that the total allocated credit is $\ce w(\MST)$  for some parameter $\ce$.  If  $\ce w(\MST)$ credits are sufficient to buy all spanner edges, then $w(S) \leq \ce w(\MST)$. In what follows, we will elaborate the credit allocation scheme.

We first guarantee that every edge of $\MST$ has weight at most $w_0$ by subdividing every edge $e$ of weight more than $w_0$ into $\lceil  \frac{w(e)}{w_0}\rceil$ edges of weight at most $w_0$. We then allocate $\ce w_0$ credits to each new $\MST$ edge (now of weight at most $w_0$).  Observe that  the total allocated credit is $O(\ce)w(\MST)$ (see the proof in Appendix B in~\cite{BLW19}). Thus, $\ce$ would finally still be the asymptotic upper bound on the weight of the spanner.   One minor issue concerning the subdivision of $\MST$ edges is that subdividing vertices  are not in $G$, so they cannot be involved in any oracle call; in fact, during our construction, no oracle call would involve subdividing vertices. The purpose of the subdivision is to guarantee that  (a) $\MST$ edges are significantly shortest than diameter of level-$i$ clusters for any $i \geq 1$ and (b) level-$i$ clusters have roughly the same amount of credits. These two properties would significantly simplify the proof of Credit Lemma (Lemma~\ref{lm:credit-informal}) which is central to bounding the spanner weight.

After allocating credits to $\MST$ edges, we start the construction of level-$0$ clusters. Credits of $\MST$ will be used to build credits for these clusters. Then, credits of level-$0$ clusters will be used to build credits for level-$1$ clusters, and so on. With the credits built from lower level clustering, level-$i$ clusters will pay for edges in $E(S_{i+1})\setminus E(S_{i})$. (When $i = 0$, level-$0$ clusters will pay for $E(S_{1})\setminus E(S_{0})$.) This will guarantee that when we finish the construction in level $I$, all edges of $S_I$ are already paid for.   Since level-$i$ edges are longer, level-$i$ clusters must have more credits to pay for their spanners. To this end, we guarantee that:

\begin{tcolorbox}
(DC2) Each cluster in $\mc_{i-1}$ of diameter $d$ has at least $\ce\max(d,\ell_{i-1}/2)$ credits. 
\end{tcolorbox}

This is the same (DC2) invariant in~\cite{BLW19}. The seemingly artificial credit lower bound  $\ce\ell_{i-1}/2$ in (DC2) is because we have no lower bound on the diameter of clusters; invariant (DC1) only provides an upper bound. We allow low diameter clusters as long as they have enough credits to pay for the weight of the spanners from the spanner oracle. Since we only allocate credits once, we cannot use all the credit of clusters in $\mc_{i-1}$ to pay for the spanner edges added in level $i$. That is, we need to use credits of  $\mc_{i-1}$ to allocate credits to $\mc_{i}$ to maintain invariant (DC2) for level-$i$ clusters. The goal is to show that it is possible to construct level-$i$ clusters in a way that after maintaining invariant (DC2) for level $i$,  level-$(i-1)$ clusters still have significant leftover credits to pay for the spanner edges added in level $i$. In fact, showing such a construction is the heart of all arguments following the iterative clustering framework~\cite{CW16,BLW17,BLW19,LS19}, including BCLW technique.

We now go into details of the construction. We first greedily break the $\MST$ into sub-trees of diameter at least $\ell_0$ and at most $6\ell_0$. Note that each $\MST$ edge has length at most $w_0 < \ell_0$. For each subtree $T$ broken from $\MST$, we define $\kappa(T)$ to be a level-$0$ cluster.  Since $S_0 = \kappa(\MST)$, level-$0$ clusters are subgraphs of $S_0$. Therefore, by choosing $g\geq 6$,  invariant (DC1) is maintained for level-$0$ clusters. We now show invariant (DC2).  Since $\dm(T) \geq \dm(\kappa(T))$, we can use the credit of edges on the diameter path of $T$ to ensure that $\kappa(T)$ has at least $\ce \max(\dm(\kappa(T)),\ell_0/2)$ credits.  This is possible because each edge of $T$ has a credit at least $\ce$ times its length, and $\dm(T)\geq \ell_0 > \ell_0/2$. Since  $E_j^0 = \emptyset$, we do not need to pay for any level-$0$ edge. 

We now construct level-$i$ clusters and spanners for level-$i$ edges, assuming that two invariants (DC1) and (DC2) hold for level $i-1$. Recall that $S_{i-1}$ is the spanner constructed before level $i$.  Following the notation of~\cite{BLW19}, we call level-$(i-1)$ clusters \emph{$\epsilon$-clusters}. A level-$i$ edge is said to \emph{connect}  two $\epsilon$-clusters if its endpoints are contained in the $\epsilon$-clusters.  Let $\mathcal{K}$ be the \emph{cluster graph} where each node of $\mathcal{K}$ corresponds to an $\epsilon$-cluster in $\mc_{i-1}$ and each edge of $\mathcal{K}$ corresponds to a level-$i$ edge that connects the two corresponding  $\epsilon$-clusters. 

Note that there could be many level-$i$ edges that connect the same two $\epsilon$-clusters, but we only keep the least weighted edge in $\mathcal{K}$. Also note that  there would be no level-$i$ edge that have both endpoints in the same $\epsilon$-cluster since the weight of each level-$i$ edge is $\ell_i/2 ~=~ \frac{\ell_{i-1}}{2\epsilon} > g\ell_{i-1}$ when $\epsilon$ is sufficiently big, while $\epsilon$-clusters have diameter at most $g\ell_{i-1}$. We further remove from $\mathcal{K}$ any edge whose shortest path in $S_{i-1}$ is at most $(1+(6g+1)\epsilon)$ its weight since the distance between its endpoints is already preserved in $S_{i-1}$ (by setting $s\geq 6g+1$). Constant $(6g+1)$ comes from the analysis of a special case in our argument that will appear later.

In~\cite{BLW19}, the packing property of doubling metrics was used to argue that $\mathcal{K}$ has bounded degree (see Lemma~3.1 in~\cite{BLW19}), so they can afford to buy every edge of $\mathcal{K}$ to the spanner using $\epsilon$-clusters' credits. Problems considered in prior work~\cite{BLW17,BLW19,LS19} enjoy the same degree boundedness or average-degree boundedness. In our setting, we do not have any constraint on the degree of $\mathcal{K}$; it could be a complete graph.  This is when sparse spanner oracles come into play. We gather all high degree nodes of $\mathcal{K}$ and call oracle $\mathcal{O}$ to construct a sparse spanner for these nodes. The weak sparsity of $\mathcal{O}$ guarantees that on average, each  high degree node of $\mathcal{K}$ only pays for the weight equal to the total weight of a constant number of edges of $\mathcal{K}$. Intuitively, sparse spanner oracles allow us to ``reduce the degree''  of $\mathcal{K}$ to constant.  The rest of the argument can therefore be adapted directly from prior work~\cite{BLW17,BLW19,LS19}. In the following section, we give a formal argument. 

\subsubsection{Spanner construction}

 To avoid confusion, we refer to vertices of $\mathcal{K}$ as \emph{nodes}.  For each node $\mbx \in \mathcal{K}$, we use $\mc_{i-1}(\mbx)$ to denote the $\epsilon$-cluster corresponding to $\mbx$. (The bold font  will be used to denote $\mathcal{K}$'s nodes.)
 
 A node $\mbc $ has \emph{high degree} if its degree in $\mathcal{K}$ is at least $\frac{2g}{\epsilon} + 1$. Otherwise, we say that $\mbc$ has \emph{low degree}. Let $V_{low}$ and $V_{high}$ be the set of low and high degree nodes in $\mathcal{K}$, respectively. We construct $S_i$ from $S_{i-1}$ in two steps:
 
 \begin{itemize}
 \item[]{\textbf{(Step 1)}} For each node $\mbc \in V_{low}$ and each edge $e$ incident to $\mbc$ in $\mathcal{K}$, we add path $\kappa(e)$ to $S_{i-1}$. 
 \item[]{\textbf{(Step 2)}} For each node $\mbc \in V_{high}$, we choose a vertex of $M_T$ (a terminal) in  $\mc_{i-1}(\mbc)$. Let $T'$ be the set of chosen vertices. We add to $S_{i-1}$ the spanner $\mathcal{O}(T', 2\ell_i)$. 
 \end{itemize}
 Let the resulting spanner be $S_i$. We now argue that for every level-$i$ edge $e \in E_{j}^i$, there is a path in $S_i$ between its endpoints of length at most $(1+s\epsilon)w(e)$ when $s$ is sufficiently large.

\begin{claim}\label{clm:stretch} For any edge $e\in E_{j}^i$, there is a shortest path in $S_j$ between $e$'s endpoints of length at most $\left(1+(16g+1)\epsilon\right) w(e)$.
\end{claim}
\begin{proof}

There are three possibilities: (a) $e \not\in \mathcal{K}$ and there is a shorter level-$i$ edges connecting the two $\epsilon$-clusters that $e$ connects, (b) $e$ was initially in $ \mathcal{K}$ but then removed from $\mathcal{K}$ because the shortest path between its endpoints in $S_{i-1}$ has length at most $(1+(6g+1)\epsilon)w(e)$ and  (c) $e \in \mathcal{K}$ and it is not removed from $\mathcal{K}$. Case (b) directly implies the claim. Thus, we only need to consider two  other cases.

\noindent\textbf{Case 1: $e \in \mathcal{K}$}. If $e$ is incident to a node in $V_{low}$, then the shortest path between $e$'s endpoints in $G$ is added to $S_{i-1}$ in Step 1; the claim holds. Suppose $e$'s endpoints, say $\mbx$ and $\mby$, are in $V_{high}$.  Let $x$ and $y$ be the two chosen terminals in $T'$ of $\mc_{i-1}(\mbx)$ and $\mc_{i-1}(\mby)$, respectively. By invariant (DC1) and triangle inequality, we have:
\begin{equation*}
\begin{split}
d_G(x,y) &~\leq~ w(e) + \dm(\mc_{i-1}(\mbx)) + \dm(\mc_{i-1}(\mby))\\
 &~\leq~  w(e) + 2g\ell_{i-1} ~\leq~ \ell_i + 2g\epsilon\ell_i ~<~ 2\ell_i
\end{split}
\end{equation*}
when $\epsilon$ is sufficiently smaller than $1/g$. Furthermore,
\begin{equation*}
\begin{split}
d_G(x,y) &~\geq~ w(e) - \dm(\mc_{i-1}(\mbx)) - \dm(\mc_{i-1}(\mby))\\
 &~\geq~  w(e) - 2g\ell_{i-1} ~\geq~ \ell_i/2 - 2g\epsilon\ell_i ~\geq ~ \ell_i/4
\end{split}
\end{equation*}
when $\epsilon$ is sufficiently smaller than $1/g$.  Thus, there is a shortest path $Q_{x,y}$ of length at most $(1+\epsilon)d_G(x,y)$ in $\mathcal{O}(T', 2\ell_i)$ by the definition of spanner oracles. 

Let $x'$ and $y'$ be $e$'s endpoints in $\mc_{i-1}(\mbx)$ and $\mc_{i-1}(\mby)$, respectively. Let $P$ be the path between $x'$ and $y'$ composed of (a) a shortest path from $x'$ to $x$ in $\mc_{i-1}(\mbx)$,  (b) path $Q_{x,y}$ and (c) a shortest path from $y$ to $y'$ in $\mc_{i-1}(\mby)$. Observe that $P$ is a path in $S_i$ since all of its constituent subpaths are in $S_i$. Thus, it holds that:
\begin{equation*}
\begin{split}
w(P) &\leq 2g\ell_{i-1} + w(Q_{x,y})~\leq~2g\ell_{i-1} + (1+\epsilon)d_G(x,y)\\
& \leq 2g\ell_{i-1} + (1+\epsilon)(w(e)+2g\ell_{i-1})\\
&= 6g\epsilon \ell_i + (1+\epsilon)w(e) \leq (1+(12g+1)\epsilon)w(e)  
\end{split}
\end{equation*}
since $w(e)\geq \ell_i/2$. Thus, the claim holds.

\noindent\textbf{Case 2: $e \not\in \mathcal{K}$}. By construction, there is another edge $e'$ that has $w(e') \leq w(e)$ and connects the same two nodes, say $\mbx$ and $\mby$. Let $x,y$ ($x',y'$) be $e$'s endpoints ($e'$'s endpoints) in $\mc_{i-1}(\mbx), \mc_{i-1}(\mby)$, respectively. Let $P$ be the path between $x'$ and $y'$ composed of (a) a shortest path from $x$ to $x'$ in $\mc_{i-1}(\mbx)$,  (b) shortest path $Q_{x',y'}$ between $x'$ and $y'$ in $S_i$ and (c) a shortest path from $y'$ to $y$ in $\mc_{i-1}(\mby)$. Observe that $P$ is a path in $S_i$ since all of its constituent subpaths are in $S_i$. 

If $e'$ was removed from $\mathcal{K}$ after it was added to $\mathcal{K}$ initially, then $w(Q_{x',y'}) \leq (1+(6g+1)\epsilon) w(e')$. Otherwise, by Case 1, $w(Q_{x',y'}) \leq (1+(12g+1)\epsilon) w(e')$. Both cases imply that:
\begin{equation*}
\begin{split}
w(P) &\leq 2g\ell_{i-1} + w(Q_{x',y'})~\leq~2g\ell_{i-1} + (1+(12g+1)\epsilon)w(e')\\
&= 2g\epsilon \ell_i +  (1+(12g+1)\epsilon)w(e) \leq (1+(16g+1)\epsilon)w(e)  
\end{split}
\end{equation*}
since $w(e)\geq \ell_i/2$.
\end{proof}

\subsubsection{Bounding the spanner weight} \label{sec:weight-bound}

Let $S_I$ be the final spanner after the maximum level $I$. Claim~\ref{clm:stretch} guarantees that $S_I$ will preserve distances between endpoints of edges in the set $E_j$ for a fixed $j$. The final spanner $S$ is the union of all such spanners for $j = 1,2,\ldots,J$. Therefore, for every two terminals $x\not=y \in T$, their distance in $S$ is at most $(1+s\epsilon)d_G(x,y)$ when $s = 16g+1$.

Since there are at most $J = O(\log \frac{1}{\epsilon})$ different sets $E_j$, the weight of the final spanner would be at most  $O(\log \frac{1}{\epsilon})$ times the worst case bound on the weight of $S_I$ for a fixed $j$. To bound the weight of $S_I$, we need to study the clustering procedure in details. The idea is to choose $\ce$ sufficiently large so that the total allocated credit (of value $O(\ce)w(\MST)$) can buy all the spanner edges in $S_I$.  That would imply $w(S_I)  = O(\ce)w(\MST)$.

If a node $\mbx$ is grouped in to a level-$i$ cluster $C$, we say $\mbx$ is a \emph{child} of $C$, and $C$ is $\mbx$'s \emph{parent}.   The goal of the clustering procedure is to guarantee that:

\begin{lemma}[Credit Lemma]\label{lm:credit-informal} There is a way to group $\mathcal{K}$'s nodes into level-$i$ clusters such that after each node has used its credits to guarantee invariant (DC2) for its parent, it still has $\Omega(\epsilon \ce \ell_{i-1})$ leftover credits, except when:
\begin{itemize}[nolistsep,noitemsep]
\item[(i)] it has low degree, all of its neighbors also have low degree with non-zero leftover credits or
\item[(ii)] $\mathcal{K}$ only has $O(\frac{1}{\epsilon^2})$ edges.
\end{itemize}
\end{lemma}

Note that by invariant (DC2), each $\epsilon$-cluster has at least $\ce \ell_{i-1}/2$ credits. Thus, Lemma~\ref{lm:credit-informal} essentially says that roughly $\Omega(\epsilon)$ fraction of the credit of each node is leftover.  To focus on the main idea, let us put aside special cases (i) and (ii) in Lemma~\ref{lm:credit-informal}; we will come back to deal with them later.  For the moment,  we assume that each node $\mathcal{K}$ has at least $\Omega(\epsilon \ce \ell_{i-1})$ credits left.

There are two cases: if $\mbx$ has low degree (it is incident to at most $\frac{2g}{\epsilon} + 1 ~=~ O(\frac{g}{\epsilon})$ level-$i$ edges.), it can afford to buy all of these edges (and hence all the shortest paths corresponding to the edges added to $S_{i-1}$ in Step 2) when:
\begin{equation}\label{eq:ce-value-low}
\ce ~=~\Omega(g\epsilon^{-3})~ = ~ \Omega(\epsilon^{-3}),
\end{equation}
since the total weight $\mbx$ needs to pay for  is $O(g\epsilon^{-1}\ell_i)~=~O(g\epsilon^{-2} \ell_{i-1})$.  Thus, handling low degree vertices is an easy case; the hard case is to handle high degree vertices.

Recall that in Step 2, we call  the sparse spanner oracle on the terminals selected from high degree nodes.  Observe that there are $|V_{high}|$ such nodes; $|T'| = |V_{high}|$. The sparsity of the oracle guarantee that each high node of $\mathcal{K}$ must pay for at most:

\begin{equation}\label{eq:avg-high-deg-pay}
\frac{w(\mathcal{O}(T',2\ell_i))}{|V_{high}|} ~\leq~\frac{2 \wsp_{\mathcal{O}} |T'| \ell_i }{|V_{high}|}~=~  O(\wsp_{\mathcal{O}})\ell_i
\end{equation}
which is equivalent to the weight of at most $O(\wsp_{\mathcal{O}})$ level-$i$ edges since each level-$i$ edge has weight in range $(\ell_i/2,\ell_i]$.

By Lemma~\ref{lm:credit-informal}, each high degree node has at least $\Omega(\epsilon\ce\ell_{i-1})~=~\Omega(\epsilon^{2}\ce\ell_i)$ leftover credits.  By Equation~\ref{eq:avg-high-deg-pay}, this amount of credit is sufficient to pay for the weight of $\mathcal{O}(T',2\ell_i)$ when:
\begin{equation}\label{eq:ce-value-high}
\ce = \Omega\left(\wsp_{\mathcal{O}}\epsilon^{-2}\right)
\end{equation}
 By Equation~\ref{eq:ce-value-low} and Equation~\ref{eq:ce-value-high}, choosing $\ce = \Theta\left(\max(\wsp_{\mathcal{O}}\epsilon^{-2}, \epsilon^{-3})\right)$ suffices.  Inductively, all the spanner edges added at a level will be paid for at that level. Thus, the total weight of the spanner for $E_j$ is at most $O(\ce)w(\MST) ~=~ O\left(\max(\wsp_{\mathcal{O}}\epsilon^{-2}, \epsilon^{-3})\right)w(\st)$.

We now handle two special cases (i) and (ii) in Lemma~\ref{lm:credit-informal}. Let $\mbc$ be a node in case (i). Since all neighbors of $\mbc$  have low degree and have leftover credits, they have already paid for their incident level-$i$ edges. Thus, $\mbc$ do not need to pay for any incident level-$i$ edge and hence its credits can be taken entirely by its parent to maintain (DC2).

For case (ii), we simply do not pay for edges of $\mathcal{K}$ at level $i$ using clusters' credits. Instead, we pay for  these edges altogether after we finish the construction at level $I$. Recall that level-$i$ edges have weight at most $\ell_i$. Thus, summing over all levels, the total weight we pay is at most:

\begin{equation}
O(\epsilon^{-2})\sum_{i=1}^I O(\ell_i)~=~O(\epsilon^{-2})\ell_{\max}\sum_{i=0}^{\infty} \epsilon^i~=~O\left(\epsilon^{-2}w(\MST)\right)
\end{equation}
where $\ell_{max}$ is the maximum length of any edge in $M_T$, which cannot exceed $w(\MST)$. Thus, all these edges only contribute $O(\epsilon^{-2})$ additively to the final lightness bound.

We now focus on proving Lemma~\ref{lm:credit-informal}. The ideas sketched here are a combination of the ideas from two papers of Borradaile, Le and Wulff-Nilsen~\cite{BLW19,BLW17}. There are  some minor details specific to our presentation of the proof, mostly involving the calculation of diameter upper bounds  because our clusters are constructed via $\mathcal{K}$, whose edges do not belong to $S_i$. In retrospect, both papers of Borradaile, Le and Wulff-Nilsen~\cite{BLW19,BLW17} implicitly proved Lemma~\ref{lm:credit-informal} but technical details specific to their problems obfuscate a clean statement. By phrasing their techniques in a single lemma, we believe that it would be of independent interest.

Herein, we only present the high level ideas of the proof of Lemma~\ref{lm:credit-informal} and focus on revealing the intuition behind the special cases. For readers who are interested in seeing all technical details, we provide a complete proof in Appendix~\ref{app:cluster-const}.

 The cluster construction is divided into four phases. It can be seen by carefully following the construction that the diameter of level-$i$ clusters are bounded by $g\ell_{i}$ for some sufficiently large $g$ and small $\epsilon$. Hence, for the rest of the discussion, we assume that invariant (DC1) is maintained correctly. We now focus on maintaining invariant (DC2) and guaranteeing the credit lower bound as stated in Lemma~\ref{lm:credit-informal}. The following observation allows us to simplify much of the proof.
 
 \begin{observation}\label{obs:rich} If a level-$i$ cluster has at least $\frac{2g}{\epsilon}+1$ children, it can maintain invariant (DC2) by taking the credit of its children, while each child  still has at least $\Omega(\epsilon \ce \ell_{i-1})$ leftover credits. 
\end{observation}
\begin{proof}
Take the credit of any $\frac{2g}{\epsilon}$ children of $C$ to maintain invariant (DC2). This suffices because by invariant (DC2) for level $i-1$, $\frac{2g}{\epsilon}$ children has at least
\[\frac{2g}{\epsilon} \ce \ell_{i-1}/2~=~g\ce \ell_{i} ~\geq~\ce\max(\dm(C),\ell_i/2)\] credits since $g > 1$ and $\dm(C) \leq g\ell_{i}$ by invariant (DC1). We then can take the credit of any other node in $C$ to redistribute to the children whose credits were taken by $C$. The redistribution guarantees that each node has at least $\Omega(\ce\epsilon \ell_{i-1}/g) = \Omega(\ce\epsilon \ell_{i-1})$ credits. Remaining children of $C$ can keep their own credits of amount $\Omega(\ce \ell_{i-1})$ by invariant (DC2) as leftover. 
\end{proof}

Most of the technical bulk is devoted to show that for every level-$i$ cluster formed in the first three phases, say $C$, after maintaining invariant (DC2) by taking their children credits, has \emph{at least one child} whose credits remain intact. This suffices to imply Lemma~\ref{lm:credit-informal} since by invariant (DC2) for level-$(i-1)$, the child has at least $\Omega(\ce\ell_{i-1})$ credits and by redistributing this credit to all the children of $C$, each has at least $\Omega(\ce\epsilon \ell_{i-1})$ leftover credits. (Here we assume that $\mathcal{C}$ has at most $\frac{2g}{\epsilon}$ children by Observation~\ref{obs:rich}.)

Now we go through intuition of each phase of the cluster construction.

\paragraph{Phase 1} In this phase, every constructed cluster contains a high degree node and all if its neighbors. This guarantees that Phase 1 clusters have at least $\frac{2g}{\epsilon}+1$ children each.  Thus, by Observation~\ref{obs:rich}, Lemma~\ref{lm:credit-informal} holds for the nodes involved in this phase.  Furthermore, the construction guarantees that any node of low degree adjacent to a high degree node is also included in a Phase 1 cluster. This implies that when case (i) in Lemma~\ref{lm:credit-informal} happens to a node, all the neighbors  have low degree.

In the following phases, the construction is based on a \emph{cluster tree} $\mathcal{T}$, whose vertices are $\epsilon$-clusters and edges are $\MST$ edges connecting the $\epsilon$-clusters. A crucial property of $\mathcal{T}$ is that the credit of edges of $\mathcal{T}$ has not been taken by clusters in lower levels.

\paragraph{Phase 2} In this phase, each cluster is a sub-tree (of $\epsilon$-clusters), say $C$, of $\mathcal{T}$. $C$ is guaranteed by the construction to have at least one \emph{branching node}, i.e, a node with at least three neighbors in $C$.  Since each $\epsilon$-cluster has at least $\ce d$ credits where $d$ is its diameter by invariant (DC2), and each $\MST$ edge has a credit at least $\ce$ times its length, the credit of $\epsilon$-clusters and $\MST$ edges of the diameter path, say $\mathcal{D}$, of $\mathcal{T}$ is sufficient to maintain invariant (DC2) for $C$. Since $C$ has a branching node, say $\mbx$, at least one of the neighbors of $\mbx$, say $\mby$, is not in $\mathcal{D}$.  Thus, $\mby$'s credit will not be taken by $C$, and this is the node we are looking for. As discussed above, we try to show that there is at least one node in each cluster whose credit is not taken by its parent.

\paragraph{Phase 3} There are three subcases in this phase (see Figure~\ref{fig:P3N} in Appendix~\ref{app:cluster-const}) where in the first two subcases, a level-$i$ cluster consists of two subpaths of $\mathcal{T}$ connected by a level-$i$ edge. The two paths have total diameter roughly $4\ell_i$, which is equivalent to having at least $4\ce\ell_i$ credits.  A remarkable property of the cluster is that it has diameter at most $3\ell_i$, thus only $3\ce\ell_i$ credits will be taken by the cluster, leaving at least $\ce\ell_i~=~\ce\frac{\ell_{i-1}}{\epsilon}$ credits as leftover, which is more than the amount of credit possessed by any $\epsilon$-cluster when $\epsilon$ is  smaller than $1/g$.

 The hardest case is  the third case, where a cluster, say $C$, is a subpath, say $\mathcal{P}$, of $\mathcal{T}$ and there is a level-$i$ edge $e$ connecting two verttices, say $\mbx$ and $\mby$,  of the subpath. Recall that when we construct $\mathcal{K}$, we remove from $\mathcal{K}$ every edge whose stretch in $S_{i-1}$ is at least $(1+ (6g+1)\epsilon)$. Thus, the presence of $e$ in $C$ implies that the weight of the subpath $\mathcal{P}[\mbx, \mby]$ is at least $(1+ (6g+1)\epsilon)$ times longer than $w(e)$. Hence, we can deduce that the diameter path of $C$ must go through $e$, if it goes through both $\mbx$ and $\mby$, and that the credit of the path   $\mathcal{P}(\mbx, \mby)$ must be at least $\ce w(e) + \ce g\ell_{i-1}$. If we assign $\ce w(e)$  credits to $e$ and allow $C$ to take all the credit of the edges and vertices in the diameter path, we still have at least $\ce g \ell_{i-1}$ credits left, which is more than the amount of credit owned by an $\epsilon$-cluster.

 \paragraph{Phase 4} In this phase, clusters are subpaths of $\mathcal{T}$, thus they can maintain invariant (DC2) by taking all credits of the children and $\MST$ edges on the paths. However, there could be no leftover credits. There are two ideas to resolve this issue: (a) show that all of the level-$i$ edges incident to a Phase 4 cluster are also incident to clusters formed in previous Phases or (b) let a Phase 4 cluster steal the leftover credits of  the nodes in the nearest cluster formed in the first three phases. If idea (a) can be realized, then the nodes in the Phase 4 cluster fall into case (i) of Lemma~\ref{lm:credit-informal}.  To realize idea (b), we argue that the leftover credit of each node is stolen at most once by nodes in Phase 4 clusters. Furthermore, we show that, each node shares its leftover credits, of value at least $\Omega(\epsilon \ce \ell_{i-1})$, to at most $O(1)$ other nodes. Thus, each of them gets at least $\Omega(\epsilon \ce\ell_{i-1})$ credits as desired. 
 
An exception  is when there are no clusters formed in Phase 1, 2 or 3. Thus,we cannot implement both ideas (a) and (b). In this case, we show that the tree $\mathcal{T}$ is highly structural: it is a path with level-$i$ edges connecting its affices only. The special structure of $\mathcal{T}$ allows us to show that $\mathcal{K}$ only has $O(\frac{1}{\epsilon^2})$ edges. This falls into case (ii) in Lemma~\ref{lm:credit-informal}, thereby completing the proof.  

\section{Strongly sparse spanner oracles for metric spaces}\label{sec:spare-oracle-Euc}

In this section, we show that various metric spaces have strongly sparse spanner oracles. 

\subsection{Euclidean metrics}

\begin{lemma}\label{lm:oracle-euc} Any point set in the Euclidean space of dimension $d$ has a spanner oracle with strong sparsity $O(\epsilon^{1-d})$.
\end{lemma}
\begin{proof}
It is well known that any point set in Euclidean space have a spanner with sparsity $O(\epsilon^{1-d})$ that can be constructed by  $\Theta$-graph~\cite{RS91,Keil88,KG92,ADDJS93}, Yao graph~\cite{Yao82}, or greedy algorithms~\cite{CDNS92,ADDJS93}. Suppose that $T$ and $\ell$ are given as an input to the oracle. We call a sparse spanner construction on $T$ to obtain a spanner $S$. We then remove every edge of length at least $2\ell$ from $S$ to obtain $S'$ and return $S'$ as the output of the oracle. 

Observe that $|E(S')| \leq |E(S)| = O(\epsilon^{1-d})|T|$. Thus, $S'$ has strong sparsity $O(\epsilon^{1-d})$. The argument for distance preserving property  is similar to that of Lemma~\ref{lm:minor-to-oracle}.
\end{proof}

Lemma~\ref{lm:oracle-euc} and Theorem~\ref{thm:reduction} implies that the Euclidean metric of constant dimension $d$ has a spanner of lightness $\tilde{O}(\epsilon^{-(d+1)})$ when $d\geq 1$. This bound is weaker than the optimal bound $\tilde{O}(\epsilon^{-d})$ by just a factor of $\frac{1}{\epsilon}$ obtained recently (with a rather complicated proof) by Solomon and this author~\cite{LS19}.

\subsection{Doubling metrics}
 
 Before stating our results, let us remind the reader of the formal definition of doubling metrics. Given a metric space $(X,\delta)$ where $\delta$ is the distance function, the doubling dimension of $(X,\delta)$ is the smallest value $d$ such that every ball $B$ in the metric space can be covered by at most $2^d$ balls of half the radius of $B$. This notion was introduced by Gupta, Krauthgamer and Lee~\cite{GKL03}, and was inspired by  Assouad~\cite{Assouad83}.

 It is well known that doubling metrics of constant dimension $d$ have a spanner with sparsity $\tilde{O}(\epsilon^{-d})$~\cite{CGMZ16,Smid09}. Since a sub-metric of a doubling metric of dimension $d$ has dimension at most $2d$. If we apply the same argument in the proof of Lemma~\ref{lm:oracle-euc}, we would obtain an oracle with sparsity $\tilde{O}(\epsilon^{-2d})$. To obtain the sparsity bound $\tilde{O}(\epsilon^{-d})$, we use a different argument, which crucially exploits the fact that we only preserve distances in range $[\frac{\ell}{8},\ell]$.

 Let $Y\subseteq X$ be a subset of points in a doubling metric $(X,\delta)$. A  set $N\subseteq Y$ is an \emph{$r$-net} of $Y$ if (a) for every point $y \in Y$, there is a point $p \in N$ such that $\delta(p,y) \leq r$ and (b) for every two distinct points $p,q \in N$, $\delta(p,q) > r$. An $r$-net of a given subset can be constructed greedily in polynomial time. We will use the following well-known packing property of doubling metrics~\cite{GKL03}.
 
 \begin{lemma}[Packing property]\label{lm:packing} Let $(X,\delta)$ be a doubling metric of dimension $d$. If a set of point $Y\subseteq X$  is contained in a ball of radius $R$ and $\delta(x,y) > r$ for every $x\not=y \in Y$, then $|Y| \leq (\frac{4R}{r})^d$.
\end{lemma}

\begin{lemma}\label{lm:oracle-dd}A metric of constant doubling dimension $d$ has a spanner oracle with strong sparsity $O(\epsilon^{-d})$ when $\epsilon< 1$.
\end{lemma} 
\begin{proof}
Let $T \subseteq X$ and $\ell$ be inputs given to the oracle.  Let $N$ be an $\frac{\epsilon \ell}{96}$-net of $T$. We construct a set of edges $E_S$ of the spanner in two steps. (Step 1) for every two distinct points $p\not= q \in N$, we add an edge between $p,q$ if $\ell/16 \leq \delta(p,q) \leq 2\ell$ to $E_S$. (Step 2) for each point $t \in T\setminus N$, we add an edge from $t$ to a nearest point in $N$ to $E_S$. We finally return the graph $S(T,E_S)$ as the output.

To bound the sparsity of $S(T,E_S)$, we observe that for each point $p \in N$, by Lemma~\ref{lm:packing}, the number edges incident to $p$ added to $E_S$ in Step 1 is bounded by:
\begin{equation}\label{eq:deg-p-dd}
\left( \frac{2\ell}{\epsilon \ell/96} \right)^d = O(\epsilon^{-d})
\end{equation}
Thus, the size of $E_S$ after Step 1 is $O(\epsilon^{-d})|N|$. In step 2, we add one edge per point in $T\setminus N$. Thus, the size of $E_S$ after Step 2 is:
\begin{equation*}
|E_S| \leq O(\epsilon^{-d})|N| + |T\setminus N| = O(\epsilon^{-d}) |T|
\end{equation*}  
Therefore, the  strong sparsity of  the oracle is $O(\epsilon^{-d})$.

It remains to bound the stretch of the oracle. Let $p,q$ be any two distinct points in $T$ where $\ell/8 \leq \delta(p,q) \leq \ell$. Let $x$ and $y$ be two points in $N$  closest to $p$ and $q$, respectively. By triangle inequality, $\delta(x,y) \geq d(p,q) - 2\frac{\epsilon\ell}{96}~\geq~\ell/8 - \frac{\epsilon \ell}{48}~\geq~\ell/16$ and $\delta(x,y) ~\leq ~ d(p,q) + 2\frac{\epsilon\ell}{96}~\leq~2\ell$ when $\epsilon < 1$. Thus, there is an edge between $x$ and $y$ in $E_S$ by the construction in Step 1. Also by the triangle inequality, the stretch of the shortest path between $p$ and $q$ in $S$ is at most:
\begin{equation}
\frac{\delta(x,y) + 2\frac{\epsilon \ell}{96}}{\delta(x,y) - 2\frac{\epsilon \ell}{96}} \leq \frac{\frac{\ell}{16}+ \frac{\epsilon \ell}{48}}{\frac{\ell}{16} - \frac{\epsilon \ell}{48}} = \frac{1 + \epsilon/3}{1-\epsilon/3} \leq 1 +\epsilon
\end{equation} 
when $\epsilon < 1$. The first inequality is due to $\delta(x,y) \geq \ell/16$. 
\end{proof}

Since strong sparsity implies weak sparsity (Equation~\ref{eq:wsp-ssp}), Theorem~\ref{thm:dd-metric} follows directly from  Lemma~\ref{lm:oracle-dd} and Theorem~\ref{thm:reduction}.

 \subsection{Metrics of bounded correlation dimension}
 
Let $(X,\delta)$ be a metric space.  Let $B(x,r) = \{y \in X| \delta(x,y) \leq r\}$ be the ball of radius $r$ centered at $x \in X$. Given a subset $Y \subseteq X$, let $B_{Y}(x,r) = B(x,r)\cap Y$.  A subset $N\subseteq X$ is a \emph{net} of $X$ if it is an $\zeta$-net for some $\zeta > 0$. The \emph{correlation dimension} of a metric space $(X,\delta)$ is the smallest $d$ such that:
\begin{equation}
\sum_{x\in N} |B_N(x,2r)| \leq 2^d\sum_{x\in N}|B_N(x,r)|
\end{equation}
  for any net $N \subseteq X$ and any positive number $r$.
  
Correlation dimension was introduced by Chan and Gupta~\cite{CG12} to capture global growth rate of a metric, as opposed to doubling dimension which captures local growth rate.  Chan and Gupta showed that if $(X,\delta)$ has doubling dimension $k$, then it has  correlation dimension at most $9k$ (Theorem 1.1 in~\cite{CG12}). Intuitively, this is because slow local growth implies slow global growth. However, the converse statement does not hold. 
 
Unlike Euclidean or doubling metrics, metrics of bounded correlation dimension are not closed under taking sub-metrics; a sub-metric of a metric with bounded correlation dimension can have arbitrarily large dimension (see the discussion on this property in the paragraph below Theorem 1.1 in the paper of Chan and Gupta~\cite{CG12}). Remarkably, Chan and Gupta showed that metrics of bounded correlation dimension still have $(1+\epsilon)$-spanners with sublinear sparsity.

 \begin{theorem}[Theorem 1.4 in~\cite{CG12}] An $n$-point metric of constant correlation dimension $d$ has a spanner with sparsity $\epsilon^{-O(d)} \sqrt{n}$.
 \end{theorem}
 
However, it is unclear whether a (subset) spanner with sublinear lightness exists. Traditional techniques~\cite{CW16,BLW19,BLW17,LS19,FS16} rely on the closure of the input metric under taking subgaphs or sub-metrics of the input. By looking at the problem through the lens of sparse spanner oracles, we can show that light (subset) spanners exist.

By Theorem~\ref{thm:reduction}, it suffices to construct a spanner oracle with weak sparsity $O(\epsilon^{-(d/2 + 3)}\sqrt{n})$. Indeed, a strongly sparse spanner oracle with the same sparsity bound exists. The construction is similar to the construction for doubling metrics in Lemma~\ref{lm:oracle-dd}. It relies on the following packing property.

\begin{lemma}[Lemma 2.2 in~\cite{CG12}]\label{lm:packing-cd} Given a metric $(X,\delta)$ of constant correlation dimension $d$. Suppose that $N$ is an $r$-net of $X$ and $Y\subseteq N$ is contained in a ball of radius at most $R$, then:
\begin{equation*}
|Y| \leq \left(\frac{4R}{r}\right)^{d/2} \sqrt{|N|}
\end{equation*}
\end{lemma}

\noindent We are now ready to construct a strongly sparse spanner oracle.

\begin{lemma}\label{lm:oracle-cd}Any $n$-point metric of constant correlation dimension $d$ has a spanner oracle with strong sparsity $O(\epsilon^{-d/2})\sqrt{n})$ when $\epsilon< 1$.
\end{lemma} 
\begin{proof}
Let $T \subseteq X$ and $\ell$ be inputs given to the oracle.  Let $N$ be a $\frac{\epsilon \ell}{96}$-net of $T$. We also construct a set of edges $E$ of the spanner in two steps. (Step 1) add an edge between $p,q$ if $\ell/16 \leq \delta(p,q) \leq 2\ell$ to $E$ for  every two distinct points $p\not= q \in N$. (Step 2) add an edge from $t$ to a nearest point in $N$ to $E$ for each point $t \in T\setminus N$. We then return the graph $S(T,E)$ as the output.  
 
The proof that $S(T,E)$ has stretch $(1+\epsilon)$ is exactly the same as the proof in Lemma~\ref{lm:oracle-dd} for doubling metric case. To bound the strong sparsity, we also use a very similar proof. For each point $p \in N$, similar to Equation~\ref{eq:deg-p-dd}, we can show that $p$ has at most $O(\epsilon^{-d/2}\sqrt{n})$ neighbors. This is because (a) each neighbor $q$ is in a ball of radius $\ell$ from $p$ and (b) we can extend $N$ to a $\frac{\epsilon \ell}{96}$-net $N'$ of $X$ that has $|N'| \leq n$. Thus, by Lemma~\ref{lm:packing-cd}, the number of neighbors of $p$ must smaller than the size of all the net point $N$ in the ball of radius $2\ell$ centered at $p$, which is at most:
\begin{equation*}
\left(\frac{4\cdot 2\ell}{\epsilon\ell/48}\right)^{d/2}\sqrt{|N'|} = O(\epsilon^{-d/2})\sqrt{n}
\end{equation*}
assuming that $d$ is a constant.  Thus, by the same argument in the proof of Lemma~\ref{lm:oracle-dd}, the strong sparsity of  the oracle is bounded by $O(\epsilon^{-d/2}\sqrt{n})$.
\end{proof}

\section{Conclusion} 

We have introduced the notion of sparse spanner oracles, and proved a necessary and sufficient condition of the existence of light subset spanners via sparse spanner oracles. From this, we obtain several results. The most significant result is the first PTAS for the subset TSP problem in $H$-minor-free graphs. Two other interesting result  are spanners with lightness $O(\epsilon^{-(d+2)})$ in doubling metrics of constant dimension $d$ and subset spanners with lightness $O(\epsilon^{-d/2}\sqrt{n})$ for any $n$-point metric of constant correlation dimension $d$. There are several open problems arisen from our work:
\begin{itemize}
\item[1.]{\emph{Light subset spanners in $H$-minor-free graphs}}. It would be interesting to remove the $\log k$ factor in the lightness of our subset spanner in Theorem~\ref{thm:main}. This would imply an efficient PTAS for subset TSP in $H$-minor-free graphs. A possible line of attack is to construct an approximate terminal preserving minors for $H$-minor-free graphs with a linear number of Steiner vertices (see Lemma~\ref{lm:minor-to-oracle}). It should be noted that an exact distance preserving minor with a linear number of Steiner vertices is not possible due to a lower bound by Krauthgamer, Nguy{\ee}n, and Zondiner~\cite{KNZ14}. However, their lower bound does not rule out an approximate one with the desired property.

\item[2]{\emph{Tight bounds for light spanners in doubling metrics.}} In a recent joint work with Solomon~\cite{LS19}, we showed that there exists a point set in the Euclidean space of dimension $d$ such that any $(1+\epsilon)$ spanner has $\Omega(\epsilon^{-d})$ lightness. This only implies a lightness lower bound $\epsilon^{-\Omega(d)}$ on lightness of spanners in doubling metrics of dimension $d$. The upper bound $O(\epsilon^{-(d+2)})$ was proved in Theorem~\ref{thm:dd-metric}. Using a the (fairly complicated) technique in~\cite{LS19}, it is possible to shave a $\frac{1}{\epsilon}$ factor from the lightness in Theorem~\ref{thm:reduction}. Thus, lightness upper bound $O(\epsilon^{-(d+1)})$ is achievable using current machinery, and we conjecture that this is the optimal bound.
\item[3]{\emph{Other applications of sparse spanner oracles.}}  We have show a number of applications of our sparse spanner oracles. It would be interesting to see more applications of this concept.
\end{itemize} 

\paragraph{Acknowledgement:} We thank Cora Borradaile for constructive comments. We thank Ofer Neiman for asking a question regarding light spanners of doubling metrics during a workshop at ICERM, Brown, that led to results in Section~\ref{sec:spare-oracle-Euc}. We thank an anonymous reviewer for  comments that significantly improve the readability of this paper. This material is based upon work supported by the National Science Foundation under Grant Nos.\ CCF-1252833, a NSERC grant and a PIMS postdoctoral fellowship. 

\bibliographystyle{plain}
\bibliography{spanner}

\begin{thebibliography}{10}

\bibitem{AG06}
I.~Abraham and C.~Gavoille.
\newblock Object location using path separators.
\newblock In {\em Proceedings of the Twenty-fifth Annual ACM Symposium on
  Principles of Distributed Computing}, PODC '06, pages 188--197, 2006.

\bibitem{AGGNT19}
I.~Abraham, C.~Gavoille, A.~Gupta, O.~Neiman, and K.~Talwar.
\newblock Cops, robbers, and threatening skeletons: Padded decomposition for
  minor-free graphs.
\newblock {\em SIAM Journal on Computing}, 48(3):1120--1145, 2019.
\newblock (Announced at STOC`14).

\bibitem{ADDJS93}
I.~Alth\"{o}fer, G.~Das, D.~Dobkin, D.~Joseph, and J.~Soares.
\newblock On sparse spanners of weighted graphs.
\newblock {\em Discrete Computational Geometry}, 9(1):81--100, 1993.

\bibitem{AGKKW98}
S.~Arora, M.~Grigni, D.~R. Karger, P.~N. Klein, and A.~Woloszyn.
\newblock A polynomial-time approximation scheme for weighted planar graph
  {TSP}.
\newblock In {\em Proceedings of the 9th Annual ACM-SIAM Symposium on Discrete
  Algorithms}, SODA '98, pages 33--41, 1998.

\bibitem{Assouad83}
P.~Assouad.
\newblock Plongements lipschitziens dans $\mathbb{R}^n$.
\newblock {\em Soci\'{e}t\'{e} math\'{e}matique de France}, 111:429--448, 1983.

\bibitem{ABP90}
B.~Awerbuch, A.~Baratz, and D.~Peleg.
\newblock Cost-sensitive analysis of communication protocols.
\newblock In {\em Proceedings of the Ninth Annual ACM Symposium on Principles
  of Distributed Computing}, PODC'90, pages 177--187, 1990.

\bibitem{ABP91}
B.~Awerbuch, A.~Baratz, and D.~Peleg.
\newblock Efficient broadcast and light-weight spanner, 1991.
\newblock Manuscript.

\bibitem{BLR10}
P.~Biswal, J.~Lee, and S.~Rao.
\newblock Eigenvalue bounds, spectral partitioning, and metrical deformations
  via flows.
\newblock {\em Journal of the ACM}, 57(3):13:1--13:23, 2010.

\bibitem{BCKN15}
H.~L. Bodlaender, M.~Cygan, S.~Kratsch, and J.~Nederlof.
\newblock Deterministic single exponential time algorithms for connectivity
  problems parameterized by treewidth.
\newblock {\em Information and Computation}, 243(Supplement C):86--111, 2015.

\bibitem{BFLPST16}
Hans~L Bodlaender, Fedor~V Fomin, Daniel Lokshtanov, Eelko Penninkx, Saket
  Saurabh, and Dimitrios~M Thilikos.
\newblock (meta) kernelization.
\newblock {\em Journal of the ACM (JACM)}, 63(5):44, 2016.

\bibitem{BMT13}
S.~Borne, A.~R. Mahjoub, and R.~Taktak.
\newblock A branch-and-cut algorithm for the multiple steiner {TSP} with order
  constraints.
\newblock {\em Electronic Notes in Discrete Mathematics}, 41:487--494, 2013.

\bibitem{Borradaile13}
G.~Borradaile.
\newblock {TSP} in 1-planar graphs.
\newblock
  http://blogs.oregonstate.edu/glencora/2013/11/01/tsp-1-planar-graphs/,
  accessed 10/2018, 2013.

\bibitem{BDT14}
G.~Borradaile, E.~D. Demaine, and S.~Tazari.
\newblock Polynomial-time approximation schemes for subset-connectivity
  problems in bounded-genus graphs.
\newblock {\em Algorithmica}, 68(2):287--311, 2014.
\newblock Announced at STACS 09.

\bibitem{BLW17}
G.~Borradaile, H.~Le, and C.~Wulff-Nilsen.
\newblock Minor-free graphs have light spanners.
\newblock In {\em 2017 IEEE 58th Annual Symposium on Foundations of Computer
  Science}, FOCS '17, pages 767--778, 2017.

\bibitem{BLW19}
G.~Borradaile, H.~Le, and C.~Wulff-Nilsen.
\newblock Greedy spanners are optimal in doubling metrics.
\newblock In {\em Proceedings of the Thirtieth Annual ACM-SIAM Symposium on
  Discrete Algorithms}, pages 2371--2379, 2019.

\bibitem{CG12}
T~.H. Chan and A.~Gupta.
\newblock Approximating {TSP} on metrics with bounded global growth.
\newblock {\em SIAM Journal on Computing}, 41(3):587--817, 2012.
\newblock Announced at SODA'08.

\bibitem{CGMZ16}
T.{-}H.~Hubert Chan, A.~Gupta, B.~M. Maggs, and S.~Zhou.
\newblock On hierarchical routing in doubling metrics.
\newblock {\em {ACM} Trans. Algorithms}, 12(4):55:1--55:22, 2016.
\newblock {Preliminary version appeared in SODA 2005.}

\bibitem{CDNS92}
B.~Chandra, G.~Das, G.~Narasimhan, and J.~Soares.
\newblock New sparseness results on graph spanners.
\newblock In {\em Proceedings of the Eighth Annual Symposium on Computational
  Geometry}, 1992.

\bibitem{CW16}
S.~Chechik and C.~Wulff-Nilsen.
\newblock Near-optimal light spanners.
\newblock In {\em Proceedings of the 27th Annual ACM-SIAM Symposium on Discrete
  Algorithms}, SODA'16, pages 883--892, 2016.

\bibitem{CGH16}
Y.~K. Cheung, G.~Goranci, and M.~Henzinger.
\newblock Graph minors for preserving terminal distances approximately - lower
  and upper bounds.
\newblock In {\em 43rd International Colloquium on Automata, Languages, and
  Programming (ICALP 2016)}, volume~55 of {\em Leibniz International
  Proceedings in Informatics (LIPIcs)}, pages 131:1--131:14, 2016.

\bibitem{CFN85}
G.~Cornu{\'e}jols, J.~Fonlupt, and D.~Naddef.
\newblock The traveling salesman problem on a graph and some related integer
  polyhedra.
\newblock {\em Mathematical Programming}, 33(1):1--27, 1985.

\bibitem{DHN93}
G.~Das, P.~Heffernan, and G.~Narasimhan.
\newblock Optimally sparse spanners in 3-dimensional euclidean space.
\newblock In {\em Proceedings of the 9th Annual Symposium on Computational
  Geometry}, SCG '93, pages 53--62, 1993.

\bibitem{DNS95}
G.~Das, G.~Narasimhan, and J.~Salowe.
\newblock A new way to weigh malnourished euclidean graphs.
\newblock In {\em Proceedings of the 6th Annual ACM-SIAM Symposium on Discrete
  Algorithms}, SODA '95, pages 215--222, 1995.

\bibitem{DHK11}
E.~D. Demaine, M.~Hajiaghayi, and K.~Kawarabayashi.
\newblock Contraction decomposition in {H}-minor-free graphs and algorithmic
  applications.
\newblock In {\em Proceedings of the 43rd Annual ACM Symposium on Theory of
  Computing}, STOC ' 11, pages 441--450, 2011.

\bibitem{DHM07}
E.~D. Demaine, M.~Hajiaghayi, and B.~Mohar.
\newblock Approximation algorithms via contraction decomposition.
\newblock In {\em Proceedings of the Eighteenth Annual ACM-SIAM Symposium on
  Discrete Algorithms}, SODA '07, pages 278--287, 2007.

\bibitem{FS16}
A.~Filtser and S.~Solomon.
\newblock The greedy spanner is existentially optimal.
\newblock In {\em Proceedings of the 2016 ACM Symposium on Principles of
  Distributed Computing}, PODC '16, pages 9--17, 2016.

\bibitem{FLRS11}
F.~V. Fomin, D.~Lokshtanov, V.~Raman, and S.~Saurabh.
\newblock Bidimensionality and {EPTAS}.
\newblock In {\em Proceedings of the Twenty-second Annual ACM-SIAM Symposium on
  Discrete Algorithms}, SODA'11, pages 748--759, 2011.

\bibitem{Gottlieb15}
L.~A. Gottlieb.
\newblock A light metric spanner.
\newblock In {\em 2015 IEEE 56th Annual Symposium on Foundations of Computer
  Science}, pages 759--772, 2015.

\bibitem{GKP95}
M.~Grigni, E.~Koutsoupias, and G.~Papadimitriou.
\newblock An approximation scheme for planar graph {TSP}.
\newblock In {\em Proceedings of the 36th Annual Symposium on Foundations of
  Computer Science}, FOCS '95, pages 640--645, 1995.

\bibitem{GS02}
M.~Grigni and P.~Sissokho.
\newblock Light spanners and approximate {TSP} in weighted graphs with
  forbidden minors.
\newblock In {\em Proceedings of the 13th Annual ACM-SIAM Symposium on Discrete
  Algorithms}, SODA '02, pages 852--857, 2002.

\bibitem{GKL03}
A.~Gupta, R.~Krauthgamer, and J.~R. Lee.
\newblock Bounded geometries, fractals, and low-distortion embeddings.
\newblock In {\em 44th Annual IEEE Symposium on Foundations of Computer
  Science}, pages 534--543, 2003.

\bibitem{HK71}
M.~Held and R.~M. Karp.
\newblock The traveling-salesman problem and minimum spanning trees: {P}art
  {II}.
\newblock {\em Mathematical Programming}, 1(1):6--25, 1971.

\bibitem{KKS11}
K.~Kawarabayashi, P.~N. Klein, and C.~Sommer.
\newblock Linear-space approximate distance oracles for planar, bounded-genus
  and minor-free graphs.
\newblock In {\em Proceedings of the 38th International Colloquim Conference on
  Automata, Languages and Programming}, ICALP '11, pages 135--146, 2011.

\bibitem{Keil88}
J.~M. Keil.
\newblock Approximating the complete euclidean graph.
\newblock In {\em Proceedings of the first Scandinavian Workshop on Algorithm
  Theory}, SWAT `88, pages 208--213, 1988.

\bibitem{KG92}
J.~M. Keil and C.~A. Gutwin.
\newblock Classes of graphs which approximate the complete {E}uclidean graph.
\newblock {\em Discrete and Computational Geometry}, 7(1):13--28, 1992.

\bibitem{KRY93}
S.~Khuller, B.~Raghavachari, and N.~Young.
\newblock Balancing minimum spanning and shortest path trees.
\newblock In {\em Proceedings of the Fourth Annual ACM-SIAM Symposium on
  Discrete Algorithms}, SODA'93, pages 243--250, 1993.

\bibitem{Klein05}
P.~N. Klein.
\newblock A linear-time approximation scheme for planar weighted {TSP}.
\newblock In {\em Proceedings of the 46th Annual IEEE Symposium on Foundations
  of Computer Science}, FOCS '05, pages 647--657, 2005.

\bibitem{Klein06}
P.~N. Klein.
\newblock Subset spanner for planar graphs, with application to subset {TSP}.
\newblock In {\em Proceedings of the 38th Annual ACM Symposium on Theory of
  Computing}, STOC '06, pages 749--756, 2006.

\bibitem{KM14}
P.~N. Klein and D.~Marx.
\newblock A subexponential parameterized algorithm for {S}ubset {TSP}.
\newblock In {\em Proceedings of the 25th Annual ACM-SIAM Symposium on Discrete
  Algorithms}, SODA '14, pages 1812--1830, 2014.

\bibitem{Kostochka82}
A.~V. Kostochka.
\newblock The minimum hadwiger number for graphs with a given mean degree of
  vertices.
\newblock {\em Metody Diskretnogo Analiza}, 38:37--58, 1982.
\newblock In Russian.

\bibitem{KNZ14}
R.~Krauthgamer, H.~L. Nguy$\tilde{\hat{\mbox{e}}}$n, and T.~Zondiner.
\newblock Preserving terminal distances using minors.
\newblock {\em {SIAM} Journal on Discrete Mathematics}, 28(1):127--141, 2014.

\bibitem{LS19}
H.~Le and S.~Solomon.
\newblock Truly optimal euclidean spanners.
\newblock In {\em 60th Annual IEEE Symposium on Foundations of Computer Science
  (to appear)}, FOCS'19, 2019.

\bibitem{LN15}
A.~N. Letchford and S.~D. Nasiri.
\newblock The {S}teiner travelling salesman problem with correlated costs.
\newblock {\em European Journal of Operational Research}, 245(1):62--69, 2015.

\bibitem{MPP18}
D.~Marx, M.~Pilipczuk, and M.~Pilipczuk.
\newblock On subexponential parameterized algorithms for steiner tree and
  directed subset {TSP} on planar graphs.
\newblock In {\em 59th Annual Symposium on Foundations of Computer Science},
  FOCS'18, pages 474--484, 2018.

\bibitem{Mehlhorn88}
K.~Mehlhorn.
\newblock A faster approximation algorithm for the steiner problem in graphs.
\newblock {\em Information Processing Letters}, 27(3):125 -- 128, 1988.

\bibitem{Miller86}
G.~L. Miller.
\newblock Finding small simple cycle separators for 2-connected planar graphs.
\newblock {\em Journal of Computer and System Sciences}, 32(3):265 -- 279,
  1986.

\bibitem{NS07}
G.~Narasimhan and M.~Smid.
\newblock {\em Geometric Spanner Networks}, chapter Geometric Analysis: The
  Leapfrog Property, pages 257--317.
\newblock Cambridge University Press, 2007.

\bibitem{RS98}
S.~B. Rao and W.~D. Smith.
\newblock Approximating geometrical graphs via ``spanners'' and ``banyans''.
\newblock In {\em Proceedings of the 30th Annual ACM Symposium on Theory of
  Computing}, STOC '98, pages 540--550, 1998.

\bibitem{RS03}
N.~Robertson and P.~D. Seymour.
\newblock Graph minors. {XVI}. {E}xcluding a non-planar graph.
\newblock {\em Journal of Combinatoral Theory Series B}, 89(1):43--76, 2003.

\bibitem{RS91}
J.~Ruppert and R.~Seidel.
\newblock Approximating the $d$-dimensional complete {E}uclidean graph.
\newblock In {\em Proceedings of the 3rd Canadian Conference on Computational
  Geometry}, CCCG `91, page 207–210, 1991.

\bibitem{Salazar03}
J.~Salazar-Gonz{\'a}lez.
\newblock The steiner cycle polytope.
\newblock {\em European Journal of Operational Research}, 147(3):671--679,
  2003.

\bibitem{Smid09}
M.~Smid.
\newblock The weak gap property in metric spaces of bounded doubling dimension.
\newblock In Susanne Albers, Helmut Alt, and Stefan N\"{a}her, editors, {\em
  Efficient Algorithms}, pages 275--289. Springer-Verlag, 2009.

\bibitem{VV86}
L.G. Valiant and V.V. Vazirani.
\newblock {NP} is as easy as detecting unique solutions.
\newblock {\em Theoretical Computer Science}, 47(0):85 -- 93, 1986.

\bibitem{Yao82}
A.~C. Yao.
\newblock On constructing minimum spanning trees in $k$-dimensional spaces and
  related problems.
\newblock {\em SIAM Journal on Computing}, 11(4):721--736, 1982.

\bibitem{ZTXL15}
H.~Zhang, W.~Tong, Y.~Xu, and G.~Lin.
\newblock The steiner traveling salesman problem with online edge blockages.
\newblock {\em European Journal of Operational Research}, 243(1):30--40, 2015.

\bibitem{ZTXL16}
H.~Zhang, W.~Tong, Y.~Xu, and G.~Lin.
\newblock The {S}teiner traveling salesman problem with online advanced edge
  blockages.
\newblock {\em Computers and Operations Research}, 70:22--38, 2016.

\end{thebibliography}
\appendix

\section{Missing Figures}\label{app:figs}

\begin{figure}[h]
   \centering
   \begin{tabular}{@{}c@{\hspace{.5cm}}c@{}}
       (a)\includegraphics[width=.2\textwidth]{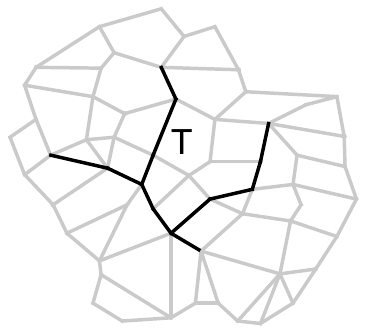}  
       (b)\includegraphics[width=.2\textwidth]{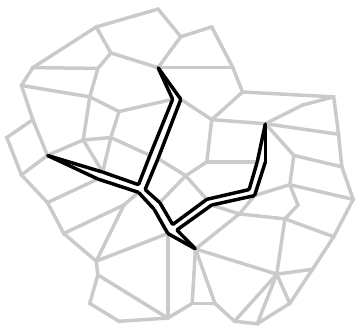} 
       (c)\includegraphics[width=.2\textwidth]{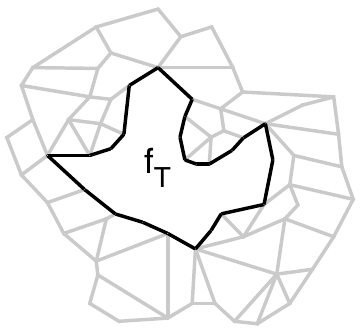} \\
   \end{tabular}
 \caption{(a) Finding a 2-approximation Steiner tree $T$, (b) doubling every edge of $T$ and (c) making a new infinite face $f_T$ from the copies of edges of $T$.  The picture is  republished by courtesy of Glencora Borradaile.}
 \label{fig:cutting}
\end{figure}

\begin{figure}[h]
   \centering
   \begin{tabular}{@{}c@{\hspace{.5cm}}c@{}}
       (a)\includegraphics[width=.2\textwidth]{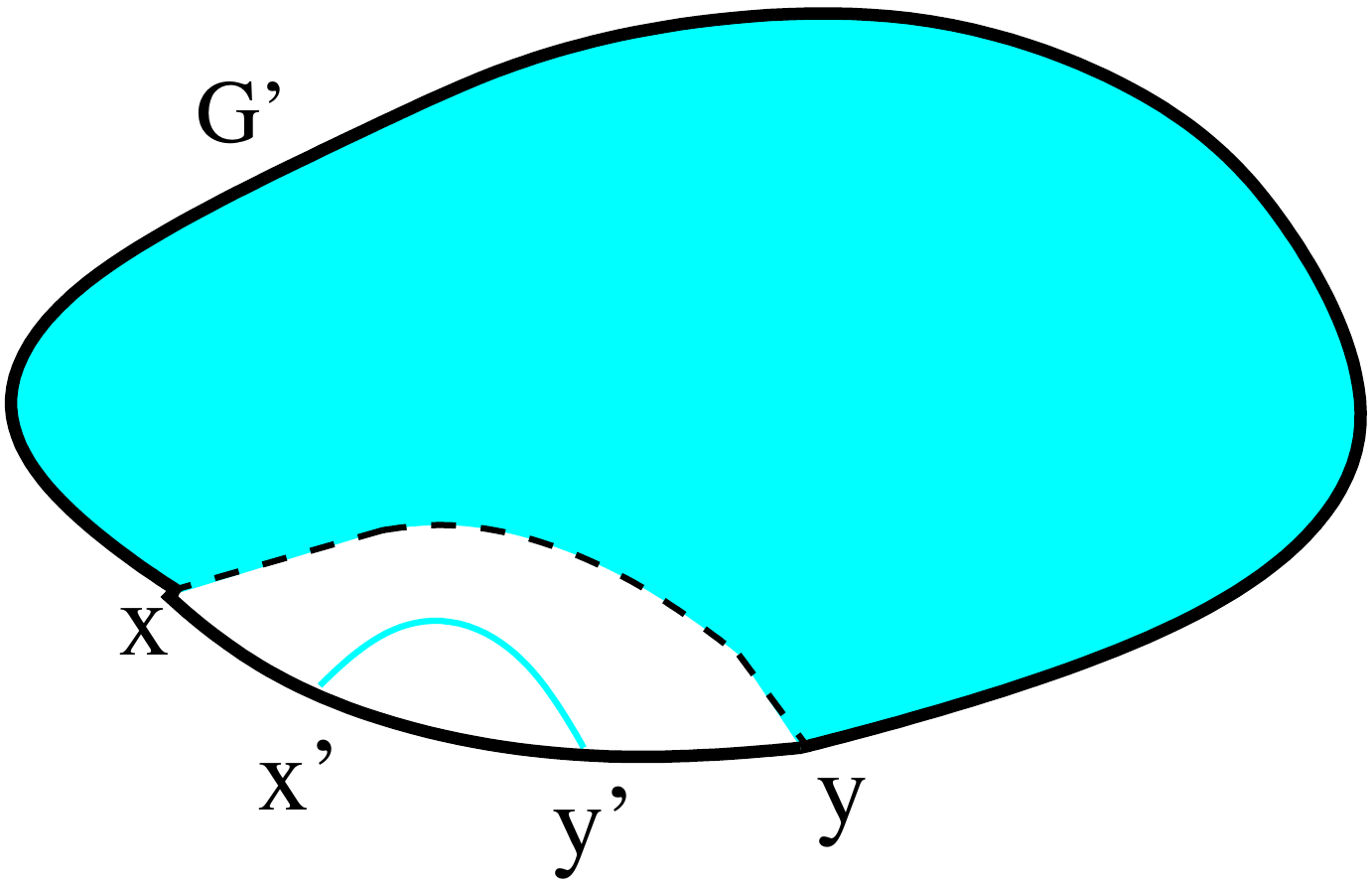} & 
       (b)\includegraphics[width=.2\textwidth]{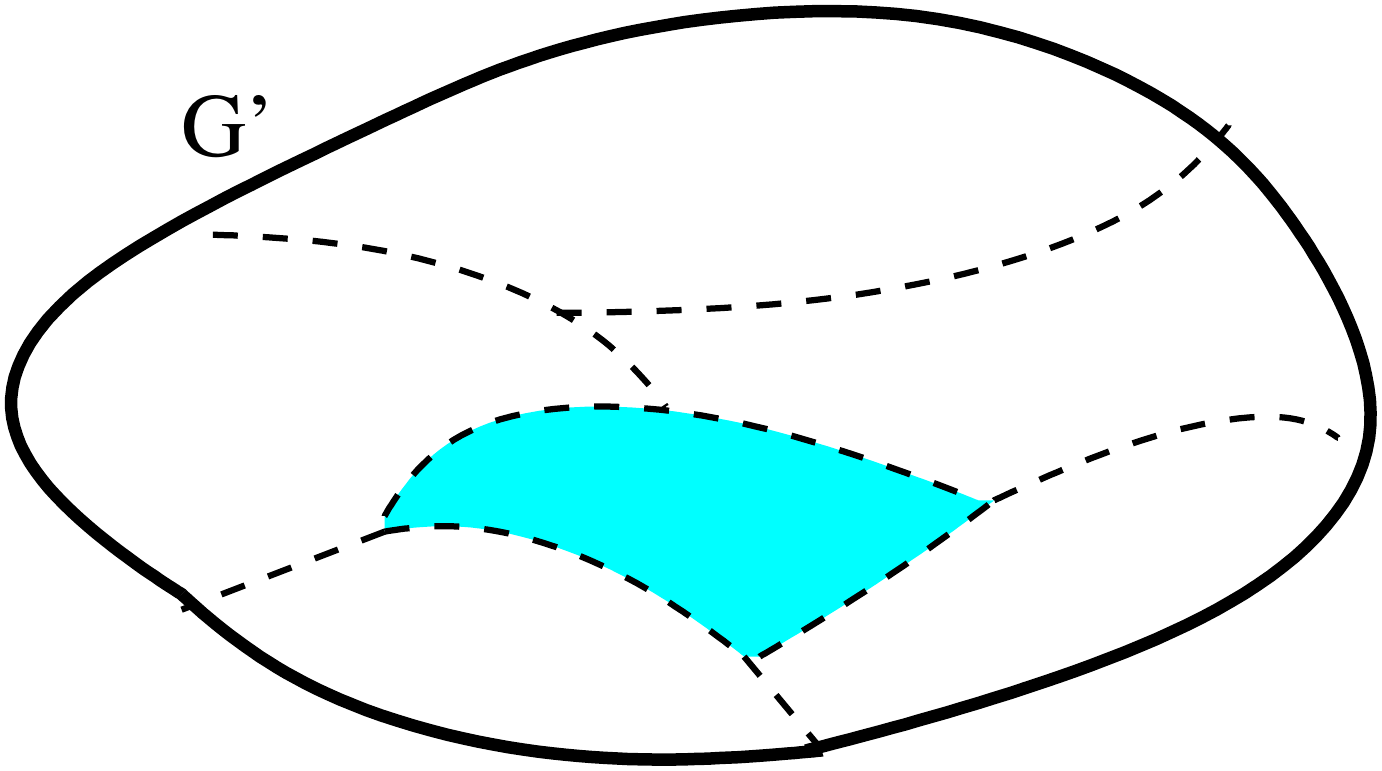} 
   \end{tabular}
 \caption{(a) The first strip is created by adding a shortest (dashed) path between $x$ and $y$ in $G'$. By minimality of $\partial G'[x,y]$, for any $x',y' \in \partial G'[x,y]$, path $\partial G'[x',y']$ well approximates the shortest path between $x'$ and $y'$ in $G'$. (b) A strip decomposition obtained by the recursive procedure. The shaded region is an example strip.   This picture is  republished by courtesy of Glencora Borradaile.}
 \label{fig:strip-decomp}
\end{figure}

\section{Cluster construction} \label{app:cluster-const}

In this section, we describe the details of the cluster construction in prior work by Borradaile, Le and Wulff-Nilsen~\cite{BLW17,BLW19}, with some minor details specific to our setting. The construction has four phases. The intuition of each phase has already been given in Section~\ref{sec:weight-bound}. Let $\mathcal{T}$ is a \emph{cluster tree}, where each node of $\mathcal{T}$ corresponds to an $\epsilon$-cluster and each edge corresponds to an $\MST$ edge connecting the two corresponding $\epsilon$-clusters. In the analysis below, a cluster never takes the credit of an $\MST$ edge outside it. Thus, credits of edges of $\mathcal{T}$ remain intact after level $i-1$. For each node $\mbc \in \mathcal{K}$, we refer to $\dm(\mc_{i-1}(\mbc))$ as its \emph{diameter}.

\paragraph{Phase 1: high degree nodes}  This phase has three steps. The main purpose is to group every high degree nodes and its neighbors into level-$i$ clusters.

(Step 1) Let $\mbx \in \mathcal{K}$ be a high degree node such that all of its  neighbors are unmarked. We form a new level-$i$ cluster $C$ from $\mbx$, its  neighbors  and the connecting edges. We then mark every node of $C$ and repeat this step. 

(Step 2) For each unmarked high-degree vertex $\mby$, there must be a neighbor, say $\mbz$ that is marked in Step 1. Let $C$ be the level-$i$ cluster formed in step 1 containing $\mbz$. We augment $C$ by $\mby$, its unmarked neighbors in $\mathcal{K}$ and the connecting edges. We then mark $\mby$, its neighbors and repeat this step until it no longer applies.

(Step 3) Let $\mby'$ be an unmarked low degree node that has a high degree neighbor $\mbz'$. By construction in Step 2, $\mbz'$ must be marked in step 1.  Let $C$ be the level-$i$ cluster containing $\mbz'$. We augment $C$ by $\mby'$ and the edge between $\mby'$ and $\mbz'$.

Note that $C$ is a subgraph of $M_T$. We then make $C$ a subgraph of $S_{i}$ by replacing each vertex $\mbx \in C$ by subgraph $\mc_{i-1}(\mbx)$ and each edge $\mbe$ by a shortest path of length at most $(1+\epsilon)w(e)$ between its endpoints in $S_{i}$ if it is incident to two high degree nodes.

Observe by the construction that step 1 clusters have diameter at most $2(1+\epsilon)\ell_i + 3g\epsilon\ell_i$ since each edges weight at most $(1+\epsilon)\ell_i$ and each node has diameter at most $g\epsilon\ell_i$. The augmentation in step 2 increases the diameter by at most $4(1+\epsilon)\ell_i + 4g\epsilon\ell_i$ and the augmentation in step 3 does not increase the worst case bound on the diameter. 

\begin{observation} \label{obs:step1-diam} Phase 1 clusters have diameter at most $19\ell_i$ when $\epsilon$ is smaller than $\min(\frac{1}{g},\frac{1}{s})$.
\end{observation}

By choosing $g > 19$, Observation~\ref{obs:step1-diam} implies invariant (DC1). However, we will augment Phase 1 clusters further in Phase 4. It is not a problem as long as the diameter blow up is at most $O(\ell_i)$ since we can choose $g$ to be an arbitrarily big constant (and independent of $\epsilon$).  Since every Phase 1 cluster contains a high degree node and all of its neighbors, it has at least $\frac{2g}{\epsilon}+1$ nodes. By Observation~\ref{obs:rich}, Lemma~\ref{lm:credit-informal} holds for $\epsilon$-clusters involved in Phase 1.

Note that after Phase 1, every unmarked node, say $\mbx$, of $\mathcal{K}$ has low degree. Thus, all $\mbx$'s incident edges (or more precisely, the shortest paths corresponding to these edges in $G$) are added to $S_i$. Therefore, in the construction below, we can conclude that every cluster is a subgraph of $S_i$.

We define \emph{effective diameter} of a path $\mathcal{P}$, denoted by $\edm(\mathcal{P})$, of $\mathcal{T}$ to be the sum of diameter of its nodes. 

\paragraph{Phase 2: Low-degree, branching vertices} In this phase, we group unmarked nodes after Phase 1 into clusters.  We say a node $\mbv$ \emph{$\mt'$-branching} in a tree $\mt'$ if it has degree at least 3 in $\mt'$. Let $\mt'$ be a minimal subtree of unmarked nodes of $\mt$ of effective diameter at least $\ell_i$ and at most $2\ell_i$ that has a $\mt'$-branching node, say $\mbx$. We form a new cluster from $\mt'$, mark every node of $\mt'$ and repeat.  

After Phase 2, unmarked nodes form a subtree of $\mt$ of effective diameter most $\ell_i$ or a path of effective diameter at least $\ell_i$. In Phase 4 below, Phase 2 clusters will be augmented further by low-diameter subtrees of $\mt$.  An important property is that  each Phase 2 cluster is still a subtree of $\mt$ after the augmentation.

The key observation to bound the diameter of a Phase 2 cluster is that any path $\mathcal{P}$ of $\mathcal{T}$ has $\dm(P)\leq 2\edm(P)$. This is because $\MST$ edges have length at most $w_0$, which is smaller than the diameter of any node. 

Let $\mathcal{D}$ be a diameter path of a Phase 2 cluster, say $C$.  Since every edge of $\mathcal{D}$ has credits at least its weight and every node of $\mathcal{D}$ has credits at least its diameter (by invariant (DC1) for level $i-1$), the total credit of vertices and edges of $\mathcal{D}$ is at least $\ce\dm(C)$. Since $\edm(C) \geq \ell_i$, invariant (DC2) is maintained.  Since $C$ has a branching node $\mbx$, at least one of $\mbx$ neighbor, say $\mby$, is not in $\mathcal{D}$, hence its credits are leftover. If $C$ has more than $\frac{2g}{\epsilon}+1$ nodes, Lemma~\ref{lm:credit-informal} is satisfied for $\epsilon$-clusters in $C$ by Observation~\ref{obs:rich}. Otherwise, by redistributing the credit of $\mby$ to every node in $C$, each gets at least $\Omega\left(\frac{\ce \ell_{i-1}}{2 |C|}\right) ~=~ \Omega(\epsilon \ce \ell_{i-1})$ credits, thereby implying Lemma~\ref{lm:credit-informal}.

\paragraph{Phase 3: High-diameter paths of $\mt$} Let $\mp$ be a high diameter path of nodes which are unmarked after Phase 1 and 2. We say a node $\mbv \in \mp$ \emph{deep} if it is not an endpoint of $\mp$ and the two subpaths of $\mp - \{\mbv\}$ have effective diameter at least $\ell_i$ each.  Let $e$ be a level-$i$ edge with two deep endpoints, say $\mbx,\mby$. Let $\mx$ and $\my$ be two paths of $\mt$ containing $\mbx$ and $\mby$, respectively. It may be that $\mx \equiv \my$ (both endpoints of $e$ are on the same path.). Let $\mp_\mbx, \mq_\mbx$ be two minimal subpaths of $\mx - \{\mbx\}$ incident to $\mbx$ that have effective diameter at least $\ell_i$.  $\mp_\mbx, \mq_\mbx$ exist since $\mbx$ is deep. We define two minimal subpaths $\mp_\mby, \mq_\mby$ of $\my$ similarly. We then group $e$, $\mp_\mbx,\mp_\mby, \mq_\mbx,\mq_\mby$ into a new level-$i$ cluster. We mark nodes in $\mp_\mbx \cup \mp_\mby\cup \mq_\mbx\cup\mq_\mby\cup \{\mbx,\mby\}$ and repeat until this phase no longer applies.

See Figure~\ref{fig:P3N} for an illustration of clusters formed in this phase. It is possible that two paths among four paths in $\{\mp_\mbx, \mp_\mby, \mq_\mbx,\mq_\mby\}$, say $\mp_\mbx, \mp_{\mby}$, overlap. We call this case the \emph{cyclic case} since the corresponding cluster contains a cycle. In the cyclic case, we redefine $\mp_\mbx = \mp_{\mby} = \mp_{\mbx\mby}$ where $\mp_{\mbx\mby} = \mp[\mbx,\mby] \setminus \{\mbx,\mby\}$ (see Figure~\ref{fig:P3N}(c)). 
 
\begin{figure}
\vspace{-2pt}
\centering
\includegraphics[scale = 1.1]{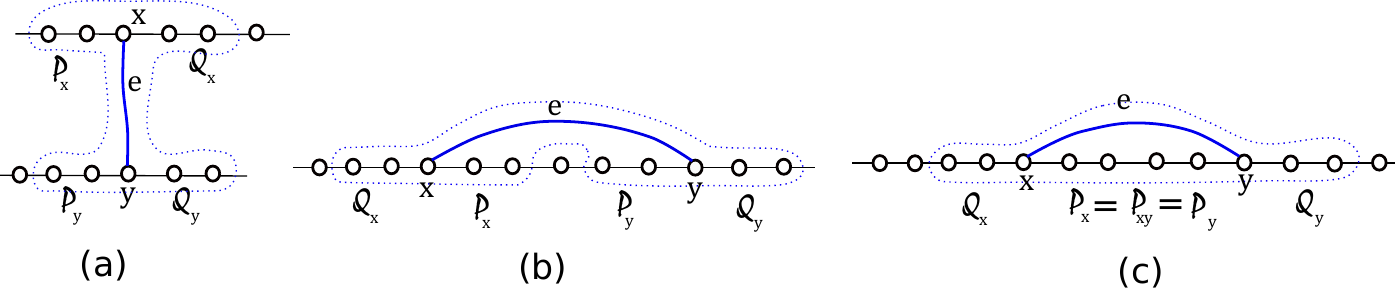}
\caption{Phase 3 clusters are enclosed in dotted blue curves. There are three different forms that a  Phase 3 cluster can take. The blue thick edge is a level $i$ edge $e$ with two endpoints $\mbx,\mby$.}
\label{fig:P3N}
\end{figure}

By minimality, $\mp_\mbx, \mp_\mby, \mq_\mbx,\mq_\mby$ all have diameter at most $2(\ell_i + g\epsilon\ell_i)$. Since $w(e) \leq \ell_i$, we have:
\begin{observation} \label{obs:phase3-diam}
Phase 3 clusters have diameter at most $9\ell_i + 10g\epsilon\ell_i + 4w_0$. 
\end{observation}

The extra term $2g\epsilon\ell_i$ is the total weight of $\mbx$ and $\mby$ and the term $4w_0$ is the total weight of four tree edges incident to $\mbx$ and $\mby$. By definition, they do not belong to any of four paths $\mp_\mbx, \mp_\mby, \mq_\mbx,\mq_\mby$. Notte that $w_0 \leq \ell_i$. By choosing $\epsilon$ sufficiently smaller than $1/g$, the diameter of Phase 3 clusters is $O(\ell_i)$. In Phase 4 below, we further augment Phase 3 clusters by subtrees of $\mt$ of diameter at most $O(\ell_i)$. The resulting clusters still have diameter at most $O(\ell_i)$.

Our goal to argue that the credit of one node is not needed to maintain invariant (DC2). This is exactly what Borradaile, Le and Wulff-Nilsen showed in Case 2 in their paper~\cite{BLW19}. Their proof does not use any special property of doubling metrics, so it is readily applicable to our case. Let us sketch the intuition behind their argument to handle the cyclic case. Recall that an edge $e$ is kept in $\mathcal{K}$ if the shortest path between its endpoints in $S_{i-1}$ is at least $(1+(6g+1)\epsilon)w(e)$. Thus, by adding $e$ to the subpath of $\mathcal{T}$ containing its endpoints, the diameter is reduced by at least $(6g+1)\epsilon w(e) - \dm(\mc_{i-1}(\mbx)) - \dm(\mc_{i-1}(\mby))~>~(6g+1)\epsilon w(e)  - 2g\ell_{i-1}~>~ g\ell_{i-1}$ which is bigger than the diameter of any $\epsilon$-cluster. This diameter reduction is equivalent to having $g\ce\ell_{i-1}$ leftover credits. Thus, the credit of at least one $\epsilon$-cluster can be reserved as leftover.

\paragraph{Phase 4: Remaining nodes} After Phase 3, by removing marked nodes from $\mt$, we obtain a forest $\mf$ such that for every tree $\mt' \in \mf$, either $\mt'$ has effective diameter at most $\ell_i$ or $\mt'$ is a path of effective diameter at least $\ell_i$.  We decompose $\mf$ into two forests $\mf^{aug}$ and $\mf^{cluster}$ as follows. If a tree $\mt' \in \mf$ has diameter at most $\ell_i$, we include $\mt'$ in $\mf^{aug}$. Otherwise, we greedily break $\mt'$ into paths of effective diameter at least $\ell_i$ and at most $2\ell_i$ ($\mt'$ is a path in this case.). Let $\mp$ be a path broken from $\mt$. If $\mp$ has a tree edge connecting it to a cluster formed in the first phases, we include $\mp$ in $\mf^{aug}$. Otherwise, we include it in $\mf^{cluster}$ (see Figure~\ref{fig:affices} for an illustration).

\begin{figure}[h]
   \centering
   \includegraphics[width=.5\textwidth]{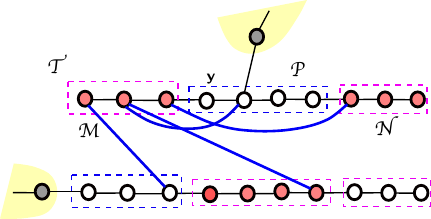} 
   \caption{ A long path $\mt'$ that are broken into short paths (enclosed by dash curves) in Phase 4. The subpaths enclosed by purple dashed curves are in $\mf^{cluster}$ and the subpaths enclosed by blue dashed curves are in $\mf^{aug}$. Black edges are tree edges and thick blue edges are path edges.  $\mt'$ has two affices $\mathcal{M}$ and $\mathcal{N}$. By construction, $\mt'$ has an $\MST$ edge to a cluster (the yellow-shaded region) formed in the first three phases. }
 \label{fig:affices}
\end{figure}

Observe that  by construction, for every tree $\mt'$ in $\mf^{aug}$, there is at least one tree edge, say $e$, connecting it to a  cluser $C$ formed in prior phases. We augment $C$ by attaching $\mt'$ to it via $e$. 

By construction, $\mf^{cluster}$ is a set of paths.  For every path in $\mf^{cluster}$, we form an independent Phase 4 cluster.  This completes the cluster construction. We now show how to maintain invariants.

\paragraph{Invariant (DC1)} By construction, Phase 4 clusters have diameter at most $4\ell_i$, which is at most $g\ell_i$ when $g\geq 4$.  

Let $C$ be a cluster formed in the first three phases and $C'$ be the augmentation of $C$  after Phase 4. Since $C$ is attached subtrees of effective diameter at most $\ell_i$ via tree edges (of length at most $w_0$), we have:
\begin{equation}\label{eq:diam-S}
\dm(C') \leq \dm(C) + 4\ell_i+ 2w_0 
\end{equation}
By construction, Phase 2 clusters have diameter at most $4\ell_i$.  By Observation~\ref{obs:step1-diam} and Observation~\ref{obs:phase3-diam}, clusters in Phase 1 and Phase 3 have diameter at most $23\ell_i$ since $w_0 \leq \ell_i$. By Equation~\ref{eq:diam-S}, we have $\dm(C') \leq 29\ell_i$. By choosing $g = 29$, invariant (DC1) is satisfied.

\paragraph{Invariant (DC2) and Lemma~\ref{lm:credit-informal}}

We have shown that clusters originated in the first three phases can both maintain Invariant (DC2) and each node in these clusters has at least $\Omega(\epsilon\ce\ell_{i-1})$ leftover credits. 

However, credits of nodes in Phase 4 clusters are only enough to guarantee invariant (DC2). By Observation~\ref{obs:rich}, it suffices to consider a Phase 4 cluster $C$ with at most $\frac{2g}{\epsilon}$ nodes.  We show that all nodes of $C$ fall into case (i) of Lemma~\ref{lm:credit-informal}. 

First, observe that by construction in Phase 1, every node involving in Phase 4 clusters has low degree and all of its neighbors also have low degree.   Recall that trees in $\mf^{cluster}$ are subpaths greedily broken from paths of $\mt$ of effective diameter at least $\ell_i$. We distinguish two types of paths in $\mf^{cluster}$: \emph{internal} paths and \emph{affix} paths. Observe that there is no level-$i$ edge between two internal subpaths since otherwise, both endpoints of such an edge would be deep and hence, it will be grouped in Phase 3.  Thus, any internal terminal subpath must have a level-$i$ edge to an affix subpath.

Let $\mt'$ be a long path that are broken in Phase 4. Let $\my$ be the set of nodes that are in at most two affices of $\mt'$ in $\mf^{cluster}$. We can assume that  $|\my| \leq \frac{4g}{\epsilon}$ by Observation~\ref{obs:rich}. Observe that by construction, $\mt'$ must have a tree edge connecting it to subgraphs originated in the first three phases. (The only exception is when there is no cluster formed in the first three phase; we will come to this case later.) That is, at least one subpath, say $\mp$, of $\mt'$ is included in $\mf^{aug}$, and then augmented to a cluster formed in the first three phases. Recall that $\mp$ has effective diameter at least $\ell_i$. Thus, by invariant (DC1), it has at least $\Omega(\frac{1}{g\epsilon})$ nodes.  Since each node of $\mp$ has at  least $\Omega(\epsilon\ce\ell_{i-1})$ leftover credits, by taking the credit of   $\Theta(\frac{1}{g\epsilon})$ nodes of $\mp$ and redistributing to (at most $\frac{4g}{\epsilon}$) nodes of $\my$, each gets at least $\Omega(\epsilon \ce \ell_{i-1})$ leftover credits. Since there is no level-$i$ between two internal paths of $\mf^{cluster}$, nodes in internal paths of $\mf^{cluster}$ satisfy case (i) of Lemma~\ref{lm:credit-informal}.

Finally, we need to handle the exception where there is no cluster formed in the first three phases.  That is,  every node of $\mt$ has degree at most $\frac{2g}{\epsilon}$  and either (i) $\mt$ has effective diameter less than $\ell_i$  or (ii) $\mt$ is a path of  effective diameter at least $\ell_i$ and every level-$i$ edge is incident to a node in two affices of $\mt$. We consider each case separately.

\begin{enumerate}[noitemsep]
\item  If $\edm(\mt) < \ell_i$, our cluster construction stops at this level. We let each vertex keep its own credits as leftover credits. Note that in this case, we do not have level $i+1$ or higher edges since any such edge would have length more than $\dm(\mt)$ when $\epsilon$ is sufficiently smaller than $1$;  contradicting that $M_T$ represents a metric.

\item  If $\mt$ is a path of effective diameter at least $\ell_i$ and every level-$i$ edge is  incident to a node in affices of $\mt$, then $\mf^{cluster}$ only contains subpaths of $\mt$ and $\mf^{aug} = \emptyset$. We use all node and edge credits of each cluster in $\mf^{cluster}$ to guarantee invariant (DC2).  By Observation~\ref{obs:rich}, we can assume that both affices of $\mt$ have at most $\frac{2g}{\epsilon}$ nodes each since otherwise, we can redistribute leftover credits of one affix to another.  Thus, there are at most:
\begin{equation}
2(\frac{2g}{\epsilon})\cdot(\frac{2g}{\epsilon}) = O (\frac{1}{\epsilon^2})
\end{equation}
 level-$i$ edges. This is case (ii) in Lemma~\ref{lm:credit-informal}.
\end{enumerate}

\section{Missing proofs} \label{app:missing}
\subsection{Completing the proof of Lemma~\ref{lm:ss-spanner}}
\begin{proof}
To prove (3), we use the argument in the proof of Theorem 4.1 of Klein~\cite{Klein06}, that we elaborate here for completeness. 
\begin{equation}\label{eq:ss-spanner-proof-of-3}
\begin{split}
w(Q_I) &< (1+\epsilon)^{-1}(w(Q_{I-1}) + d_P(y_{I-1},y_{I}))\\
&\leq (1+\epsilon)^{-1}w(Q_{I-1}) + d_P(y_{I-1},y_{I})\\
&< (1 -\epsilon/2) w(Q_{I-1})+ d_P(y_{I-1},y_{I}) \qquad \mbox{(since }\epsilon < 1)\\
&\leq w(Q_{I-1}) - \frac{\epsilon R}{2} + d_P(y_{I-1},y_{I})\\
&\leq w(Q_0) - I\frac{\epsilon R}{2}+ d_P(y_0,y_{I}) \qquad \mbox{(by solving the recurrent relation) }\\
&= (1 -\epsilon I/2)R + d_P(y_0,y_{I}) \\
&\leq (1 -\epsilon I/2)R + 4\epsilon^{-1}R \qquad \mbox{(by Equation~\eqref{eq:ss-spanner-proof-of-4})}
\end{split}
\end{equation}
Since $w(Q_I) \geq R$, by Equation~\eqref{eq:ss-spanner-proof-of-3}, we have $I < 8\epsilon^{-2}$. By a similar argument, we can show that $J < 8\epsilon^{-2}$.
\end{proof}

\subsection{Proof of Corollary~\ref{cor:path-to-path-spanner}}\label{app:path-to-path-proof}
\begin{proof}
Let  $W = \{v_0,v_1,\ldots, v_r\}$ where $r$ is the length of $W$. Note that there may be a vertex that appears multiple times along $W$.  We define a sequence of vertices $Y = \{y_0 = v_0, y_1,\ldots, y_I\}$ along $W$ as follows: (i) $y_0 = v_0$ and (ii) $y_{i}$ is a closet vertex after $y_{i-1}$ such that:
\begin{equation} \label{eq:path-to-path-spanner}
d_W(y_i,y_{i-1}) > \epsilon d_G(y_i, P)
\end{equation}

For each $y_i$, let $\mathcal{Q}_i \leftarrow$ \textsc{SSSpanner($G, P,y_i,\epsilon$)}. $\mathcal{Q}_i$ is a collection of shortest paths with source $y_i$. Let $H = \mathcal{Q}_0 \cup \ldots \cup \mathcal{Q}_{I}$. We first bound the weight of $H$. Let $R_i =  d_G(y_i,P)$. By Equation~\eqref{eq:path-to-path-spanner}, we have:
\begin{equation*}
\sum_{i=1}^I R_i \leq  \epsilon^{-1} d_W(y_0,y_I) = \epsilon^{-1}w(W)
\end{equation*}
Since $R_0 \leq w(W) + R$, we have:

\begin{equation}\label{eq:path-to-path-spanner-Ri}
\sum_{i=0}^I R_i \leq  (\epsilon^{-1} + 1)w(W) + R
\end{equation}

By (2) of Lemma~\ref{lm:ss-spanner}, we have:

\begin{equation}\label{eq:H-vs-R_i}
w(H) \leq \sum_{i=0}^{I}w(\mathcal{Q}_i) \leq 8\epsilon^{-2}  \sum_{i=0}^{I}R_i
\end{equation}
From Equation~\eqref{eq:path-to-path-spanner-Ri} and Equation~\eqref{eq:H-vs-R_i}, we obtain the desired upper bound on the weight of $H$.

We now show property (1). If $p \in Y$, then property (1) is satisfied by construction and Lemma~\ref{lm:ss-spanner}. Thus, we can assume that $p \not\in Y$. Let $\ell$ be such that $p \in W[y_{\ell}, y_{\ell+1}]$. (If $\ell = I$, we define $ y_{\ell+1}$ to be the endpoint of $W$ after $y_{\ell}$). Since $p \not\in Y$, by Equation~\eqref{eq:path-to-path-spanner}, $d_W(p,y_{\ell}) < \epsilon d_G(p,P)$ which is at most $\epsilon d_G(p,q)$.  Let $M$ be a path from $p$ to $q$ that consists of $W[p,y_\ell]$ and a shortest $y_\ell$-to-$q$ path in $H\cup P$. We have: 

\begin{equation}
\begin{split}
w(M) &\leq w(d_W(p,y_{\ell}) ) + d_{H\cup P}(y_\ell, q) \\
&\leq w(d_W(p,y_{\ell}) ) +  (1+\epsilon)d_{G}(y_\ell, q) \\
&\leq w(d_W(p,y_{\ell}) ) + (1+\epsilon)(d_{G}(y_\ell, p) + d_{G}(p, q)) \qquad (\mbox{by triangle inequailty})\\
&\leq (2 + \epsilon)d_{W}(y_\ell, p) + (1+\epsilon) d_{G}(p, q)  \qquad (d_G(p,y_{\ell}) \leq d_W(p,y_{\ell}))\\
&< (2 + \epsilon)\epsilon d_{G}(p, q) + (1+\epsilon)d_{G}(p, q)\qquad  (d_W(p,y_{\ell}) <\epsilon d_G(p,q)) \\
&\leq (1+4\epsilon)d_{G}(p, q) \qquad (\mbox{since }\epsilon < 1)
\end{split}
\end{equation} 
 By setting $\epsilon' = 4\epsilon$ we have property (1).
\end{proof}

\section{A singly exponential time algorithm for the subset TSP in bounded treewidth graphs} \label{app:dp}

In this section, we give a dynamic program that can solve subset TSP in $2^{O(\mbtw)}n^{O(1)}$. Our algorithm is based on a method introduced by Bodlaender, Cygan, Kratsch, Nederlof ~\cite{BCKN15} to design  deterministic singly exponential time algorithms for connectivity problems in bounded treewidth graphs. 

\subsection{Representing partitions}

Let $U = [n]$ be a ground set of $n$ elements and $\Pi(U)$ be the set of all partitions of $U$.  We abuse notation by using $U$ to denote the partition $\{U\} \in \Pi(U)$, i.e, the partition that has $U$ as the only set. For each partition $\pi \in \Pi(U)$, define a \emph{partition graph} $G_\pi$ where $V(G_\pi) = U$ and  there is an edge between $u$ and $v$ in $G_\pi$ if they are both in the same set of $\pi$. Thus, there is a bijection between sets in $\pi$ and cliques in $G_\pi$. For two elements $u,v \in U$, we denote by $U[uv]$ the partition of $U$ that has $\{u,v\}$ as a set and other sets are singletons. By $\pi\setminus \{v\}$, we denote the partition of $U\setminus \{v\}$ obtained from $\pi$ by removing $v$ from $\pi$.

Let $\alpha, \beta$ be two partitions of $\Pi(U)$ and $G_\alpha,G_\beta$ be two corresponding partition graphs. We define a join operation $\sqcup$ as follows: $\alpha \sqcup \beta$ is a partition $\Pi(U)$ where each set of $\alpha \sqcup \beta$ is a connected component of the graph with vertex set $U$ and edge set $ E(G_\alpha)\cup E(G_\beta)$.

We say partition $\beta$ is an \emph{extension} of partition $\alpha$ if $\alpha\sqcup \beta = U$. Note that a partition can have many different extensions.

Let $\Gamma \subseteq \Pi(U)$ be a set of partitions of $U$.  We say $\widehat{\Gamma}$ is a \emph{representative set} of $\Gamma$ if (i) $\widehat{\Gamma} \subseteq \Gamma$ and (ii) for any partition $\alpha\in \Gamma$ and any extension, say  $\beta$, of $\alpha$, then there is a partition $\hat{\alpha} \in \widehat{\Gamma}$ such that $\beta$ is also an extension of $\hat{\alpha}$ ($\hat{\alpha} \sqcup \beta = U$). We say $\hat{\alpha}$ is a $\beta$-representation of $\alpha$ in $\widehat{\Gamma}$.

Suppose every partition $\alpha \in \Gamma$ has a weight $w(\alpha)$. We say $\widehat{\Gamma}$ is a \emph{min representative set} of $\Gamma$, denoted by $\widehat{\Gamma} \smin \Gamma$ if (i) $\widehat{\Gamma}$ is a representative set of $\Gamma$ and (ii) for every $\alpha\in \Gamma$ and any extension $\beta$ of $\alpha$, there is a $\beta$-representation $\hat{\alpha}$ of $\alpha$ in $\widehat{\Gamma}$ such that $w(\hat{\alpha}) \leq w(\alpha)$.

The key idea in speeding up dynamic programs~\cite{BCKN15} is the following \emph{representation theorem}.
\begin{theorem}[Theorem 3.7~\cite{BCKN15}] \label{thm:representation}
Any set of weighted partitions  $\Gamma$ of $U$ has a min representative set $\widehat{\Gamma}$ of size at most $2^{n-1}$ that can be found in time $|\Gamma| 2^{(\omega-1)n}n^{O(1)}$ where $|U| = n$ and $\omega$ is the matrix multiplication exponent.
\end{theorem}
 
 We note that size of $\Gamma$ can be up to $2^{\Omega(n \log n)}$ but the representation theorem said that it has a min representative set of size at most $2^{n-1}$. 

\subsection{Tree decompositions}

A \emph{tree decomposition} of a graph $G$ is a pair $(\mt,\mx)$ where $\mx$ is a family of subsets of $V$, called \emph{bags}, and $\mt$ is a tree whose nodes are bags in $\mx$ such that:

\begin{enumerate}[nolistsep,noitemsep]
\item[(i)] $\cup_{X \in \mx} X = V(G)$.
\item[(ii)] For every edge $uv \in E$, there is a bag $X \in \mx$ that contains both $u$ and $v$.
\item[(iii)] For every $u \in V$, the set of bags containing $u$ induces a (connected) subtree of $\mt$. 
\end{enumerate} 

The \emph{width} of $(\mt,\mx)$ is $\max_{X \in \mx} |X|-1$ and the \emph{treewidth} of $G$ is the minimum width over all possible tree decompositions of $G$. For each node $t \in \mt$, we denote its corresponding bag by $X_t$.

Traditionally, each bag $X_t$ is a set of vertices of $G$. However, for simplifying presentation of the dynamic program, we think of $X_t$ as a bag of vertices and edges of $G$. That is, $X_t$ is a subgraph of $G$. A tree decomposition $(\mt, \mx)$ is \emph{nice} if it is rooted at a node $r$ where $|X_r| = \emptyset$ and other nodes are one of five following types:
\begin{description}
\item[Leaf node] A leaf node $t$ of $\mt$  has $|X_t| = \emptyset$.
\item[Introduce vertex node] An introduce vertex node $t \in \mt$ has only one child $t'$ such that $X_{t'}$ is a subgraph of $X_t$,  $|V(X_t)| = |V(X_{t'})| + 1$ and $E(X_t) = E(X_{t'})$.
\item[Introduce edge node] An introduce edge node $t \in \mt$ has only one child $t'$ such that $X_{t'}$ is a subgraph of $X_t$. $V(X_{t'}) = V( X_t)$ and $|E(X_t)| = |E(X_{t'})| + 1$.
\item[Forget node] A forget node $t \in \mt$ has only one child $t'$ such that $X_t$ is an induced subgraph of $X_{t'}$ and $|V(X_t)| = |V(X_{t'})| -1$. 
\item[Join node] A join node $t$ has two children $t_1, t_2$ such that $V(X_t) = V(X_{t_1}) = V(X_{t_2})$, $E(X_{t_1})\cap E(X_{t_2}) = \emptyset$ and $E(X_{t}) = E(X_{t_1})\cup E(X_{t_2})$. 
\end{description}
 
A nice tree decomposition has $O(n)$ nodes and can be obtained from any tree decomposition of the same width of $G$ in $O(n)$ time (see Proposition 2.2~\cite{BCKN15}).

\subsection{A dynamic programming algorithm for subset TSP}

Recall $T$ is a set of terminals in a treewidth-$\mbtw$ graph $G$. Let $k = |T|$.  We modify $G$ by adding $k-1$ parallel edges to each edge $e \in G$ and subdividing each new edge by a single vertex. Weight of each edge is splitted equally in the two new edges. The resulting graph is simple and has treewidth $\max(\mbtw, 2)$. This modification of $G$ would guarantee that there is an optimal tour $\Opst$ that visits every edge at most once. ($\Opst$ is an Eulerian subgraph of $G$.)    

For simplicity of presentation, we assume that the optimal solution  $\Opst$  is unique. This assumption can also be technically enforced by imposing a lexicographic order on optimal solutions or by perturbation using Isolation Lemma~\cite{VV86}.

For two edge sets $E_1,E_2$ of $E$. We use $E_1\uplus E_2$ to be the \emph{multiset addition} of $E_1$ and $E_2$. That is, we keep two copies of an edge in $E_1\uplus E_2$ if it appears in both $E_1$ and $E_2$. 

Let $t$ be a node in $\mt$. If $t'$ is a descendant of $t$, we write $t' \preceq t$. Note that $t$ is a descendant of itself. Let $G_t = \cup_{t' \preceq t}X_{t'}$. We regard the optimal solution $W$ as a graph of $G$ with vertex set spans by edges of $W$.  Let $\Opst_t = G_t \cap W$. Note that there could be connected components of $\Opst_t$ that are isolated vertices. We call $\Opst_t$ a \emph{partial solution}. It is straightforward to see that $\Opst_t$ satisfies one of the following two conditions for every node $t$:
\begin{enumerate} [nolistsep,noitemsep]
\item $\Opst_t$ is a feasible solution. That is, $\Opst_t$ is an Eulerian subgraph of $G$ and spans $T$.
\item Every vertex of $T$ in $G_t \setminus X_t$ is in $\Opst_t$, every vertex of $(\Opst_t \cap G_t)\setminus X_t$ has even degree  and every connected component of  $\Opst_t$ contains at least one vertex of $X_t$. 
\end{enumerate}

For each vertex $v\in X_t$, we assign a label $c_t(v) \in \{0,1,2\}$, where $c_t(v) = 0$ if $v$ is not in $\Opst_t$, $c_t(v) = 1$ if $v$ has odd degree in $\Opst_t$ and $c_t(v) = 2$ if $v$ has even degree in $\Opst_t$. We denote the labeling restricted to a subset $Y$ of $X_t$ by $c_t(Y)$.

Let $Y_t = V(\Opst_t \cap X_t)$. Let $\alpha_t$ be the partition of $Y_t$ \emph{induced by} $\Opst_t$. That is, vertices in the same connected component of $\Opst_t$ are in the same set of $\alpha_t$. Let $R_t = E(\Opst)\setminus E(\Opst_t)$ be a subset edges of $W$ not in $G_t$. Let $\beta_t$ be the partition of $Y_t$ induced by $R_t$. Since $W$ is connected, $\alpha_t \sqcup \beta_t = Y_t$. We define the weight of $\alpha_t$ to be $w_t(\alpha_t) = w(\Opst_t)$.

We call tuple $(c_t(Y_t), \alpha_t, Y_t)$ the \emph{encoding} of $\Opst_t$, denoted by $\enc(W_t)$. By definition of $Y_t$, any vertex $ v \in X_t\setminus Y_t$ is not in $W$, hence, $c_t(v) = 0$. Thus, the labeling of vertices $X_t$ is implicitly defined by labeling of vertices in $Y_t$. A encoding is \emph{valid} if it encodes at least one partial solution. We only keep track of valid encodings during dynamic programming. There could be many partial solutions that have the same encoding. However, we only keep track of one partial solution, denoted by $\dec(c_t(Y),\alpha_t, Y_t)$ for each encoding $(c_t(Y_t), \alpha_t, Y_t)$, that has smallest $w(\alpha_t)$. The correctness follows from the following observation.

\begin{observation}\label{obs:remove-encoding} Let $W_t$ and $W_t'$ be two partial solutions that have the same encoding $(c_t(Y_t), \alpha_t, Y_t)$ such that $w(W_t) < w(W_t)'$. If $R_t$ is the set of edges such that $R_t \uplus W_t$ is a feasible solution, then $R_t \uplus W'_t$ is also a feasible solution but has smaller weight.  
\end{observation}

\begin{claim}\label{clm:unique-decode} If $(c_t(Y_t),\alpha_t, Y_t)$ is the encoding of the partial solution $W_t$ of the optimal solution $W$ in $G_t$, then $\dec(c_t(Y_t),\alpha_t, Y_t) = W_t$.
\end{claim}
\begin{proof}
Let $\widehat{W}_t = \dec(c_t,\alpha_t, Y_t)$ and $\widehat{W} = E(\widehat{W}_t) \uplus R_t$.  By definition of decoding, $w(\widehat{W}_t) \leq w(W_t)$. Since $\widehat{W}_t \cap X_t = W_t\cap X_t$ (both are equal to $Y_t$) and labels of vertices in $Y_t$ are the same in both $\widehat{W}_t$ and $W_t$, every vertex in $\widehat{W}$ has even degree. Since $\widehat{W}_t$ is a partial solution, $\widehat{W}$ spans all terminals. Since $R_t$ has no edge in $G_t$, there are no parallel edges in $\widehat{W}_t$. Thus,  $\widehat{W}$  is a feasible solution of subset TSP problem.

However, $w(\widehat{W}) = w(R_t) + w(\widehat{W}_t) \leq w(R_t) + w(W_t)  = w(W)$. By the uniqueness assumption, $W_t = \widehat{W}_t$; the claim follows.  
\end{proof}

For each node $t\in \mt$, we would inductively maintain a set of encodings $\hat{\eta}_t$ that satisfies the following correctness invariant:

\begin{quote}
\textbf{Correctness invariant:} $\hat{\eta}_t$ contains the encoding of the partial solution $W_t$ of $W$. 
\end{quote}

By Claim~\ref{clm:unique-decode}, the correctness invariant implies that we are keeping track of $W$ via encodings and their decodings. The key idea of an efficient dynamic program is to guarantee that $|\hat{\eta}_t| \leq 2^{O(\mbtw)}$ for every node $t$. We do that by applying size reduction based on the representation theorem (Theorem~\ref{thm:representation}).

\begin{quote}
\textbf{Size reduction:} We guarantee that $|\hat{\eta}_t| \leq 12^{\mbtw}$ for every node $t$ as follows. For a fixed labeling $c_t$ of $X_t$ and a fixed susbet $Y \subseteq X_t$, let $\eta_t(c_t,Y) = \{(c_t(Y'),\alpha,Y') | (c_t(Y'),\alpha,Y')\in \hat{\eta}_t, Y' = Y\}$ be the set of all encodings in $\hat{\eta}_t$ with the same set $Y$ and vertex labeling $c_t$ but different partitions of $Y$. Let $\Gamma$ be the set of partitions of $Y_t$ associated with encodings in $\eta_t(c_t,Y)$. Let $\widehat{\Gamma} \subseteq_{\min}\Gamma$. By Theorem~\ref{thm:representation}, $|\widehat{\Gamma}| \leq 2^{\mbtw-1}$. We now construct a new set of encodings $\hat{\eta}_t(c_t,Y)$ from $\eta_t(c_t,Y)$ as follows:  for each partition $\hat{\alpha} \in \widehat{\Gamma}$, we add the encoding $(c_t(Y),\hat{\alpha},Y)$ to $\hat{\eta}_t(c_t,Y)$.

Then, we set $\hat{\eta}_t \leftarrow (\hat{\eta}_t  \setminus \eta_t(c_t,Y))\cup \hat{\eta}_t(c_t,Y)$.  We repeat the  reduction for every fixed $Y$ and $c_t$. Since there are at most $2^{\mbtw}$ different subsets $Y$ and $3^{\mbtw}$ different labelings $c_t$, $\hat{\eta}_t \leq 12^{\mbtw}$.  We denote by $\sr(\hat{\eta}_t)$ the set of encodings obtained by applying size reduction to $\hat{\eta}_t$.

\end{quote}

To see the correctness invariant of $\hat{\eta}_t$ after size reduction,  consider encoding $(c_t(Y_t), \alpha_{t}, Y_t)$ of $W_t$. Before reduction, $(c_t(Y_t), \alpha_{t}, Y_t) \in \hat{\eta}_t$. Recall $\beta_t$ is the partition of $Y_t$ induced by $R_t$. By Theorem~\ref{thm:representation}, there is an encoding $(c_t(Y_t), \hat{\alpha}_{t}, Y_t) \in \hat{\eta}_t$ after reduction such that $\hat{\alpha}_{t} \sqcup \beta_t = Y_t$ and $w(\hat{\alpha}_{t}) \leq w(\alpha_t)$. Let $\hat{W}_t = \dec(c_t(Y_t), \hat{\alpha}_{t}, Y_t)$. Since labels of vertces in $Y_t$ are the same for $\hat{W}_t$ and $W_t$, every vertex of $R_t\uplus \hat{W}_t$ has even degree. Recall $R_t$ has no edge in $G_t$, thus,  $R_t\uplus \hat{W}_t$ is an Eulerian subgraph of $G$ that spans $T$. However, $w(R_t\uplus \hat{W}_t) \leq w(R_t\uplus W_t)$ since $w(\hat{W}_t) \leq w(W_t)$. By the uniqueness of $W$, $\hat{W}_t = W_t$. Hence, $\dec(c_t(Y_t), \hat{\alpha}_{t}, Y_t) \in \hat{\eta}_t$. Thus, $\hat{\eta}_t$ satisfies correctness invariant.

We denote the empty encoding $(\emptyset, \{\emptyset\}, \emptyset)$ by $\emptyset$. If $G_t$ has a feasible solution, then $\dec(\emptyset)$ is the smallest weight feasible solution, say $S_t$, in $G_t$ and the weight of the corresponding empty partition is  $w(S_t)$. Otherwise, $\dec(\emptyset) = \emptyset$ and the weight of the corresponding empty partition is $+\infty$. 

Since the root node $r$ has $X_r = \emptyset$, $G_r = G$. Thus, the feasible solution $\dec(\emptyset)$ is the optimal solution $W$.

\paragraph{Leaf node} For each leaf node $t$, $\hat{\eta}_t$ only contains the empty encoding $\emptyset$.

\paragraph{Introduce vertex node} Let $t$ be an introduce vertex node and $t'$ be a child of $t$. Let $v = X_t\setminus X_{t'}$. By the definition of introduce vertex nodes, $v$ is an isolated vertex in $G_t$. For each encoding $(c_{t'}(Y'), \alpha', Y')$ of $\hat{\eta}_{t'}$, we construct a new encoding  $(c_{t}(Y), \alpha, Y)$ where:
\begin{enumerate}[nolistsep,noitemsep]
\item[(i)] $Y = Y'\cup \{v\}$. 
\item[(ii)] $c_t(v) = 0$ and $c_t(u) = c_{t'}(u)$ for every $u \in Y'$.  
\item[(iii)] $\alpha = \alpha' \cup \{\{v\}\}$ (add $v$ as a singleton to $\alpha'$).  
\end{enumerate}
Let $\dec(c_t(Y), \alpha, Y)= \dec(c_{t'}(Y'),\alpha', Y') \cup \{v\}$. Let $\eta^{new}_t$ be the set of new encodings. Let $\eta_t = \eta^{new}_t \cup \hat{\eta}_{t'}$.  We now show the correctness invariant for $\eta_t$.

Recall $W_t$ and $W_{t'}$ are the partial solutions of $W$ in $G_t$ and $G_{t'}$, respectively. Since $V(G_t) = V(G_{t'}) \cup \{v\}$ and $E(G_t) = E(G_{t'})$, either (a) $W_t = W_{t'}$ or (b) $W_t = W_{t'} \cup \{v\}$ ($v$ is added to $W_{t'}$ as an isolated vertex).  In case (a), $(c_t(Y_t), \alpha_t,Y_t ) = (c_{t'}(Y_{t'}), \alpha_{t'},Y_{t'})$. Thus, encoding of $W_t$ is in $\eta_{t'}$. In case (b), $v$ is an isolated vertex of $W_t$, thus has $c_t(v) = 0$.  Since we add $(c_t(Y_t), \alpha_t,Y_t )$ to $\eta^{new}_t$ where $Y_t = Y_{t'}\cup \{v\}$ and $\alpha_t = \alpha_{t'} \cup \{\{v\}\}$, $\eta_t$ contains the encoding of $W_t$.

Let $\hat{\eta}_t = \sr(\eta_t)$. Since $|\eta^{new}_t| \leq |\hat{\eta}_{t'}|  \leq 12^{\mbtw}$, $|\eta_t|\leq |\eta^{new}_t| + |\hat{\eta}_{t'}|  \leq 2\cdot12^{\mbtw} = 2^{O(\mbtw)}$. Thus, by Theorem~\ref{thm:representation}, the running time of size reduction is at most $2^{O(\mbtw)}tw^{O(1)}$.

\paragraph{Introduce edge node} Let $t$ be an introduce edge node where an edge $uv$ is introduced. Let $t'$ be the only child of $t$. By the definition of introduce edge nodes, $V(H_t) = V(H_{t'})$ and $E(H_t)  = E(H{t'})\cup \{uv\}$. 

Let $g(x) = ((x+1)\mod 2) + 1$. Function $g(x)$ has following properties: $g(x+1) = 1$ when $x = 0$ or $x= 2$ and $g(x+1) = 2$ when $x = 1$.

For each encoding $(c_{t'}(Y'),\alpha', Y')$ of $\hat{\eta}_{t'}$, we construct a new encoding $(c_t(Y), \alpha, Y)$ where:
\begin{enumerate}[nolistsep,noitemsep]
\item[(i)] $Y = Y'$.
\item[(ii)] $c_t(u) = g(c_{t'}(u)+1), c_t(v) = g(c_{t'}(v)+1)$ and $c_{t}(w) = c_{t'}(w)$ for every $w \in Y'\setminus \{u,v\}$.
\item[(iii)]  $\alpha_t = \alpha_{t'} \sqcup Y'[uv]$. We the assign $w(\alpha_t) = w(\alpha_{t'}) + w(uv)$.
\end{enumerate}
Let $\dec(c_t(Y), \alpha, Y)= \dec(c_{t'}(Y'),\alpha', Y') \cup \{uv\}$.  Let $\eta^{new}_t$ be the set of new encodings. We then remove duplicates from $\eta^{new}_t$: if there are two encodings  $(c_t(Y), \alpha, Y),(c_t(Y), \beta, Y)$ in $\eta^{new}_t$ where $\alpha = \beta$ but $w(\alpha) < w(\beta)$ or $(c_t(Y), \beta, Y)$ is just another version of the same encoding $(c_t(Y), \alpha, Y),c_t(Y))$ (two versions are constructed from different encodings in $\hat{\eta}_t$.), we remove $\dec(c_t(Y), \beta, Y)$ from $\eta^{new}_{t}$. Let  $\eta_t = \eta^{new}_t \cup \eta_{t'}$.  We now show the correctness invariant for $\eta_t$.

Since $V(G_t) = V(G_{t'}) $ and $E(G_t) = E(G_{t'}) \cup \{uv\}$, either (a) $W_t = W_{t'}$ or (b) $W_t = W_{t'} \cup \{uv\}$.  In case (a), $(c_t(Y_t), \alpha_t,Y_t ) = (c_{t'}(Y_{t'}), \alpha_{t'},Y_{t'})$. Thus, the encoding of $W_t$ is in $\eta_{t'}$. In case (b), adding edge $uv$ change the label of $u$ and $v$ in $W_{t'}$ to $ g(c_{t'}(u)+1)$ and $ g(c_{t'}(v)+1)$, respectively.  If $u,v$ are in two different components of $W_{t'}$, say $C_u', C_v'$, respectively, adding $uv$ merges $C_u'$ and $C_v'$ into one connected component. Thus, $\alpha_t = \alpha_{t} \cup Y_{t'}[uv]$. That implies the encoding $(c_t(Y_t), \alpha_t,Y_t )$ of $W_t$ is in $\eta^{new}_t$. By Observation~\ref{obs:remove-encoding},   $(c_t(Y_t), \alpha_t,Y_t )$ is not removed  in $\eta^{new}_t$ during the duplicate removal; the correctness invariant of $\eta_t$ follows.

Let $\hat{\eta}_t = \sr(\eta_t)$.  Since $|\eta^{new}_t| \leq |\hat{\eta}_{t'}|  \leq 12^{\mbtw}$, $|\eta_t| \leq 2\cdot12^{\mbtw} = 2^{O(\mbtw)}$. Thus, the running time of size reduction is at most $2^{O(\mbtw)}tw^{O(1)}$.

\paragraph{Forget node} Let $t$ be a forget node and $t'$ be the only child of $t$. Let $v = X_{t'}\setminus X_t$. We first discard any encoding $(c_{t'}(Y'),\alpha',Y)$ in $\hat{\eta}_{t'}$ that satisfies one of three following conditions:
\begin{enumerate}[nolistsep,noitemsep]
\item  $c_{t'}(v) = 1$.
\item $c_{t'}(v) = 0$ and $v \in T$.
\item $c_{t'}(v) = 2$, $v$ is a singleton in the partition $\alpha'$ and $\dec(c_{t'}(Y'),\alpha', Y')$ is not a feasible solution.
\end{enumerate}

For each remaining encoding, say  $(c_{t'}(Y'),\alpha', Y')$, of $\hat{\eta}_{t'}$, we construct a new encoding $(c_t(Y), \alpha, Y)$ where:
\begin{enumerate}[nolistsep,noitemsep]
\item[(i)] $Y = Y'\setminus \{v\}$.
\item[(ii)] $c_t(u) = c_{t'}(u)$ for every $u \in Y$
\item[(iii)]  $\alpha = \alpha' \setminus \{v\}$ and $w(\alpha) = w(\alpha')$.
\end{enumerate}
Let $\dec(c_t(Y), \alpha, Y)= \dec(c_{t'}(Y'),\alpha', Y')$. Let $\eta^{new}_t$ be the set of new encodings. Let  $\eta_t = \eta^{new}_t \cup \hat{\eta}_{t'}$.  We then remove duplicates from $\eta_t$. We now show the correctness invariant for $\eta_t$. 

Observe that if $v \in W_{t'}$, it must have label $2$ in the encoding of $W_{t'}$ since $V(G_t) = V(G_{t'} \setminus \{v\})$. Furthermore, if $v$ is a singleton in $\alpha_{t'}$, $W_{t'} = W$. That implies $\dec(c_{t'}(Y_{t'}),\alpha_{t'}, Y_{t'})$ is a feasible solution. Thus, $\dec(c_{t'}(Y_{t'}),\alpha_{t'}, Y_{t'})$  is not discarded at the beginning (the new encoding constructed from $\dec(c_{t'}(Y_{t'}),\alpha_{t'}, Y_{t'})$ is empty.).

We consider two cases: (a) $W_{t'}$ does not contain $v$ and (b) $W_{t'}$ contains $v$. In case (a), $(c_t(Y_t), \alpha_t,Y_t ) = (c_{t'}(Y_{t'}), \alpha_{t'},Y_{t'})$. Thus, the encoding of $W_t$ is in $\hat{\eta}_{t'}$. In case (b), $Y_t = Y_{t'}\setminus \{v\}$, $c_t(Y_t) = c_{t'}(Y_{t'}\setminus \{v\})$ and $\alpha_t = \alpha_{t'} \setminus \{v\}$. Thus,  $(c_t(Y_t), \alpha_t,Y_t )$ is in $\eta^{new}_t$; the correctness invariant of $\eta_t$ follows.

Let $\hat{\eta}_t = \sr(\eta_t)$.  Since $|\eta^{new}_t| \leq |\hat{\eta}_{t'}|  \leq 12^{\mbtw}$, $|\eta_t| \leq 2\cdot12^{\mbtw} = 2^{O(\mbtw)}$. Thus, the running time of size reduction is at most $2^{O(\mbtw)}tw^{O(1)}$.

\paragraph{Join node} Let $t$ be a join node with two children $t_1,t_2$. Note that $X_t = X_{t_1} = X_{t_2}$. Let $h(x,y)$ be a function where:
\begin{equation*}
h(x,y) = \begin{cases} 0, & \mbox{if } x = y = 0 \\ 1, & \mbox{if } x+y\mbox{ is odd} \\ 2, & \mbox{otherwise}\end{cases}
\end{equation*}

For each encoding $(c_{t_1}(Y_1), \alpha_1), Y_1$ of $\hat{\eta}_{t_1}$ and $(c_{t_2}(Y_2), \alpha_2, Y_2)$ of $\hat{\eta}_{t_1}$ such that $Y_1 = Y_2$, we construct a new encoding $(c_t(Y), \alpha, Y)$ where:
\begin{enumerate}[nolistsep,noitemsep]
\item[(i)] $Y = Y_1 = Y_2$.
\item[(ii)] $c_t(u) = h(c_{t_1}(u), c_{t_2}(u))$ for every $u \in Y$
\item[(iii)]  $\alpha = \alpha_{1} \sqcup \alpha_{2}$ and $w(\alpha) = w(\alpha_{1}) + w(\alpha_{2})$. 
\end{enumerate}
Since $E(X_{t_1}) \cap E(X_{t_2}) = \emptyset$, $ E(\dec(c_{t_1}(Y_1), \alpha_1, Y_1))\cap E(\dec(c_{t_2}(Y_2), \alpha_2, Y_2)) = \emptyset$. Let $\dec(c_t(Y), \alpha, Y)= \dec(c_{t_1}(Y_1), \alpha_1, Y_1) \cup \dec(c_{t_2}(Y_2), \alpha_2, Y_2)$. Let $\eta_t$ be the set of new encodings.  We then remove duplicates from $\eta_t$. We now show the correctness invariant for $\eta_t$.

Recall $W_{t_1},W_{t_2}$ are the partial solutions of $W$ in $G_{t_1}$ and $G_{{t_2}}$, respectively. We consider the relationship between the encoding $(c_t(Y_t), \alpha_t, Y_t)$ of $W_t$ and the encodings of its two children $(c_{t_1}(Y_{t_1}), \alpha_{t_1}, Y_{t_1})$ and $(c_{t_2}(Y_{t_2}), \alpha_{t_2}, Y_{t_2})$. 

 Since $E(G_{t_1}) \cap E(G_{t_2}) = \emptyset$, $E(W_{t_1})\cap E(W_{t_2}) = \emptyset$. Since $X_t = X_{t_1} = X_{t_2}$, we have $Y_t = Y_{t_1} = Y_{t_2}$. Since degree in $W_t$ of a vertex $v \in Y_t$ is the sum of its degrees in $Y_{t_1}$ and $Y_{t_2}$, $c_t(v) = h(c_{t_1}(v), c_{t_2}(v))$.  Since $W_t = W_{t_1}\cup W_{t_2}$,  we have $\alpha_t = \alpha_{t_1} \sqcup \alpha_{t_2}$. That implies $(c_t(Y_t), \alpha_t, Y_t)$ is in $\eta_t$.

Let $\hat{\eta}_t = \sr(\eta_t)$. Since $|\eta_t| \leq |\hat{\eta}_{t_1}||\hat{\eta}_{t_2}| \leq 12^{2\mbtw} = 2^{O(\mbtw)}$, size reduction can be done in $2^{O(\mbtw)}tw^{O(1)}$ time.

\begin{claim} \label{clm:dp-time} The dynamic programming table of each node can be constructed in time $2^{O(\mbtw)}tw^{O(1)}n$.
\end{claim}
The $n$ factor in Claim~\ref{clm:dp-time} is for maintaining decodings in each step. This factor can be removed, but it is not the purpose of our paper. Thus, the total running time of the dynamic programming algorithm is $2^{O(\mbtw)}tw^{O(1)}n^2$.

\end{document}